\documentclass[a4paper,UKenglish]{lipics-v2018}

\usepackage{microtype}

\title{Graphical Conjunctive Queries}

\author{Filippo Bonchi}{University of Pisa}{}{}{}

\author{Jens Seeber}{IMT School for Advanced Studies Lucca}{}{}{}
\author{Pawe{\l} Soboci{\'n}ski}{University of Southampton}{}{}{}

\authorrunning{F. Bonchi, J. Seeber, P. Soboci{\'n}ski}

\Copyright{Filippo Bonchi, Jens Seeber, Pawe{\l} Soboci{\'n}ski}

\subjclass{}


\category{}

\relatedversion{}

\supplement{}

\funding{}

\EventEditors{John Q. Open and Joan R. Access}
\EventNoEds{2}
\EventLongTitle{42nd Conference on Very Important Topics (CVIT 2016)}
\EventShortTitle{CVIT 2016}
\EventAcronym{CVIT}
\EventYear{2016}
\EventDate{December 24--27, 2016}
\EventLocation{Little Whinging, United Kingdom}
\EventLogo{}
\SeriesVolume{42}
\ArticleNo{23}
\nolinenumbers 
\hideLIPIcs  

\usepackage{tikz}
\usetikzlibrary{cd}
\usepackage[all,cmtip]{xy}
\usepackage{mathtools}
\usepackage{mathdots} 

\usepackage{todonotes}

\usepackage{comment}

\usepackage{prooftree}
\usepackage{enumerate}

\theoremstyle{plain}
\newtheorem{proposition}[theorem]{Proposition}
\theoremstyle{definition}

\newcommand{\scalefact}{0.8} 
\newcommand{\exscale}{0.75} 

\newcommand{\N}{\mathbb{N}}

\newcommand{\Fin}{\mathbb{F}}
\newcommand{\from}{\colon}
\newcommand{\op}[1]{#1^{\opposite}}
\newcommand{\Cospan}[1]{\mathsf{Cospan}(#1)}
\newcommand{\Cospanleq}[1]{\mathsf{Cospan}^{\leq} #1}

\newcommand{\Span}[1]{\mathsf{Span}(#1)}
\newcommand{\Spanleq}[1]{\mathsf{Span}^{\leq} #1}
\newcommand{\Spantilde}[1]{\mathsf{Span}^{\sim} #1}
\newcommand{\DiscCospan}[1]{\mathsf{Csp}^{\leq} #1}

\DeclareMathOperator{\opposite}{op}
\DeclareMathOperator{\id}{id}

\DeclareMathOperator{\Hyp}{\mathbf{FHyp}}
\DeclareMathOperator{\iHyp}{\mathbf{Hyp}}
\DeclareMathOperator{\Rel}{\mathbf{Rel}}
\DeclareMathOperator{\Set}{\mathbf{Set}}
\DeclareMathOperator{\Hom}{Hom}

\newcommand{\gcqprop}{\mathbb{G}}
\newcommand{\gcqpobc}{\mathbb{CB}_{\Sigma}}

\newcommand{\preOrdSyntaxPROP}{\mathbb{CB}^{\leq}_{\Sigma}}
\newcommand{\preOrdSyntaxPROPEmpty}{\mathbb{CB}^{\leq}_{\emptyset}}

\newcommand{\PosetCat}[1]{#1^{\sim}} 

\newcommand{\syneq}{=_{\gcqpobc}} 
\newcommand{\synleq}{\leq_{\gcqpobc}}

\tikzset{x=1em, y=1.5ex, baseline=-0.5ex}
\tikzset{ihbase/.style={inner sep=0,circle,draw,fill=lightgray,minimum size=0.4em}}
\tikzset{ihblack/.style={ihbase,fill=black}}
\tikzset{ihwhite/.style={ihbase,fill=white}}
\tikzset{mat/.style={draw,fill=white,rectangle,node font=\scriptsize}}
\tikzset{ha/.style={mat, rectangle, rectangle }}
\tikzset{haop/.style={mat, rectangle, rectangle}}
\tikzset{anti/.style={inner sep=0,isosceles triangle,fill=black,draw=black,
    minimum width=0.75em}}
\tikzset{antiop/.style={anti,shape border rotate=180}}
\tikzset{antisq/.style={inner sep=0,rectangle,fill=black,
    minimum height=1em, minimum width=0.6em}}
\tikzset{count/.style={above,font=\scriptsize}}
\tikzset{axiom/.style={above,font=\small}}
\tikzset{dir/.style={-Latex}}

\newcommand{\blk}{node[ihblack]}

\newcommand{\Bcomult}{\tikz{
    \draw (1, 0) \blk (copy) {} .. controls (1.25, 0.75) .. (2, 0.75);
    \draw (0, 0) -- (copy) .. controls (1.25, -0.75) .. (2, -0.75);
  }}

\newcommand{\Bcounit}{\tikz{\draw (0, 0) -- (1, 0) \blk
{};}}

\newcommand{\Bmult}{\tikz{
    \draw (0,  0.75) .. controls (0.75,  0.75) .. (1, 0) \blk (copy) {};
    \draw (0, -0.75) .. controls (0.75, -0.75) .. (copy) -- (2, 0);
  }}

\newcommand{\Bunit}{\tikz{\draw (0, 0) \blk {} -- (1, 0);}}

\newcommand{\Rcircuit}{\lower7pt\hbox{$\includegraphics[height=20pt]{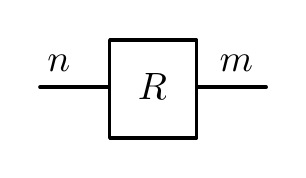}$}}

\newcommand{\BmultL}[1]{\tikz{
    \draw (0,  0.75) .. controls (0.75,  0.75) .. (1, 0) \blk (copy) {};
    \draw (0, -0.75) .. controls (0.75, -0.75) .. (copy) -- node[pos=0.5,anchor=south] {\tiny{\ensuremath{#1}}} (2, 0);
  }}

\newcommand{\BcomultL}[1]{\tikz{
    \draw (1, 0) \blk (copy) {} .. controls (1.25, 0.75) .. (2, 0.75);
    \draw node[pos=0.75,anchor=south] {\tiny{\ensuremath{#1}}} (0, 0) -- (copy) .. controls (1.25, -0.75) .. (2, -0.75);
  }}

  \newcommand{\BunitL}[1]{\tikz{\draw (0, 0) \blk {} -- node [pos=0.5,anchor=south] {\tiny{\ensuremath{#1}}} (1, 0);}}

  \newcommand{\BcounitL}[1]{\tikz{\draw (0, 0) -- node [pos=0.5,anchor=south] {\tiny{\ensuremath{#1}}} (1, 0) \blk {};}}

\newcommand{\dsymNetL}[2]{\tikz{
    \draw (0, 0.5) -- node[anchor=south] {\tiny{\ensuremath{#1}}} (0.5, 0.5)
    .. controls (0.75, 0.5) and (1.25, -0.5) .. (1.5, -0.5)
    -- node[anchor=north] {\tiny{\ensuremath{#1}}} (2, -0.5);
    \draw (0, -0.5) -- node[anchor=north] {\tiny{\ensuremath{#2}}} (0.5, -0.5)
    .. controls (0.75, -0.5) and (1.25, 0.5) .. (1.5, 0.5)
    -- node[anchor=south] {\tiny{\ensuremath{#2}}} (2, 0.5);
  }}

\newcommand{\idoneL}[1]{\tikz{
    \draw (0, 0) -- node[anchor=south] {\tiny{\ensuremath{#1}}} (1, 0);
}}

\newcommand{\BcomultX}{\BcomultL{X}}
\newcommand{\BmultX}{\BmultL{X}}
\newcommand{\BcounitX}{\BcounitL{X}}
\newcommand{\BunitX}{\BunitL{X}}

\newcommand{\idone}{
\tikz {
\draw (0, 0) -- (1, 0);
}
}

\newcommand\dsymNet{\tikz {
    \draw (0, 0.5) -- (0.5, 0.5)
    .. controls (0.75, 0.5) and (1.25, -0.5) .. (1.5, -0.5)
    -- (2, -0.5);
    \draw (0, -0.5) -- (0.5, -0.5)
    .. controls (0.75, -0.5) and (1.25, 0.5) .. (1.5, 0.5)
    -- (2, 0.5);
  }}

\newcommand{\Zeronet}{\lower4pt\hbox{$\includegraphics[height=12pt]{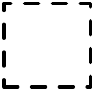}$}}

\newcommand{\tns}{\oplus}
\newcommand{\poi}{\ensuremath{\mathrel{;}}}
 
\newcommand{\bnfEq}{\; ::= \;}
\newcommand{\bnfSep}{\; |\;}

\newcommand{\sort}[2]{\ensuremath{(#1,\,#2)}}

\newcommand{\reductionRule}[2]{{\prooftree{\scriptstyle #1}\justifies{\scriptstyle #2}\endprooftree}}

\newcommand{\typ}{\mathrel{:}}
\newcommand{\typeJudgment}[3]{{ {#2} \,\typ\, {#3}}}

\newcommand{\densem}[2]{[\![ #1 ]\!]_{#2}} 

\newcommand{\df}{\coloneqq}
\newcommand{\dfop}{\eqqcolon}

\begin{document}

\maketitle

\begin{abstract}
 The Calculus of Conjunctive Queries (CCQ) has foundational status in database theory. A celebrated theorem of Chandra and Merlin states that CCQ query inclusion is decidable. Its proof transforms logical formulas to graphs: each query has a \emph{natural model}---a kind of \emph{graph}---and query inclusion reduces to the existence of a graph homomorphism between natural models. 

We introduce the diagrammatic language \emph{Graphical Conjunctive Queries} (GCQ) and show that it has the same expressivity as CCQ. GCQ terms are \emph{string diagrams}, and their algebraic structure allows us to derive a sound and complete axiomatisation of query inclusion, which turns out to be exactly Carboni and Walters' notion of \emph{cartesian bicategory of relations}. Our completeness proof exploits the combinatorial nature of string diagrams as (certain cospans of) hypergraphs: Chandra and Merlin's insights inspire a theorem that relates such cospans with spans. Completeness and decidability of the (in)equational theory of GCQ follow as a corollary. Categorically speaking, our contribution is a model-theoretic completeness theorem of free cartesian bicategories (on a relational signature) for the category of sets and relations.
\end{abstract}

\section{Introduction}

Conjunctive queries (CCQ) are first-order logic formulas that use only relation symbols, equality,
truth, conjunction, and existential quantification. They are a kernel language of queries to relational databases and are the foundations of several languages:
they are select-project-join queries in
relational algebra \cite{codd1970relational}, or select-from-where queries in SQL \cite{chamberlin1974sequel}.
While expressive enough to encompass queries of practical interest,
they admit algorithmic analysis: in~\cite{chandra1977optimal}, Chandra and Merlin showed that the problem of \emph{query inclusion} is NP-complete.

For an example of query inclusion in action, consider formulas
\[
  \phi = \exists z_0 \colon (x_0 = x_1) \wedge R(x_0,z_0) \text{ and }
  \psi = \exists z_0, z_1 \colon R(x_0,z_0) \wedge R(x_1,z_0) \wedge R(x_0,z_1) \wedge R(x_1,z_1),
\]
  with free variables $x_0,x_1$. Irrespective of model, and thus the interpretation of the relation symbol $R$, every free variable assignment satisfying $\phi$ satisfies $\psi$: i.e.\ $\phi$ is included in $\psi$.

  Chandra and Merlin's insight involves an elegant reduction to graph theory, namely the existence of a hypergraph homomorphism from a graphical encoding of $\psi$ to that of $\phi$. Below on the left we give a graphical rendering of $\psi$ and $\phi$, respectively: vertices represent variables, while edges are labelled with relation symbols. The dotted connections are not, strictly speaking, a part of the underlying hypergraphs. They constitute an \emph{interface}: a mapping
 \begin{tabular}{cc|c}
\raise0pt\hbox{$\includegraphics[height=38pt]{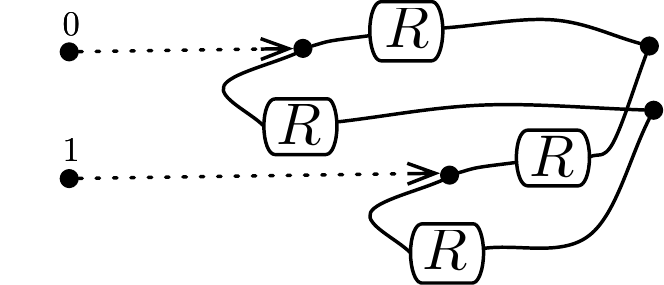}$} &
\raise13pt\hbox{$\includegraphics[height=20pt]{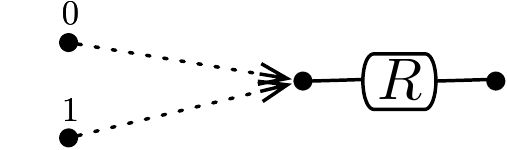}$} &
$\includegraphics[height=47pt]{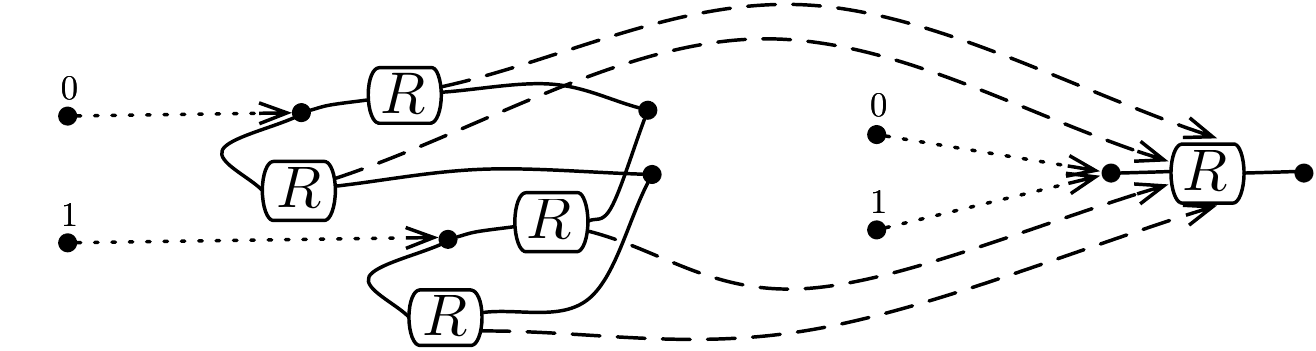}$
\end{tabular}\\
from the free variables $\{x_0,\,x_1\}$ to the vertices. The aforementioned query inclusion is witnessed by an interface-preserving hypergraph homomorphism, displayed above on the right. In category-theoretic terms, hypergraphs-with-interfaces are \emph{discrete cospans}, and the homomorphisms are \emph{cospan homomorphisms}.

In previous work~\cite{bonchi2016rewriting}, the first and third authors with Gadducci, Kissinger and Zanasi showed that such cospans characterise an
important family of \emph{string diagrams}---i.e.\ diagrammatic representations of the arrows of monoidal categories---namely those equipped with
an algebraic structure known as a special Frobenius algebra.
This motivated us to study the connection between this fashionable algebraic structure---which has been used in fields as diverse as quantum computing~\cite{Abramsky2008:CQM,PQP,OxfordCompleteness,LoriaCompleteness}, concurrency theory~\cite{Bruni2006,Bruni2011,brunihierarchical,Bruni2013}, control theory~\cite{Bonchi2015,BaezErbele-CategoriesInControl} and linguistics~\cite{SadrzadehCC14}---and conjunctive queries.

We introduce the logic of Graphical Conjunctive Queries (GCQ).
Although superficially unlike CCQ, we show that it is equally expressive. Its syntax lends itself to string-diagrammatic representation and diagrammatic reasoning respects the underlying logical semantics. GCQ string diagrams for $\psi$ and $\phi$ are drawn below. Note that, while GCQ syntax does not have variables, the concept of CCQ free variable is mirrored by ``dangling'' wires in diagrams.
\[
  \raisebox{-6mm}{\scalebox{\exscale}{
\begin{tikzpicture}
\draw (0.75,2.0) .. controls (0.75,2.0) and (0.5,2.0) .. (0.0,2.0);
\draw (0.75,2.0) .. controls (0.75,1.0) and (1.0,1.0) .. (1.5,1.0);
\draw (0.75,2.0) .. controls (0.75,3.0) and (1.0,3.0) .. (1.5,3.0);
\draw (0.75,2.0) \blk {};
\draw (0.75,6.0) .. controls (0.75,6.0) and (0.5,6.0) .. (0.0,6.0);
\draw (0.75,6.0) .. controls (0.75,5.0) and (1.0,5.0) .. (1.5,5.0);
\draw (0.75,6.0) .. controls (0.75,7.0) and (1.0,7.0) .. (1.5,7.0);
\draw (0.75,6.0) \blk {};
\draw (2.25,1.0) .. controls (2.25,1.0) and (2.0,1.0) .. (1.5,1.0);
\draw (2.25,1.0) .. controls (2.25,1.0) and (2.5,1.0) .. (3.0,1.0);
\node[mat] at (2.25,1.0) {$R$};
\draw (2.25,3.0) .. controls (2.25,3.0) and (2.0,3.0) .. (1.5,3.0);
\draw (2.25,3.0) .. controls (2.25,3.0) and (2.5,3.0) .. (3.0,3.0);
\node[mat] at (2.25,3.0) {$R$};
\draw (2.25,5.0) .. controls (2.25,5.0) and (2.0,5.0) .. (1.5,5.0);
\draw (2.25,5.0) .. controls (2.25,5.0) and (2.5,5.0) .. (3.0,5.0);
\node[mat] at (2.25,5.0) {$R$};
\draw (2.25,7.0) .. controls (2.25,7.0) and (2.0,7.0) .. (1.5,7.0);
\draw (2.25,7.0) .. controls (2.25,7.0) and (2.5,7.0) .. (3.0,7.0);
\node[mat] at (2.25,7.0) {$R$};
\draw (3.0,1.0) .. controls (3.5,1.0) and (4.0,1.0) .. (4.5,1.0);
\draw (3.0,3.0) .. controls (3.5,3.0) and (4.0,5.0) .. (4.5,5.0);
\draw (3.0,5.0) .. controls (3.5,5.0) and (4.0,3.0) .. (4.5,3.0);
\draw (3.0,7.0) .. controls (3.5,7.0) and (4.0,7.0) .. (4.5,7.0);
\draw (5.25,2.0) .. controls (5.25,1.0) and (5.0,1.0) .. (4.5,1.0);
\draw (5.25,2.0) .. controls (5.25,3.0) and (5.0,3.0) .. (4.5,3.0);
\draw (5.25,2.0) .. controls (5.25,2.0) and (5.5,2.0) .. (6.0,2.0);
\draw (5.25,2.0) \blk {};
\draw (6.25,2.0) .. controls (6.25,2.0) and (6.5,2.0) .. (6.0,2.0);
\draw (6.25,2.0) \blk {};
\draw (5.25,6.0) .. controls (5.25,5.0) and (5.0,5.0) .. (4.5,5.0);
\draw (5.25,6.0) .. controls (5.25,7.0) and (5.0,7.0) .. (4.5,7.0);
\draw (5.25,6.0) .. controls (5.25,6.0) and (5.5,6.0) .. (6.0,6.0);
\draw (5.25,6.0) \blk {};
\draw (6.25,6.0) .. controls (6.25,6.0) and (6.5,6.0) .. (6.0,6.0);
\draw (6.25,6.0) \blk {};
\end{tikzpicture}}}
\qquad
\raisebox{-2.5mm}{\scalebox{\exscale}{
\begin{tikzpicture}
\draw (0.75,2.0) .. controls (0.75,1.0) and (0.5,1.0) .. (0.0,1.0);
\draw (0.75,2.0) .. controls (0.75,3.0) and (0.5,3.0) .. (0.0,3.0);
\draw (0.75,2.0) .. controls (0.75,2.0) and (1.0,2.0) .. (1.5,2.0);
\draw (0.75,2.0) \blk {};
\draw (2.25,2.0) .. controls (2.25,2.0) and (2.0,2.0) .. (1.5,2.0);
\draw (2.25,2.0) .. controls (2.25,2.0) and (2.5,2.0) .. (3.0,2.0);
\node[mat] at (2.25,2.0) {$R$};
\draw (3.25,2.0) .. controls (3.25,2.0) and (3.5,2.0) .. (3.0,2.0);
\draw (3.25,2.0) \blk {};
\end{tikzpicture}}}
\]

While interesting in its own right as an example of a string-diagrammatic representation of a logical language---which has itself  become a topic of recent interest~\cite{ghica}---GCQ comes into its own when reasoning about query inclusion, which is characterised by the laws of \emph{cartesian bicategories}. 
This important categorical structure was introduced by Carboni and Walters~\cite{carboni1987cartesian} who were, in fact, aware of the logical interpretation, mentioning it in passing without giving the details. Our definition of GCQ, its expressivity, and \emph{soundness} of the laws of cartesian bicategories w.r.t. query inclusion is testament to the depth of their insights.

\smallskip

The main contribution of our work is the \emph{completeness} of the laws of cartesian bicategories for query inclusion (Theorem~\ref{thm:main}). 

\smallskip

As a side result, we obtain a categorical understanding of the proof by Chandra and Merlin. This uncovers a beautiful triangle relating logical, combinatorial and categorical structures, similar to the Curry-Howard-Lambek correspondence relating intuitionistic propositional logic, $\lambda$-calculus and free cartesian closed categories.

\[
\xymatrix@C=50pt{
{\relax\txt<1.5cm>{{\bf Logical} {CCQ=GCQ}}}
\ar@{<->}@/^/[rr]^-{\text{Chandra and Merlin~\cite{chandra1977optimal}}} 
\ar@{<->}@/_/[dr]_{\txt{Theorem~\ref{thm:main}}\qquad}
& & {\relax\txt<2.5cm>{{\bf Combinatorial } hypergraphs with interfaces}} 
\ar@{<->}@/^/[dl]^{\qquad\txt{ Theorem~\ref{thm:hypergraph}}}
\\
& {\relax\txt<6cm>{{\bf \phantom{\qquad}~Categorical~\phantom{\qquad}} {free cartesian bicategories}}}
}
\]

The rightmost side of the triangle (Theorem~\ref{thm:hypergraph}) provides a  combinatorial characterisation of free cartesian bicategories as discrete cospans of hypergraphs, with the Chandra and Merlin ordering: the existence of a cospan homomorphism in the opposite direction. This result can also be regarded as an extension of the aforementioned~\cite{bonchi2016rewriting} to an enriched setting. 
The fact that the Chandra and Merlin ordering is not antisymmetric forces us to consider preorder-enrichment as opposed to the usual~\cite{carboni1987cartesian} poset-enrichment of cartesian bicategories.\footnote{
While cartesian bicategories were later generalised~\cite{carboni2008cartesian} to a bona fide higher-dimensional setting, our preorder-enriched variant seems to be an interesting stop along the way.}

The step from posets to preorders is actually beneficial: it provides a one-to-one correspondence between hypergraphs and models which we see as functors, following the tradition of categorical logic. The model corresponding to a hypergraph $G$ is exactly the (contravariant) $\Hom$-functor represented by $G$.
By a Yoneda-like argument, we obtain  a ``preorder-enriched analogue'' of Theorem~\ref{thm:main} (Theorem~\ref{thm:completeness}).
With this result, proving Theorem~\ref{thm:main}  reduces to descending from the preorder-enriched setting down to poset-enrichment. 

\medskip

Working with both poset- and preorder-enriched categories means that there is a relatively large number of categories at play. We give a summary of the most important ones in the table below, together with pointers to their definitions. The remainder of this introduction is a roadmap for the paper, focussing on the roles played by the categories mentioned below.
\begin{figure}[h]
  \centering
  \begin{tabular}{c|c|c}
    & preordered & posetal \\ \hline
    free categories & $\preOrdSyntaxPROP$ (Def~\ref{defn:fterm}) & $\gcqpobc$ (Def~\ref{def:gcqpobc}) \\ \hline
    semantic domains for the logic & $\Spanleq{\Set}$ (Def~\ref{defn:span}) & $\Rel$ (Ex~\ref{ex:Rel}) \\ \hline
    combinatorial structures & $\DiscCospan{\Hyp_{\Sigma}}$ (Def~\ref{defn:disccospan}) & - \\
  \end{tabular}
\end{figure}

We begin by justifying the ``equation'' CCQ=GCQ in the triangle above: we recall CCQ and introduce GCQ in Sections~\ref{sec:CCQ} and~\ref{sec:GCQ}, respectively, and show that they have the same expressivity.
We explore the algebraic structure of GCQ in Sections~\ref{sec:todiagram} and~\ref{sec:axioms}, which---as we previously mentioned---is exactly that of cartesian bicategories. As instances of these, we introduce $\gcqpobc$, the free cartesian bicategory,  and $\Rel$, the category of sets and relations.

In Section~\ref{sec:cospans} we introduce \emph{preordered} cartesian bicategories (the free one denoted by $\preOrdSyntaxPROP$) and the category of discrete cospans of hypergraphs with the Chandra and Merlin preorder, denoted by $\DiscCospan{\Hyp_{\Sigma}}$. Theorem~\ref{thm:hypergraph} states that these two are isomorphic.

Theorem~\ref{thm:completeness} is proved in Section~\ref{sec:CompSpan}. Rather than
 $\Rel$, the preordered setting calls for models in $\Spanleq{\Set}$, the preordered cartesian bicategory of spans of sets. In Section~\ref{sec:rel}, we explain the passage from preorders to posets, completing the proof of Theorem~\ref{thm:main}.

We delay a discussion of the ramification of our work, a necessarily short and cursory account---due to space restrictions---of the considerable related work, and directions for future work to Section~\ref{sec:discussion}. We conclude with the observation that \emph{(i)} the diagrammatic language for formulas, \emph{(ii)} the semantics, e.g. of composition of diagrams---what we understand in modern terms as the combination of conjunction and existential quantification---and \emph{(iii)} the use of diagrammatic reasoning as a powerful method of logical reasoning actually go back to the pre-Frege work of the 19th century American polymath CS Peirce on \emph{existential graphs}. Interestingly, it is only recently (see, e.g.~\cite{mellieszeilberger}) that this work has been receiving the attention that it richly deserves.
\begin{figure}\label{fig:namestab}
\begin{tabular}{c|c|c}
\end{tabular}
\end{figure}

\noindent {\bf Preliminaries.}
We assume familiarity with basic categorical concepts, in particular symmetric monoidal,
ordered and
preordered categories. We do not assume familiarity with cartesian bicategories: the acquainted reader should note that what we call ``cartesian bicategories'' are ``cartesian bicategories of relations'' in~\cite{carboni1987cartesian}.
A \emph{prop} is a symmetric strict monoidal category where objects are natural numbers, and the monoidal product on objects is addition $m\tns n := m+n$.
Due to space restrictions, most proofs are in the Appendix.

\section{Calculus of Conjunctive Queries}
\label{sec:CCQ}

Assume a set $\Sigma$ of relation symbols with arity function $ar:\Sigma\to\N$ and a countable set 
$Var = \{ x_i \,|\, i\in \N\}$ of variables. The grammar for the calculus of conjunctive queries is:
\begin{equation}\label{eq:calculus}\tag{CCQ}
\Phi ::= \top \;|\; \Phi \wedge \Phi \;|\; x_i = x_j \;|\; R(\overset{\rightarrow}{x})  \;|\; \exists x. \Phi
\end{equation}
where $R\in\Sigma$, $ar(R)=n$, and $\overset{\rightarrow}{x}$ is a list of length $n$ of variables from $Var$. We assume the standard bound variable conventions and some basic metatheory of formulas; in particular we write $\phi[\overset{\rightarrow}{x}/\overset{\rightarrow}{y}]$, where $\overset{\rightarrow}{x},\overset{\rightarrow}{y}$ are variable lists of equal length, for the simultaneous substitution of variables from $\overset{\rightarrow}{x}$ for variables in $\overset{\rightarrow}{y}$. We write $\overset{\rightarrow}{x}_{[m,n]}$, where $m\leq n$, for the list of variables $x_m,x_{m+1},\dots,x_n$. 
Given a formula $\phi$, $fv(\phi)$ is the set of its free variables. 

\smallskip
The semantics of~\eqref{eq:calculus} formulas is standard and inherited from first order logic. 
\begin{definition}\label{defn:CCQmodel}
A model $\mathcal{M}=(X,\rho)$ is a set $X$ and, for each $R\in\Sigma$, a set $\rho(R)\subseteq X^{ar(R)}$.
\end{definition}
Given a model $\mathcal{M}=(X,\rho)$, the semantics $\densem{\phi}{\mathcal{M}}$ is the set of all assignments of elements from $X$ to $fv(\phi)$ that makes it evaluate to truth, given the usual propositional interpretation.

\smallskip
In order to facilitate a principled definition of the semantics (Definition~\ref{defn:semanticsCCQ}) and to serve the needs of our diagrammatic approach, we will need to take a closer look at free variables. To this end, we give an alternative, sorted presentation of~\eqref{eq:calculus} that features explicit free variable management. As we shall see, the system of judgments below will allow us to derive
$n \vdash \phi$ where $n\in \N$, whenever $\phi$ is a formula of CCQ and $fv(\phi)\subseteq \{x_0,\dots,x_{n-1}\}$. 
\[
\begin{prooftree}
{}
\justifies
0 \vdash \top
\using {(\top)}
\end{prooftree}
\quad  
\begin{prooftree}
{R \in \Sigma \quad ar(R)=n }
\justifies
n \vdash R(x_0,\dots, x_{n-1})
\using {(\Sigma)}
\end{prooftree}
\quad
\begin{prooftree}
n \vdash \phi
\justifies
n-1 \vdash \exists x_{n-1}.\phi
\using
{(\exists)}
\end{prooftree} \quad
\]

\[
\begin{prooftree}
{}
\justifies
2 \vdash x_0 = x_1
\using (=)
\end{prooftree} \quad
\begin{prooftree}
m \vdash \phi \quad 
n \vdash \psi 
\justifies
m+n \vdash \phi \wedge (\psi[\overset{\rightarrow}{x}_{[m,m+n-1]}/\overset{\rightarrow}{x}_{[0,x_{n-1}]}])
\using (\wedge)
\end{prooftree}
\]
Note that the above are restrictive: e.g.\ $(\wedge)$ enforces disjoint sets of variables, and $(\exists)$ allows  quantification only over the last variable. To overcome these limitations we include three structural rules that allow us to manipulate (swap, identify, and introduce) free variables.
\[
\begin{prooftree}
n \vdash \phi \quad (0\leq k < n-1)
\justifies
n \vdash \phi[x_{k+1},x_{k}/x_k,x_{k+1}]
\using {(\mathsf{Sw}_{n,k})}
\end{prooftree}
\quad
\begin{prooftree}
n \vdash \phi 
\justifies
n-1 \vdash \phi[x_{n-2}/x_{n-1}]  
\using{ (\mathsf{Id}_{n}) }
\end{prooftree}
\quad 
\begin{prooftree}
n\vdash \phi \justifies
n+1 \vdash \phi \using{ (\mathsf{Nu}_n)}
\end{prooftree}
\]
Rule $\mathsf{Sw}$ allows us to swap two free variables.
Alone, $\mathsf{Id}$ identifies the final and the penultimate free variable; used together
with $\mathsf{Sw}$ it allows for the identification of any two. Finally, 
$\mathsf{Nu}$ introduces a free variable.
The eight suffice for any CCQ formula, in the following sense:
\begin{proposition}\label{pro:syntax}
$\phi$ is a formula derived from~\eqref{eq:calculus} with $fv(\phi) \subseteq \{x_0,\dots,x_{n-1}\}$ iff $n\vdash \phi$. 
\end{proposition}
We use the sorted presentation to define the semantics.
\begin{figure*}
  \footnotesize{
    \begin{minipage}{0.22\linewidth}
    \[
      \densem{{\scriptstyle 0\vdash\top}}{\mathcal{M}} = \{\bullet\} \tag{$\top$}
  \]
    \end{minipage}
    \begin{minipage}{0.72\linewidth}
    \[
    (\overset{\rightarrow}{u},v,w,\overset{\rightarrow}{x}) \in  \densem{{\scriptstyle n \vdash \phi[x_{k+1},x_{k}/x_k,x_{k+1}] }}{\mathcal{M}}
    \Leftrightarrow
    (\overset{\rightarrow}{u},w,v,\overset{\rightarrow}{x}) \in \densem{{\scriptstyle n \vdash \phi}}{\mathcal{M}} \tag{$\mathsf{Sw}_{n,k}$}
    \]
    \end{minipage}\hfill
    \begin{minipage}{0.35\linewidth}
    \[
    \densem{{\scriptstyle n\vdash R(x_0,\dots,x_{n-1})}}{\mathcal{M}} = \rho(R) \tag{$\Sigma$}
\]
    \end{minipage}
    \begin{minipage}{0.62\linewidth}
    \[
    (\overset{\rightarrow}{v},w)\in \densem{{\scriptstyle n-1 \vdash \phi[x_{n-2}/x_{n-1}]}}{\mathcal{M}} 
    \Leftrightarrow
    (\overset{\rightarrow}{v},w,w)\in \densem{{\scriptstyle n \vdash \phi}}{\mathcal{M}} \tag{$\mathsf{Id}_n$}
\]
    \end{minipage}\hfill
    \begin{minipage}{0.37\linewidth}
    \[
      \densem{{\scriptstyle 2\vdash x_0 = x_1}}{\mathcal{M}} = \{(v,v)\,|\, v\in X\} \tag{$=$}
    \]
    \end{minipage}
    \begin{minipage}{0.55\linewidth}
    \[
    \overset{\rightarrow}{v}\in\densem{{\scriptstyle n-1\vdash \exists x_{n-1}.\phi}}{\mathcal{M}} 
    \Leftrightarrow \exists w\in X.\,(\overset{\rightarrow}{v},w)\in \densem{{\scriptstyle n\vdash \phi}}{\mathcal{M}} \tag{$\exists$}
    \]
    \end{minipage}\hfill
    \begin{minipage}{0.36\linewidth}
    \[
    \densem{{\scriptstyle n+1 \vdash \phi}}{\mathcal{M}} = 
    \densem{{\scriptstyle n\vdash\phi}}{\mathcal{M}}\times X \tag{$\mathsf{Nu}_n$}
    \]
    \end{minipage}\hfill
    \begin{minipage}{0.47\linewidth}
    \[
    \densem{{\scriptstyle {m+n \vdash \phi \wedge (\psi[\dots])}}}{\mathcal{M}} = 
    \densem{{\scriptstyle m \vdash \phi}}{\mathcal{M}} \times \densem{{\scriptstyle n\vdash \psi}}{\mathcal{M}} \tag{$\wedge$}
    \]
    \end{minipage}\hfill\\
  }
  \caption{Semantics of CCQ for a model $\mathcal{M}=(X,\rho)$.  We write $\bullet$ for the unique element of $X^0$.}
\label{fig:semanticsCCQ}
\end{figure*}
\begin{definition}\label{defn:semanticsCCQ}
Given a model $\mathcal{M}=(X,\rho)$, the semantics of $n \vdash \phi$ is a set of tuples $\densem{n\vdash \phi}{\mathcal{M}}\subseteq X^n$.
We define it in Figure~\ref{fig:semanticsCCQ} by recursion on the derivation of $n\vdash \phi$.
\end{definition}

Finally, we define the concepts that are of central interest: \emph{query equivalence} and \emph{inclusion}. 
\begin{definition}\label{defn:equivalence}
Given $n \vdash \phi$ and $n\vdash \psi$, we say that $\phi$ and $\psi$ are equivalent and write $\phi \equiv \psi$ if for all models $\mathcal{M}$ we have $\densem{n\vdash \phi}{\mathcal{M}}=\densem{n\vdash \psi}{\mathcal{M}}$. We write $\phi \leqq \psi$
when, for all $\mathcal{M}$, $\densem{n\vdash \phi}{\mathcal{M}}\subseteq\densem{n\vdash \psi}{\mathcal{M}}$. Clearly $\phi \leqq \psi$ and $\psi \leqq \phi$ implies $\phi \equiv \psi$.
\end{definition}

\section{Graphical conjunctive queries}\label{sec:GCQ}
\begin{figure}[ht]
  \resizebox{0.75\hsize}{!}{
\[
\reductionRule{}{ \typeJudgment{}{\Bcomult}{\sort{1}{2}} }\quad
\reductionRule{}{ \typeJudgment{}{\Bcounit}{\sort{1}{0}} }\quad
\reductionRule{R\in \Sigma_{n,m}}{ \typeJudgment{}{R}{\sort{n}{m}} }\quad
\reductionRule{}{ \typeJudgment{}{\Bmult}{\sort{2}{1}} }\quad
\reductionRule{}{ \typeJudgment{}{\Bunit}{\sort{0}{1}} } \quad
\reductionRule{}{ \typeJudgment{}{\Zeronet}{\sort{0}{0}} }\quad
\reductionRule{}{ \typeJudgment{}{\idone}{\sort{1}{1}} }\quad
\reductionRule{}{ \typeJudgment{}{\dsymNet}{\sort{2}{2}} }\quad
\reductionRule{ \typeJudgment{}{c}{\sort{n}{z}} \quad \typeJudgment{}{d}{\sort{z}{m}} }
{ \typeJudgment{}{c \poi d}{\sort{n}{m}} }\quad
\reductionRule{ \typeJudgment{}{c}{\sort{n}{m}} \quad \typeJudgment{}{d}{\sort{p}{q}} }
{ \typeJudgment{}{c \tns d}{\sort{n+p}{m+q}} }
\]}
\caption{Sort inference rules.\label{fig:sortInferenceRules}}
\end{figure}

We  introduce an alternative logic, called Graphical Conjunctive Queries (GCQ).
GCQ and CCQ are---superficially---quite different. Nevertheless, in Propositions~\ref{prop:CCQGCQ} and~\ref{prop:GCQCCQ} we show that they have the same expressive power. 
The grammar of GCQ formulas is given below. 
\begin{equation}\label{eq:GCQsyntax}\tag{GCQ}
c \bnfEq   \Bcounit \bnfSep \Bcomult \bnfSep \Bmult \bnfSep \Bunit \bnfSep
 \Zeronet \bnfSep \idone \bnfSep \dsymNet  \bnfSep c\tns c \bnfSep c \poi c \bnfSep R
\end{equation}
GCQ syntax is a radical departure from~\eqref{eq:calculus}. Rather than use CCQ's existential quantification and conjunction, GCQ uses the operations of monoidal categories: composition and monoidal product. 
There are no variables, thus no assumptions of their countable supply, nor any associated metatheory of capture-avoiding substitution.

The price is
a simple sorting discipline. A \emph{sort} is a pair $\sort{n}{m}$, with $n,m\in \N$. We consider only terms sortable according to Figure~\ref{fig:sortInferenceRules}. 
There and in~\eqref{eq:GCQsyntax}, $R$ ranges over the symbols of a \emph{monoidal} signature $\Sigma$, a set of relation symbols equipped with both an arity and a \emph{coarity}: $\Sigma_{n,m}$ consists of the symbols in $\Sigma$ with arity $n$ and coarity $m$. A GCQ signature plays a similar role to relation symbols in CCQ: we abuse notation for this reason. 
A simple induction shows sort uniqueness: if $\typeJudgment{}{c}{\sort{n}{m}}$ and $\typeJudgment{}{c}{\sort{n'}{m'}}$ then $n=n'$ and $m=m'$.

In~\eqref{eq:GCQsyntax} we used a graphical rendering of GCQ constants. Indeed, we will not write terms of GCQ as formulas, but instead represent them as 2-dimensional diagrams. The justification for this is twofold: the diagrammatic conventions introduced in this section mean that a diagram is a readable, faithful and unambiguous representation of a sorted~\eqref{eq:GCQsyntax} term. More importantly, our characterisation of query inclusion in subsequent sections consists of intuitive topological deformations of the diagrammatic representations of formulas.

A GCQ term $c: \sort{n}{m}$ is drawn as a diagram with $n$ ``dangling wires'' on the left, and $m$ on the right. Roughly speaking, dangling wires are GCQ's answer to the free variables of CCQ. 
Composing $(\poi)$ means connecting diagrams in series and tensoring means stacking. The shorthand {\lower3pt\hbox{$\includegraphics[height=15pt]{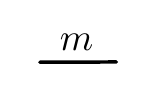}$}} stands for $m$ wires in parallel. The box $\Rcircuit$ stands for a relation symbol $R\in \Sigma_{n,m}$.  Thus, given $c: \sort{n}{m}$, $c': \sort{m}{k}$, $c \poi c': \sort{n}{k}$ is drawn \lower10pt\hbox{$\includegraphics[height=25pt]{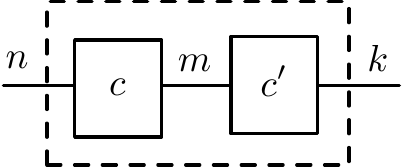}$}, and given $d: \sort{p}{q}$, $c\tns d: \sort{n+p}{m+q}$ is drawn \lower20pt\hbox{$\includegraphics[height=40pt]{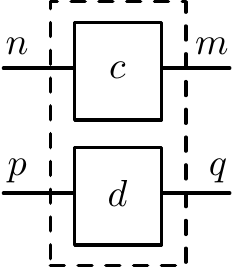}$}.
\begin{minipage}{0.85\textwidth}
\begin{example}
Consider 
$((\idone \tns \Bcomult) \tns \Zeronet) ; (R \tns S): \sort{2}{1}$, assuming $R\in \Sigma_{2,0}$, $S\in\Sigma_{1,1}$. Its diagrammatic rendering is on the right.
Note that the use of the dotted boxes induces a tree-like quality to diagrams. Indeed, they are a faithful representation for syntactic terms constructed from~\eqref{eq:GCQsyntax}.
\end{example}
\end{minipage}
\begin{minipage}{0.15\textwidth}
  \[
    \raisebox{4mm}{\includegraphics[height=55pt]{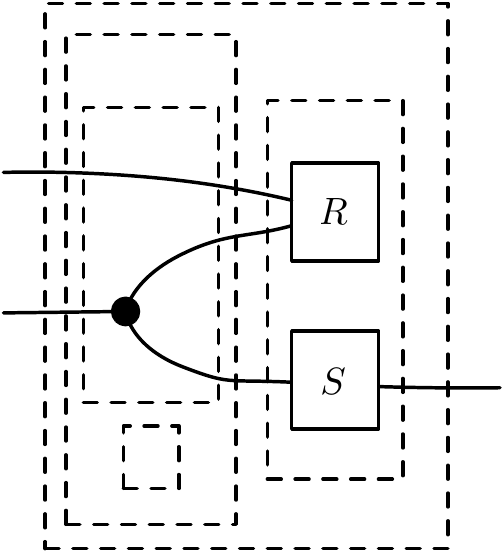}}
  \]
\end{minipage}
We now turn to semantics. First, the notion of model of GCQ is similar to a model of CCQ.
\begin{definition}\label{defn:GCQmodel}
  A model $\mathcal{M}=(X,\rho)$ is a set $X$ and, for each $R\in\Sigma_{n,m}$, $\rho(R)\subseteq X^n\times X^m$.
\end{definition}

Given a model $\mathcal{M}=(X,\rho)$, the semantics of $\typeJudgment{}{c}{\sort{n}{m}}$ is the relation $\densem {c}{\mathcal{M}} \subseteq X^n\times X^m$  defined recursively in Figure~\ref{fig:gcqsemantics}.
Armed with a notion of semantics, we can define query equivalence ($\equiv$) and inclusion ($\leqq$) for GCQ terms analogously to Definition~\ref{defn:equivalence}.

\begin{figure}[t]
  \resizebox{\linewidth}{!}{
    \begin{minipage}{\linewidth}
\begin{align*}
\densem{\Bcomult}{\mathcal{M}} &= \left\{\left(x, 
{\tiny\begin{pmatrix}
    x \\
    x
\end{pmatrix}}
\right)\; |\; x\in X\right\} &
\densem{\Bcounit}{\mathcal{M}} &= \{(x, \bullet)\; |\; x\in X\} &
\densem{ c\tns d }{\mathcal{M}} &= \densem{ c }{\mathcal{M}} \tns \densem{ d }{\mathcal{M}} &
\densem{\Zeronet}{\mathcal{M}} &= \{(\bullet,\bullet)\} \\
\densem{\Bmult}{\mathcal{M}} &= \left\{\left(
{\tiny\begin{pmatrix}
    x \\
    x
\end{pmatrix}}
, x\right)\; |\; x\in X\right\} &
\densem{\Bunit}{\mathcal{M}} &= \{(\bullet,x)\; |\; x\in X\} &
\densem{ c\poi d }{\mathcal{M}} &= \densem{ c }{\mathcal{M}} \poi \densem{ d }{\mathcal{M}} &
\densem{ R }{\mathcal{M}} &= \rho(R) \\
\densem{ \dsymNet}{\mathcal{M}} &= \left\{\left(
{\tiny \begin{pmatrix}
    x \\
    y
\end{pmatrix}},
{\tiny \begin{pmatrix}
    y \\
    x
\end{pmatrix}}\right) \; |\; x,y\in X\right\} &
\densem{\idone}{\mathcal{M}} &= \{(x,x) \; |\; x\in X\}
\end{align*}
\end{minipage}
}

\caption[Semantics of GCQ]{Semantics of GCQ for a model $\mathcal{M}=(X,\rho)$.
We used the notation $R\poi S = \{(x,z) \;|\; \exists y\in Y \text{ s.t. }(x,y)\in R \text{ and } (y,z)\in S \}$ and $R\tns S = \{\left({\tiny\begin{pmatrix}
    x \\
    u
\end{pmatrix}},{\tiny\begin{pmatrix}
    y \\
    v
\end{pmatrix}}\right) \;|\; (x,y)\in R \text{ and } (u,v)\in S \}$.
$\bullet$ is the unique element of $X^0$ and \scalebox{0.5}{$\begin{pmatrix}
    x_0 \\
    \scriptsize{\vdots} \\
    x_{n-1}
\end{pmatrix}$} an element of $X^n$.}
\label{fig:gcqsemantics}
\end{figure}

\begin{example}\label{ex:boneempty}
Consider the GCQ term $\lower5pt\hbox{$\includegraphics[height=15pt]{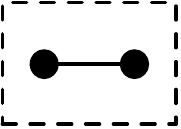}$}$ of sort $\sort{0}{0}$. For a model $\mathcal{M}=(X,\rho)$, its semantics $\densem{\lower5pt\hbox{$\includegraphics[height=15pt]{graffles/boxedbone.pdf}$}}{\mathcal{M}} \subseteq X^0\times X^0$ is either the empty relation $\emptyset$, if $X$ is empty, or the relation $\{(\bullet, \bullet)\}$, if $X$ is not empty. Since $\emptyset \subseteq \{(\bullet, \bullet)\}$, and since $\densem{\Zeronet}{\mathcal{M}} = \{(\bullet, \bullet)\} $ for all models $\mathcal{M}$, it holds that $\lower5pt\hbox{$\includegraphics[height=15pt]{graffles/boxedbone.pdf}$} \leqq \Zeronet$. Intuitively, the first term corresponds to the CCQ formula $\exists x. \top$, holding in all non empty models, while the second corresponds to the formula $\top$. In the remainder of this section we will make this intuition precise.
\end{example}

\subsection{Expressivity} \label{sec:expressivity}

\begin{figure*}
  \begin{minipage}{0.3\textwidth}
    \[
      \Theta \left( 0 \vdash \top \right) = \Zeronet \tag{$\top$}
    \]
  \end{minipage}\hfill
  \begin{minipage}{0.7\textwidth}
    \[
      \Theta \left( n \vdash \phi[x_{k+1},x_{k}/x_k,x_{k+1}] \right) =
      \lower20pt\hbox{$\includegraphics[height=40pt]{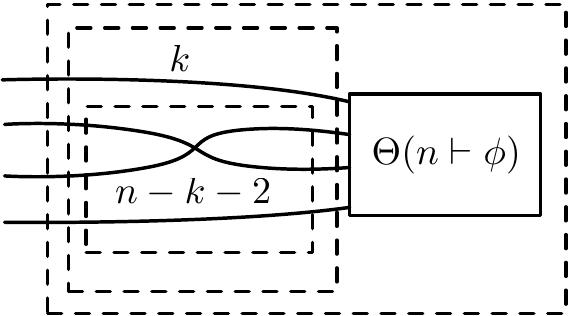}$} \tag{$\mathsf{Sw}_{n,k}$}
    \]
  \end{minipage}\hfill

  \begin{minipage}{0.35\textwidth}
    \[
      \Theta \left( 2 \vdash x_0 = x_1 \right) =
      \lower6pt\hbox{$\includegraphics[height=20pt]{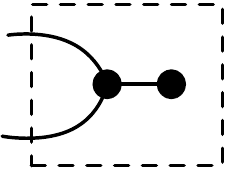}$} \tag{$=$}
    \]
  \end{minipage}\hfill
  \begin{minipage}{0.65\textwidth}
    \[
      \Theta \left( n-1 \vdash \phi[x_{n-2}/x_{n-1}] \right) =
      \lower16pt\hbox{$\includegraphics[height=45pt]{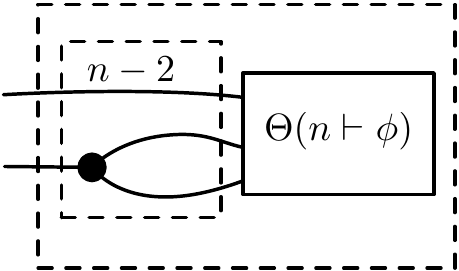}$} \tag{$\mathsf{Id}_{n}$}
    \]
  \end{minipage}\hfill

  \begin{minipage}{0.45\textwidth}
    \[
      \Theta \left( n \vdash R(x_0,\dots,x_{n-1}) \right) = \lower6pt\hbox{$\includegraphics[height=20pt]{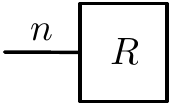}$} \tag{$\Sigma$}
    \]
  \end{minipage}
  \begin{minipage}{0.55\textwidth}
    \[
      \Theta \left( n-1 \vdash \exists x_{n-1}.\phi \right) =
      \lower10pt\hbox{$\includegraphics[height=30pt]{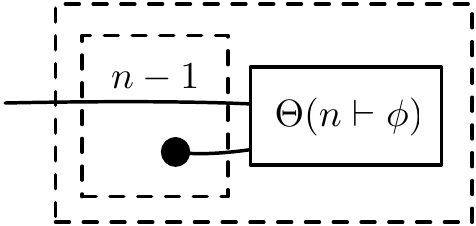}$} \tag{$\exists$}
    \]
  \end{minipage}\hfill

  \begin{minipage}{0.5\textwidth}
    \[
      \Theta \left( n+1 \vdash \phi \right) = \lower16pt\hbox{$\includegraphics[height=40pt]{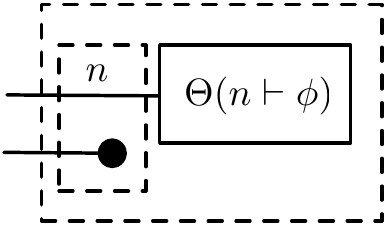}$} \tag{$\mathsf{Nu}_n$}
    \]
  \end{minipage}\hfill
  \begin{minipage}{0.5\textwidth}
    \[
      \Theta \left( m+n \vdash \phi \wedge (\psi[\dots]) \right) =
      \lower20pt\hbox{$\includegraphics[height=40pt]{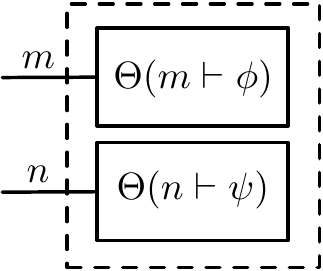}$} \tag{$\wedge$}
    \]
  \end{minipage}\hfill

\caption{Translation $\Theta$ from CCQ to GCQ.\label{fig:trans}}
\end{figure*}

We now give a semantics preserving translation $\Theta$ from CCQ to GCQ. For each CCQ relation symbol $R\in\Sigma$ of arity $n$, we assume a corresponding GCQ symbol $R\in\Sigma_{n,0}$.
Using Proposition~\ref{pro:syntax}, it suffices to consider judgments $n\vdash \phi$.
 For each, we obtain a GCQ term $\Theta(n\vdash \phi):\sort{n}{0}$.
The translation $\Theta$, given in Figure~\ref{fig:trans}, is defined by recursion on the derivation of $n \vdash \phi$. Given a CCQ model $\mathcal{M}=(X,\rho)$, let $\Theta(\mathcal{M})=(X,\rho')$ be the obvious corresponding GCQ model: $\rho'(R)=\rho(R)\times\{\bullet\}$.
The following confirms that semantics is preserved.
\begin{proposition}\label{thm:logictodiagrams}
For a CCQ model $\mathcal{M}=(X,\rho)$:
$\overset{\rightarrow}{v}\in \densem{n\vdash \phi}{\mathcal{M}}$ iff
$(\overset{\rightarrow}{v}, \bullet) \in \densem{\Theta(n\vdash \phi)}{\Theta(\mathcal{M})}$.
\end{proposition}
Furthermore, to characterise query inclusion in CCQ, it is enough to characterise it in GCQ.
\begin{proposition}\label{prop:CCQGCQ}
For all CCQ formulas $n\vdash \phi$ and $n\vdash \psi$, $\phi \leqq_{CCQ} \psi$ iff
$\Theta(\phi) \leqq_{GCQ} \Theta(\psi)$.
\end{proposition}
Proposition~\ref{thm:logictodiagrams} yields the left-to-right direction. For right-to-left, we give a semantics-preserving translation $\Lambda$ from GCQ to CCQ in Appendix~\ref{s:GCQtoCCQ}. Modulo $\equiv$, $\Lambda$ is inverse of $\Theta$. 
\begin{proposition}\label{prop:GCQCCQ} 
There exists a semantics preserving translation $\Lambda$ from GCQ to CCQ such that
for all GCQ terms $c,d \colon \sort{n}{m}$,  
it holds that  $c \leqq_{GCQ} d$ iff
$\Lambda(c) \leqq_{CCQ} \Lambda(d)$.
\end{proposition}
Propositions~\ref{prop:CCQGCQ} and~\ref{prop:GCQCCQ} together imply that $CCQ$ and $GCQ$ have the same expressive power.

\begin{example}
Recall from Example~\ref{ex:boneempty}, that $\lower5pt\hbox{$\includegraphics[height=15pt]{graffles/boxedbone.pdf}$}$ is related to $\exists x. \top$. By translating the CCQ formula $0\vdash \exists x_0.\top$ via $\Theta$, one obtains $\lower10pt\hbox{$\includegraphics[height=25pt]{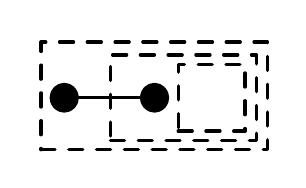}$}$. The latter and $\lower5pt\hbox{$\includegraphics[height=15pt]{graffles/boxedbone.pdf}$}$ are different---syntactically---but they are equal modulo $\equiv$. Note that their diagrams are similar: in the next section, we prove that terms differing only by dashed boxes are equal modulo $\equiv$.
\end{example}

\section{From terms to string diagrams}\label{sec:todiagram}

The first step towards an equational characterisation of query inclusion is to move from GCQ, where the graphical notation was a faithful representation of ordinary syntactic terms, to bona fide \emph{string diagrams}; that is, graphical notation for the arrows of a prop, a particularly simple kind of symmetric monoidal category (SMC). This is an advantage of GCQ syntax: its operations are amenable to an elegant axiomatisation. A hint of the good behaviour of GCQ operations is that query inclusion (and, therefore, query equivalence is) a (pre)congruence.
\begin{lemma}~
\begin{enumerate}[(i)]
\item Let $c,c':\sort{n}{m}$ and $d,d':\sort{m}{k}$ with $c\leqq c'$ and $d\leqq d'$. Then 
$(c\poi d) \leqq (c'\poi d')$.
\item Let $c,c':\sort{n}{m}$ and $d,d:\sort{p}{q}$ with $c\leqq c'$ and $d\leqq d'$
Then $(c\tns d) \leqq (c' \tns d')$.
\end{enumerate}
\label{lem:precongruence}
\end{lemma}

We now consider the laws of strict symmetric monoidal categories (Figure~\ref{fig:axsmc}) and discover that any two GCQ terms identified by them are logically equivalent. This means that we can eliminate the clutter of dashed boxes from our graphical notation. 

\begin{figure*}
\begin{tabular}{cc}
$\begin{array}{c}
\lower11pt\hbox{$\includegraphics[height=25pt]{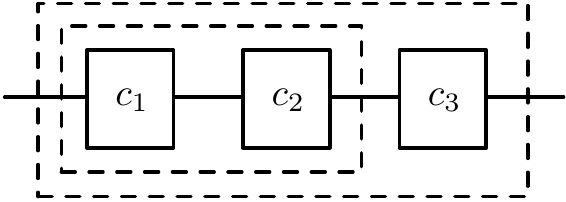}$} \stackrel{(i)}{=} 
\lower11pt\hbox{$\includegraphics[height=25pt]{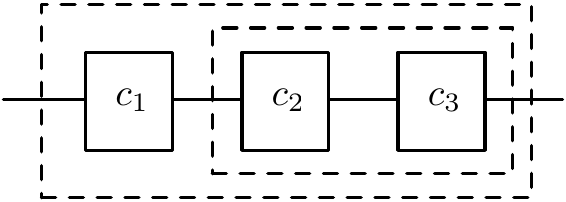}$} \vspace{2pt} \\
\lower8pt\hbox{$\includegraphics[height=20pt]{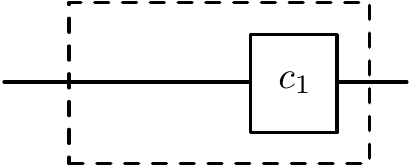}$} \stackrel{(ii)}{=} 
\lower4pt\hbox{$\includegraphics[height=10pt]{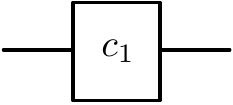}$} \stackrel{(ii)}{=} 
\lower8pt\hbox{$\includegraphics[height=20pt]{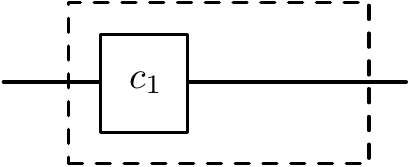}$}
\end{array}$
&
$\begin{array}{cc}
\lower23pt\hbox{$\includegraphics[height=50pt]{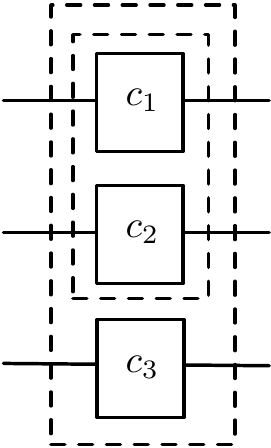}$} \stackrel{(iii)}{=}  
\lower24pt\hbox{$\includegraphics[height=50pt]{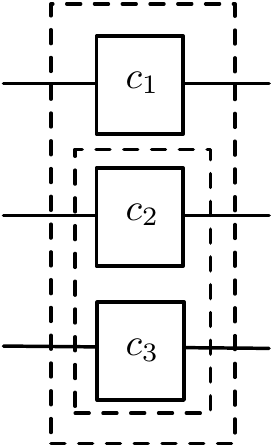}$} &
\lower15pt\hbox{$\includegraphics[height=30pt]{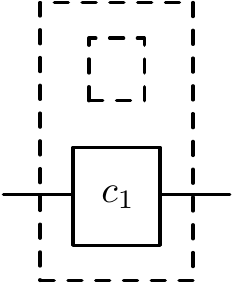}$} \stackrel{(iv)}{=} 
\lower3pt\hbox{$\includegraphics[height=10pt]{graffles/c1clean.pdf}$} \stackrel{(iv)}{=} 
\lower15pt\hbox{$\includegraphics[height=30pt]{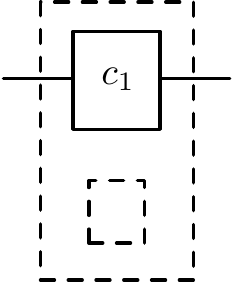}$}
\end{array}$
\end{tabular}
\begin{tabular}{ccc}
$\begin{array}{c}\lower20pt\hbox{$\includegraphics[height=40pt]{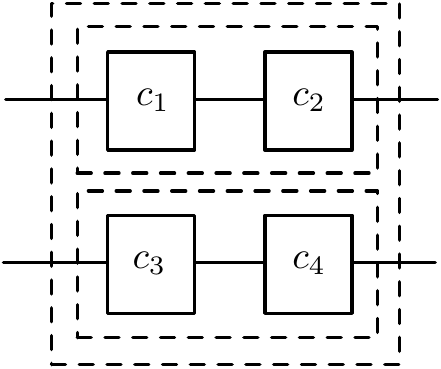}$} \stackrel{(v)}{=}  
\lower20pt\hbox{$\includegraphics[height=40pt]{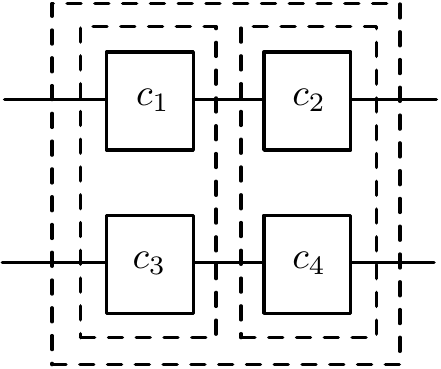}$}
\end{array}$
&
$\begin{array}{cc}
\lower9pt\hbox{$\includegraphics[height=23pt]{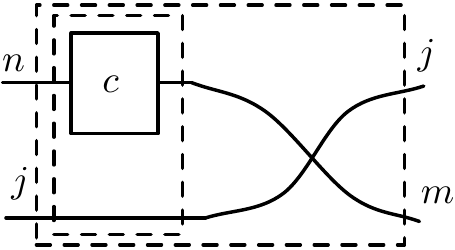}$}  \stackrel{(vi)}{=} 
\lower13pt\hbox{$\includegraphics[height=28pt]{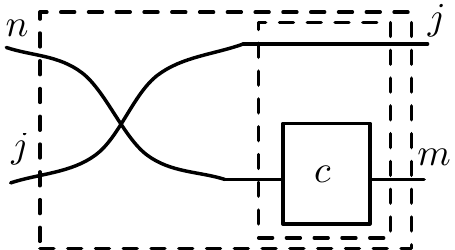}$} \vspace{2pt} \\
\lower13pt\hbox{$\includegraphics[height=28pt]{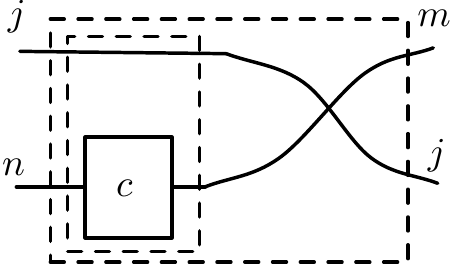}$} \stackrel{(vii)}{=} 
\lower9pt\hbox{$\includegraphics[height=23pt]{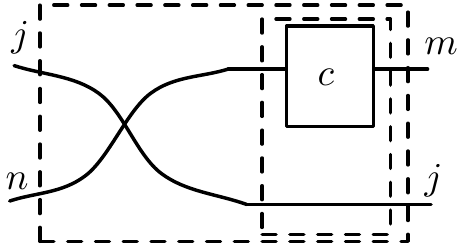}$}
\end{array}$
&
$\begin{array}{c}
\lower8pt\hbox{$\includegraphics[height=18pt]{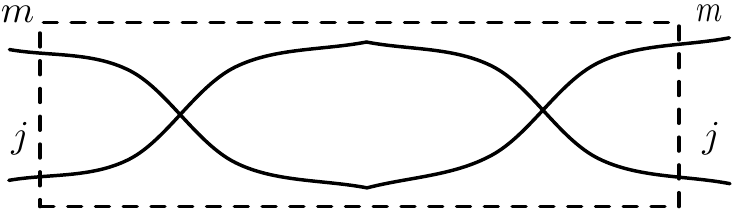}$}  \stackrel{(viii)}{=} 
\lower10pt\hbox{$\includegraphics[height=18pt]{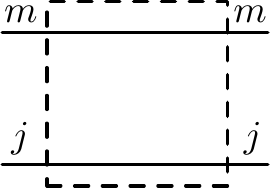}$}
\end{array}$
\end{tabular}
\caption{Axioms of strict symmetric monoidal categories. Wire annotations in $(i)$-$(v)$ have been omitted for clarity.\label{fig:axsmc}}
\end{figure*}

\begin{proposition}\label{pro:smc}
$\equiv$ satisfies the axioms in Figure~\ref{fig:axsmc}.
\end{proposition}

The terms of GCQ up-to query equivalence, therefore, organise themselves as arrows of a \emph{monoidal category} (axioms $(i)$-$(v)$), and the operation of ``erasing all dotted boxes'' from diagrams is well-defined. The resulting structure is the well-known combinatorial/topological concept of \emph{string diagram}. Equality reduces to the connectivity of their components, and is thus stable under intuitive topological transformations, known as \emph{diagrammatic reasoning}. For instance, axioms (ii) and (v) in Figure~\ref{fig:axsmc} imply that for 
$c_1:\sort{m_1}{n_1}$ and $c_2:\sort{m_2}{n_2}$  
\[
\lower14pt\hbox{$\includegraphics[height=30pt]{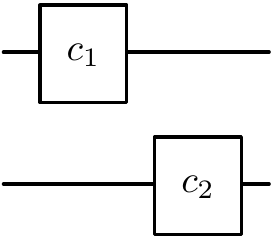}$} \quad \equiv \quad \lower14pt\hbox{$\includegraphics[height=30pt]{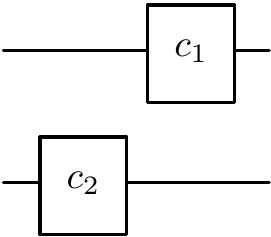}$}.
\]
Axioms $(vi)$-$(viii)$ assert that GCQ terms modulo $\equiv$ form a \emph{symmetric} monoidal category (SMC). Therein, \dsymNetL{n}{m}
 stands for the crossing of $n$ wires over $m$ wires. This has a standard recursive definition, using $\dsymNet$, $\idone$ and the operations of GCQ (see Appendix~\ref{app:sugar}). Intuitively, boxes ``slide over'' wire crossings. Moreover, it is well-known that 
$(vi)$ and $(vii)$ of Figure~\ref{fig:axsmc} imply the Yang-Baxter equation for crossings, which---with $(viii)$---implies that in diagrammatic reasoning wires do not ``tangle'' and crossings act like permutations of finite sets.

\section{Axiomatisation}
\label{sec:axioms}

We have seen that, up-to query equivalence, GCQ enjoys the structural properties of SMCs. Here we give further properties that \emph{characterise} query equivalence ($\equiv$) and inclusion ($\leqq$).

Our first observation is that $\Bmult$ and $\Bunit$ form, modulo $\equiv$, a \emph{commutative monoid}, i.e., they satisfy axioms $(A)$, $(C)$ and $(U)$ in Figure~\ref{fig:frobeniusBimonoid}. Similarly, $\Bcomult$ and $\Bcounit$ form a \emph{cocommutative comonoid} (axioms $(\op{A})$, $(\op{C})$ and $(\op{U})$). Monoid and comonoid together give rise to a \emph{special Frobenius bimonoid} (axioms $(S)$ and $(F)$), a well-known algebraic structure that is important in various domains \cite{Abramsky2008:CQM,PQP,Bruni2006,Bonchi2015}.

\begin{proposition}
$\equiv$ satisfies the axioms in Figure~\ref{fig:frobeniusBimonoid}.
\end{proposition}

Figure~\ref{fig:adjoint} shows a set of properties of query inclusion. The two axioms on the left state that $\Bcounit$ is the left adjoint of $\Bunit$ and the central axioms assert that $\Bcomult$ is the left adjoint of $\Bmult$.
For the rightmost ones, it is convenient to introduce some syntactic sugar: $\BcomultL{n}$, $\BmultL{n}$, $\BcounitL{n}$ and $\BunitL{n}$ stand for the $n$-fold versions of monoid and comonoid (see Appendix~\ref{app:sugar} for the definition). Now, axiom $(L_1)$ asserts that $\Rcircuit$ laxly commutes with $\BcounitL{m}$, while axiom $(L_2)$ states that it laxly commutes with $\BcomultL{m}$. In a nutshell, $\Rcircuit$ is required to be a \emph{lax comonoid morphism}. 
\begin{proposition}
$\leqq$ satisfies the axioms of Figure~\ref{fig:adjoint}.
\end{proposition}

\begin{figure*}[th]
\begin{tabular}{c|c|c}
$\begin{array}{ccc}
  \lower10pt\hbox{$\includegraphics[height=0.8cm,keepaspectratio]{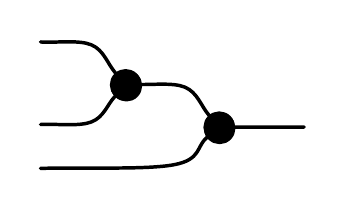}$} & \stackrel{(A)}{=} & \lower10pt\hbox{$\includegraphics[height=0.8cm,keepaspectratio]{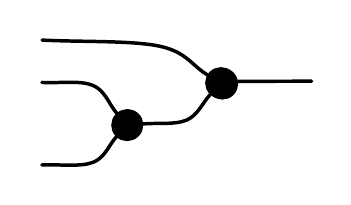}$}\\
  \lower7pt\hbox{$\includegraphics[height=0.6cm,keepaspectratio]{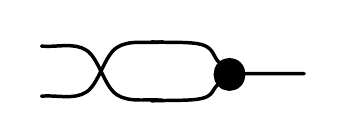}$} &  \stackrel{(C)}{=} & \lower10pt\hbox{$\includegraphics[height=0.8cm,keepaspectratio]{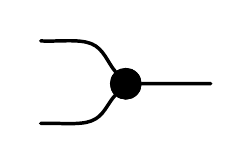}$}\\
\lower10pt\hbox{$\includegraphics[height=0.8cm,keepaspectratio]{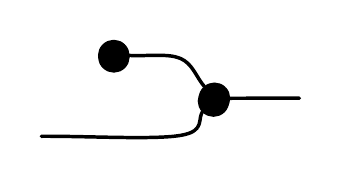}$} &  \stackrel{(U)}{=} & \lower3pt\hbox{$\includegraphics[width=1.0cm,keepaspectratio]{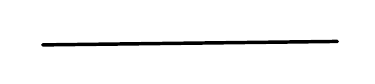}$}
\end{array}$
&
$\begin{array}{ccc}
  \lower10pt\hbox{$\includegraphics[height=0.8cm]{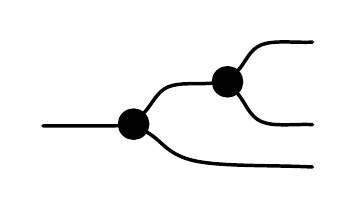}$} & \stackrel{(\op{A})}{=} & \lower10pt\hbox{$\includegraphics[height=0.8cm]{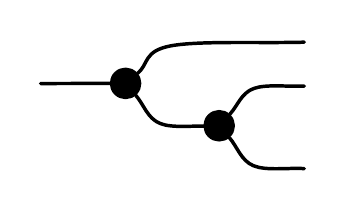}$}\\
  \lower7pt\hbox{$\includegraphics[height=0.6cm]{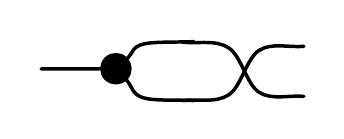}$} &  \stackrel{(\op{C})}{=} & \lower10pt\hbox{$\includegraphics[height=0.8cm]{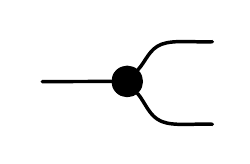}$}\\
  \lower10pt\hbox{\reflectbox{\rotatebox[origin=c]{180}{$\includegraphics[height=0.8cm]{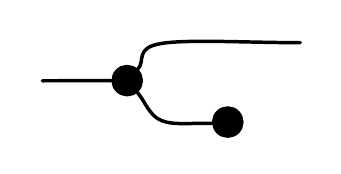}$}}} &  \stackrel{(\op{U})}{=} & \lower3pt\hbox{$\includegraphics[width=1.0cm]{graffles/id.pdf}$}
\end{array}$
&
$\begin{array}{ccc}
  \lower4pt\hbox{$\includegraphics[height=0.4cm]{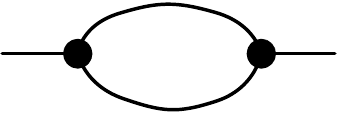}$} &  \stackrel{(S)}{=} & \lower3pt\hbox{$\includegraphics[width=1.0cm]{graffles/id.pdf}$}\\
\lower10pt\hbox{$\includegraphics[height=0.8cm]{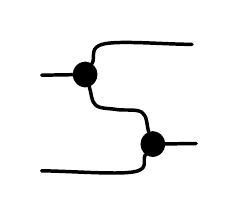}$} & \stackrel{(F)}{=} & \lower10pt\hbox{$\includegraphics[height=0.8cm]{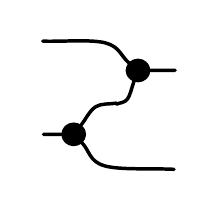}$}
\end{array}$
\end{tabular}
\caption{Axioms for special Frobenius bimonoids.}
\label{fig:frobeniusBimonoid}
\end{figure*}

\begin{figure*}[th]
\begin{tabular}{c|c|c}
$\begin{array}{ccc}
  \lower7pt\hbox{$\includegraphics[height=0.5cm,keepaspectratio]{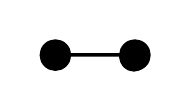}$} & \stackrel{(UC)}{\leq} & \lower10pt\hbox{$\includegraphics[height=1cm]{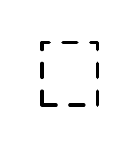}$}\\
  \lower2pt\hbox{$\includegraphics[width=1cm]{graffles/id.pdf}$} &
 \stackrel{(CU)}{\leq} & 
\lower7pt\hbox{$\includegraphics[height=0.5cm]{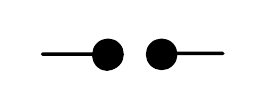}$}
\end{array}$
&
$\begin{array}{ccc}
  \lower10pt\hbox{$\includegraphics[height=0.8cm]{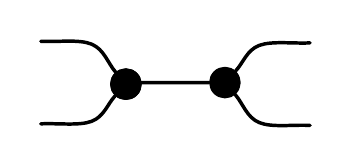}$} & \stackrel{(MC)}{\leq} & \lower6pt\hbox{$\includegraphics[width=1cm]{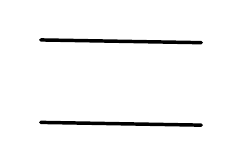}$}\\
\lower3pt\hbox{$\includegraphics[width=1cm]{graffles/id.pdf}$}
 &  \stackrel{(CM)}{\leq} &
\lower4pt\hbox{$\includegraphics[height=0.4cm,keepaspectratio]{graffles/special.pdf}$}
\end{array}$
&
$\begin{array}{ccc}
  \lower8pt\hbox{$\includegraphics[height=0.6cm]{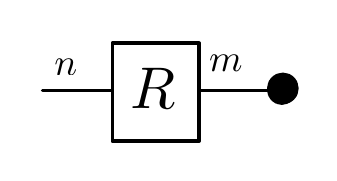}$} & \stackrel{(L_1)}{\leq} & \hspace*{-1.5cm}\lower4pt\hbox{$\includegraphics[height=0.5cm]{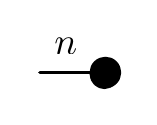}$}\\
\lower7pt\hbox{$\includegraphics[height=0.6cm]{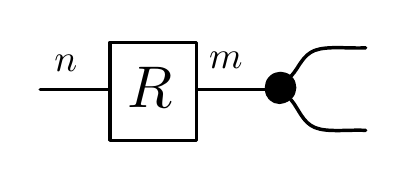}$} &  \stackrel{(L_2)}{\leq} & \lower13pt\hbox{$\includegraphics[height=1.1cm]{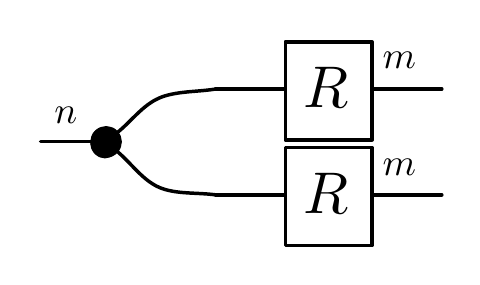}$}
\end{array}$
\end{tabular}
\caption[Axioms for adjointness and lax comonoid homomorphism]{Axioms for adjointness of $\Bcounit$ and $\Bunit$ (left) adjointness of $\Bcomult$ and $\Bmult$ (center) lax comonoid morphism (right).}
\label{fig:adjoint}
\end{figure*}

Interestingly, 
the observations we made so far suffice to characterise query equivalence and inclusion.
This is the main theorem which we will prove in the remainder of this paper.

\begin{definition}
  The relation $\synleq$ on the terms of GCQ is the smallest precongruence containing the equalities in Figures~\ref{fig:axsmc},~\ref{fig:frobeniusBimonoid}, their converses and the inequalities in Figure~\ref{fig:adjoint}. The relation $\syneq$ is the intersection of $\synleq$ and its converse.
\end{definition}

\begin{theorem}\label{thm:main}
  $ \synleq = \leqq$
\end{theorem}

\begin{remark}
  There is an apparent redundancy in Figure~\ref{fig:adjoint}:
  $(CM)$ follows immediately from $(S)$ in Figure~\ref{fig:frobeniusBimonoid}, while
  $(S)$ can by derived from $(CU)$, $(\op{U})$ and $(U)$ for one inclusion and $(CM)$ for the other.
  We kept both $(CM)$ and $(S)$ because, as we shall see in \S\ref{sec:cospans}, it is important to keep the algebraic structures of Figures~\ref{fig:frobeniusBimonoid} and~\ref{fig:adjoint} separate. 
\end{remark}

\begin{example}\label{ex:derivation}
  Recall the example from the Introduction. We can now prove the inclusion of queries using diagrammatic reasoning, as shown below. In the unlabeled equality we make use of the well-known \emph{spider theorem}, which holds in every special Frobenius algebra~\cite{lack2004composing}.
\[
\raisebox{-6mm}{\scalebox{\exscale}{
\begin{tikzpicture}
\draw (0.75,2.0) .. controls (0.75,2.0) and (0.5,2.0) .. (0.0,2.0);
\draw (0.75,2.0) .. controls (0.75,1.0) and (1.0,1.0) .. (1.5,1.0);
\draw (0.75,2.0) .. controls (0.75,3.0) and (1.0,3.0) .. (1.5,3.0);
\draw (0.75,2.0) \blk {};
\draw (0.75,6.0) .. controls (0.75,6.0) and (0.5,6.0) .. (0.0,6.0);
\draw (0.75,6.0) .. controls (0.75,5.0) and (1.0,5.0) .. (1.5,5.0);
\draw (0.75,6.0) .. controls (0.75,7.0) and (1.0,7.0) .. (1.5,7.0);
\draw (0.75,6.0) \blk {};
\draw (2.25,1.0) .. controls (2.25,1.0) and (2.0,1.0) .. (1.5,1.0);
\draw (2.25,1.0) .. controls (2.25,1.0) and (2.5,1.0) .. (3.0,1.0);
\node[mat] at (2.25,1.0) {$R$};
\draw (2.25,3.0) .. controls (2.25,3.0) and (2.0,3.0) .. (1.5,3.0);
\draw (2.25,3.0) .. controls (2.25,3.0) and (2.5,3.0) .. (3.0,3.0);
\node[mat] at (2.25,3.0) {$R$};
\draw (2.25,5.0) .. controls (2.25,5.0) and (2.0,5.0) .. (1.5,5.0);
\draw (2.25,5.0) .. controls (2.25,5.0) and (2.5,5.0) .. (3.0,5.0);
\node[mat] at (2.25,5.0) {$R$};
\draw (2.25,7.0) .. controls (2.25,7.0) and (2.0,7.0) .. (1.5,7.0);
\draw (2.25,7.0) .. controls (2.25,7.0) and (2.5,7.0) .. (3.0,7.0);
\node[mat] at (2.25,7.0) {$R$};
\draw (3.0,1.0) .. controls (3.5,1.0) and (4.0,1.0) .. (4.5,1.0);
\draw (3.0,3.0) .. controls (3.5,3.0) and (4.0,5.0) .. (4.5,5.0);
\draw (3.0,5.0) .. controls (3.5,5.0) and (4.0,3.0) .. (4.5,3.0);
\draw (3.0,7.0) .. controls (3.5,7.0) and (4.0,7.0) .. (4.5,7.0);
\draw (5.25,2.0) .. controls (5.25,1.0) and (5.0,1.0) .. (4.5,1.0);
\draw (5.25,2.0) .. controls (5.25,3.0) and (5.0,3.0) .. (4.5,3.0);
\draw (5.25,2.0) .. controls (5.25,2.0) and (5.5,2.0) .. (6.0,2.0);
\draw (5.25,2.0) \blk {};
\draw (6.25,2.0) .. controls (6.25,2.0) and (6.5,2.0) .. (6.0,2.0);
\draw (6.25,2.0) \blk {};
\draw (5.25,6.0) .. controls (5.25,5.0) and (5.0,5.0) .. (4.5,5.0);
\draw (5.25,6.0) .. controls (5.25,7.0) and (5.0,7.0) .. (4.5,7.0);
\draw (5.25,6.0) .. controls (5.25,6.0) and (5.5,6.0) .. (6.0,6.0);
\draw (5.25,6.0) \blk {};
\draw (6.25,6.0) .. controls (6.25,6.0) and (6.5,6.0) .. (6.0,6.0);
\draw (6.25,6.0) \blk {};
\end{tikzpicture}}}
\overset{(L_2)}{\geq}
  \raisebox{-6mm}{\scalebox{\exscale}{
\begin{tikzpicture}
\draw (0.0,2.0) .. controls (0.5,2.0) and (-0.05,2.0) .. (0.45000002,2.0);
\draw (1.2,2.0) .. controls (1.2,2.0) and (0.95000005,2.0) .. (0.45000002,2.0);
\draw (1.2,2.0) .. controls (1.2,2.0) and (1.45,2.0) .. (1.95,2.0);
\node[mat] at (1.2,2.0) {$R$};
\draw (2.7,2.0) .. controls (2.7,2.0) and (2.45,2.0) .. (1.95,2.0);
\draw (2.7,2.0) .. controls (2.7,1.0) and (2.95,1.0) .. (3.45,1.0);
\draw (2.7,2.0) .. controls (2.7,3.0) and (2.95,3.0) .. (3.45,3.0);
\draw (2.7,2.0) \blk {};
\draw (0.0,6.0) .. controls (0.5,6.0) and (-0.05,6.0) .. (0.45000002,6.0);
\draw (1.2,6.0) .. controls (1.2,6.0) and (0.95000005,6.0) .. (0.45000002,6.0);
\draw (1.2,6.0) .. controls (1.2,6.0) and (1.45,6.0) .. (1.95,6.0);
\node[mat] at (1.2,6.0) {$R$};
\draw (2.7,6.0) .. controls (2.7,6.0) and (2.45,6.0) .. (1.95,6.0);
\draw (2.7,6.0) .. controls (2.7,5.0) and (2.95,5.0) .. (3.45,5.0);
\draw (2.7,6.0) .. controls (2.7,7.0) and (2.95,7.0) .. (3.45,7.0);
\draw (2.7,6.0) \blk {};
\draw (3.45,1.0) .. controls (3.95,1.0) and (4.45,1.0) .. (4.95,1.0);
\draw (3.45,3.0) .. controls (3.95,3.0) and (4.45,5.0) .. (4.95,5.0);
\draw (3.45,5.0) .. controls (3.95,5.0) and (4.45,3.0) .. (4.95,3.0);
\draw (3.45,7.0) .. controls (3.95,7.0) and (4.45,7.0) .. (4.95,7.0);
\draw (5.7,2.0) .. controls (5.7,1.0) and (5.45,1.0) .. (4.95,1.0);
\draw (5.7,2.0) .. controls (5.7,3.0) and (5.45,3.0) .. (4.95,3.0);
\draw (5.7,2.0) .. controls (5.7,2.0) and (5.95,2.0) .. (6.45,2.0);
\draw (5.7,2.0) \blk {};
\draw (6.7,2.0) .. controls (6.7,2.0) and (6.95,2.0) .. (6.45,2.0);
\draw (6.7,2.0) \blk {};
\draw (5.7,6.0) .. controls (5.7,5.0) and (5.45,5.0) .. (4.95,5.0);
\draw (5.7,6.0) .. controls (5.7,7.0) and (5.45,7.0) .. (4.95,7.0);
\draw (5.7,6.0) .. controls (5.7,6.0) and (5.95,6.0) .. (6.45,6.0);
\draw (5.7,6.0) \blk {};
\draw (6.7,6.0) .. controls (6.7,6.0) and (6.95,6.0) .. (6.45,6.0);
\draw (6.7,6.0) \blk {};
\end{tikzpicture}}} =
\raisebox{-2.5mm}{\scalebox{\exscale}{
\begin{tikzpicture}
\draw (0.0,1.0) .. controls (0.5,1.0) and (-0.05,1.0) .. (0.45000002,1.0);
\draw (1.2,1.0) .. controls (1.2,1.0) and (0.95000005,1.0) .. (0.45000002,1.0);
\draw (1.2,1.0) .. controls (1.2,1.0) and (1.45,1.0) .. (1.95,1.0);
\node[mat] at (1.2,1.0) {$R$};
\draw (0.0,3.0) .. controls (0.5,3.0) and (-0.05,3.0) .. (0.45000002,3.0);
\draw (1.2,3.0) .. controls (1.2,3.0) and (0.95000005,3.0) .. (0.45000002,3.0);
\draw (1.2,3.0) .. controls (1.2,3.0) and (1.45,3.0) .. (1.95,3.0);
\node[mat] at (1.2,3.0) {$R$};
\draw (2.7,2.0) .. controls (2.7,1.0) and (2.45,1.0) .. (1.95,1.0);
\draw (2.7,2.0) .. controls (2.7,3.0) and (2.45,3.0) .. (1.95,3.0);
\draw (2.7,2.0) .. controls (2.7,2.0) and (2.95,2.0) .. (3.45,2.0);
\draw (2.7,2.0) \blk {};
\draw (4.2,2.0) .. controls (4.2,2.0) and (3.95,2.0) .. (3.45,2.0);
\draw (4.2,2.0) .. controls (4.2,1.0) and (4.45,1.0) .. (4.95,1.0);
\draw (4.2,2.0) .. controls (4.2,3.0) and (4.45,3.0) .. (4.95,3.0);
\draw (4.2,2.0) \blk {};
\draw (5.2,1.0) .. controls (5.2,1.0) and (5.45,1.0) .. (4.95,1.0);
\draw (5.2,1.0) \blk {};
\draw (5.2,3.0) .. controls (5.2,3.0) and (5.45,3.0) .. (4.95,3.0);
\draw (5.2,3.0) \blk {};
\end{tikzpicture}}}
  \overset{(MC)}{\geq}
\raisebox{-2.5mm}{\scalebox{\exscale}{
\begin{tikzpicture}
\draw (0.75,2.0) .. controls (0.75,1.0) and (0.5,1.0) .. (0.0,1.0);
\draw (0.75,2.0) .. controls (0.75,3.0) and (0.5,3.0) .. (0.0,3.0);
\draw (0.75,2.0) .. controls (0.75,2.0) and (1.0,2.0) .. (1.5,2.0);
\draw (0.75,2.0) \blk {};
\draw (2.25,2.0) .. controls (2.25,2.0) and (2.0,2.0) .. (1.5,2.0);
\draw (2.25,2.0) .. controls (2.25,1.0) and (2.5,1.0) .. (3.0,1.0);
\draw (2.25,2.0) .. controls (2.25,3.0) and (2.5,3.0) .. (3.0,3.0);
\draw (2.25,2.0) \blk {};
\draw (3.75,1.0) .. controls (3.75,1.0) and (3.5,1.0) .. (3.0,1.0);
\draw (3.75,1.0) .. controls (3.75,1.0) and (4.0,1.0) .. (4.5,1.0);
\node[mat] at (3.75,1.0) {$R$};
\draw (3.75,3.0) .. controls (3.75,3.0) and (3.5,3.0) .. (3.0,3.0);
\draw (3.75,3.0) .. controls (3.75,3.0) and (4.0,3.0) .. (4.5,3.0);
\node[mat] at (3.75,3.0) {$R$};
\draw (5.25,2.0) .. controls (5.25,1.0) and (5.0,1.0) .. (4.5,1.0);
\draw (5.25,2.0) .. controls (5.25,3.0) and (5.0,3.0) .. (4.5,3.0);
\draw (5.25,2.0) .. controls (5.25,2.0) and (5.5,2.0) .. (6.0,2.0);
\draw (5.25,2.0) \blk {};
\draw (6.75,2.0) .. controls (6.75,2.0) and (6.5,2.0) .. (6.0,2.0);
\draw (6.75,2.0) .. controls (6.75,1.0) and (7.0,1.0) .. (7.5,1.0);
\draw (6.75,2.0) .. controls (6.75,3.0) and (7.0,3.0) .. (7.5,3.0);
\draw (6.75,2.0) \blk {};
\draw (7.75,1.0) .. controls (7.75,1.0) and (8.0,1.0) .. (7.5,1.0);
\draw (7.75,1.0) \blk {};
\draw (7.75,3.0) .. controls (7.75,3.0) and (8.0,3.0) .. (7.5,3.0);
\draw (7.75,3.0) \blk {};
\end{tikzpicture}}}
\]
\[
\overset{(L_2)}{\geq}
\raisebox{-2.5mm}{\scalebox{\exscale}{
\begin{tikzpicture}
\draw (0.75,2.0) .. controls (0.75,1.0) and (0.5,1.0) .. (0.0,1.0);
\draw (0.75,2.0) .. controls (0.75,3.0) and (0.5,3.0) .. (0.0,3.0);
\draw (0.75,2.0) .. controls (0.75,2.0) and (1.0,2.0) .. (1.5,2.0);
\draw (0.75,2.0) \blk {};
\draw (2.25,2.0) .. controls (2.25,2.0) and (2.0,2.0) .. (1.5,2.0);
\draw (2.25,2.0) .. controls (2.25,2.0) and (2.5,2.0) .. (3.0,2.0);
\node[mat] at (2.25,2.0) {$R$};
\draw (3.75,2.0) .. controls (3.75,2.0) and (3.5,2.0) .. (3.0,2.0);
\draw (3.75,2.0) .. controls (3.75,1.0) and (4.0,1.0) .. (4.5,1.0);
\draw (3.75,2.0) .. controls (3.75,3.0) and (4.0,3.0) .. (4.5,3.0);
\draw (3.75,2.0) \blk {};
\draw (5.25,2.0) .. controls (5.25,1.0) and (5.0,1.0) .. (4.5,1.0);
\draw (5.25,2.0) .. controls (5.25,3.0) and (5.0,3.0) .. (4.5,3.0);
\draw (5.25,2.0) .. controls (5.25,2.0) and (5.5,2.0) .. (6.0,2.0);
\draw (5.25,2.0) \blk {};
\draw (6.75,2.0) .. controls (6.75,2.0) and (6.5,2.0) .. (6.0,2.0);
\draw (6.75,2.0) .. controls (6.75,1.0) and (7.0,1.0) .. (7.5,1.0);
\draw (6.75,2.0) .. controls (6.75,3.0) and (7.0,3.0) .. (7.5,3.0);
\draw (6.75,2.0) \blk {};
\draw (7.75,1.0) .. controls (7.75,1.0) and (8.0,1.0) .. (7.5,1.0);
\draw (7.75,1.0) \blk {};
\draw (7.75,3.0) .. controls (7.75,3.0) and (8.0,3.0) .. (7.5,3.0);
\draw (7.75,3.0) \blk {};
\end{tikzpicture}}} \overset{(S)}{=}
\raisebox{-2.5mm}{\scalebox{\exscale}{
\begin{tikzpicture}
\draw (0.75,2.0) .. controls (0.75,1.0) and (0.5,1.0) .. (0.0,1.0);
\draw (0.75,2.0) .. controls (0.75,3.0) and (0.5,3.0) .. (0.0,3.0);
\draw (0.75,2.0) .. controls (0.75,2.0) and (1.0,2.0) .. (1.5,2.0);
\draw (0.75,2.0) \blk {};
\draw (2.25,2.0) .. controls (2.25,2.0) and (2.0,2.0) .. (1.5,2.0);
\draw (2.25,2.0) .. controls (2.25,2.0) and (2.5,2.0) .. (3.0,2.0);
\node[mat] at (2.25,2.0) {$R$};
\draw (3.75,2.0) .. controls (3.75,2.0) and (3.5,2.0) .. (3.0,2.0);
\draw (3.75,2.0) .. controls (3.75,1.0) and (4.0,1.0) .. (4.5,1.0);
\draw (3.75,2.0) .. controls (3.75,3.0) and (4.0,3.0) .. (4.5,3.0);
\draw (3.75,2.0) \blk {};
\draw (4.75,1.0) .. controls (4.75,1.0) and (5.0,1.0) .. (4.5,1.0);
\draw (4.75,1.0) \blk {};
\draw (4.75,3.0) .. controls (4.75,3.0) and (5.0,3.0) .. (4.5,3.0);
\draw (4.75,3.0) \blk {};
\end{tikzpicture}}} \overset{(\op{U})}{=}
\raisebox{-2.5mm}{\scalebox{\exscale}{
\begin{tikzpicture}
\draw (0.75,2.0) .. controls (0.75,1.0) and (0.5,1.0) .. (0.0,1.0);
\draw (0.75,2.0) .. controls (0.75,3.0) and (0.5,3.0) .. (0.0,3.0);
\draw (0.75,2.0) .. controls (0.75,2.0) and (1.0,2.0) .. (1.5,2.0);
\draw (0.75,2.0) \blk {};
\draw (2.25,2.0) .. controls (2.25,2.0) and (2.0,2.0) .. (1.5,2.0);
\draw (2.25,2.0) .. controls (2.25,2.0) and (2.5,2.0) .. (3.0,2.0);
\node[mat] at (2.25,2.0) {$R$};
\draw (3.25,2.0) .. controls (3.25,2.0) and (3.5,2.0) .. (3.0,2.0);
\draw (3.25,2.0) \blk {};
\end{tikzpicture}}}
\]
\end{example}

\subsection{Cartesian bicategories}
\label{subsec:cartbicat}

The structure in Figures~\ref{fig:frobeniusBimonoid} and~\ref{fig:adjoint} is not arbitrary: these are exactly the laws of \emph{cartesian bicategories},
a concept introduced by Carboni and Walters~\cite{carboni1987cartesian}, that we recall below.

\begin{definition}\label{def:cartbicat}
  A cartesian bicategory is a symmetric monoidal category $\mathcal{B}$ with tensor $\tns$ and unit $I$, enriched over the category of partially ordered sets, such that:
  \begin{enumerate}
    \item every object $X$ has a special Frobenius bimonoid: a monoid $\BmultX \from X \tns X \to X$, $\BunitX \from I \to X$, a comonoid 
      $\BcomultX \from X \to X \tns X$, $\BcounitX \from X \to I$ satisfying the axioms in Figure~\ref{fig:frobeniusBimonoid};
    \item the monoid and comonoid on $X$ are adjoint (axioms in Figure~\ref{fig:adjoint}, left and center);
    \item every arrow $R \from X \to Y$ is a lax comonoid morphism (axioms in Figure~\ref{fig:adjoint}, right).
  \label{def:bicategoryOfRel}
  \end{enumerate}
  Furthermore, a morphism $\mathcal{F}$ of cartesian bicategories is a functor $\mathcal{F} \from \mathcal{B}_1 \to \mathcal{B}_2$ preserving
  the tensor, the partial orders and the monoid and comonoid on every object.
\end{definition}

\begin{example} \label{ex:Rel}
  The archetypal cartesian bicategory is the 
  category of sets and relations $\Rel$, with cartesian product $\times$ as tensor and $1=\{\bullet\}$ as unit. 
  To be precise, $\Rel$ has sets as objects and relations $R \subseteq X\times Y$ as arrows $X \to Y$. Composition and tensor are defined as in Figure~\ref{fig:gcqsemantics}.
 For each set $X$, the monoid and comonoid
 structure is:
 \resizebox{0.88\hsize}{!}{
 \[
   \BcomultX = \left\{ \left(x, \left( \begin{smallmatrix} x \\ x \end{smallmatrix} \right) \right) \mid x \in X \right\},
   \BcounitX = \left\{ \left(x, \bullet \right) \mid x \in X \right\},
   \BmultX = \left\{ \left( \left( \begin{smallmatrix} x \\ x \end{smallmatrix} \right) , x \right) \mid x \in X \right\},
  \BunitX = \left\{ \left(\bullet, x \right) \mid x \in X \right\}\text{.} \\
 \]}
\end{example}

Cartesian bicategories allow us to employ the usual construction from categorical logic:
the arrows of the cartesian bicategory freely generated from $\Sigma$ are GCQ terms modulo $\syneq$, and morphisms from this cartesian bicategory to $\Rel$ are exactly GCQ models.

\begin{definition} \label{def:gcqpobc}
  The ordered prop $\gcqpobc$ has GCQ terms of sort $\sort{n}{m}$ modulo $\syneq$ as arrows $n \to m$. These are ordered by $\synleq$.
\end{definition}
\begin{lemma}
  $\gcqpobc$ is a cartesian bicategory.
  \label{lem:gcqpobcbicat}
\end{lemma}

\begin{proposition}\label{lem:ModelsAreFunctors}\label{pro:ModelsAreFunctors}
  Models of GCQ  (Definition~\ref{defn:GCQmodel})  are in bijective correspondence with morphisms of cartesian bicategories $\gcqpobc \to \Rel$.
\end{proposition}

\newcommand{\catC}{\mathcal{C}}
\newcommand{\syntaxPROP}{\gcqprop}

\newcommand{\ra}{\rightarrow}
\newcommand{\la}{\leftarrow}

\newcommand{\node}{\lower0pt\hbox{$\includegraphics[width=6pt]{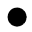}$}}
\newcommand{\hyperedge}{\lower5pt\hbox{$\includegraphics[width=30pt]{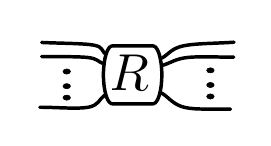}$}}
\newcommand{\SynToCsp}[1]{\ensuremath{\lfloor\!\lfloor{#1}\rfloor\!\rfloor}}
\newcommand{\CspToSyn}[1]{\ensuremath{\lceil\!\lceil{#1}\rceil\!\rceil}}

\newcommand{\cgr}[2][scale=0.45]{\raisebox{0.1em}{\begingroup
\setbox0=\hbox{\includegraphics[#1]{#2}}
\parbox{\wd0}{\box0}\endgroup}}

\section{Discrete cospans of hypergraphs}\label{sec:cospans}

In order to prove Theorem~\ref{thm:main}, in this section we give a combinatorial characterisation of free cartesian bicategories as  hypergraphs-with-interfaces, formalised as a (bi)category of \emph{cospans}
equipped with an ordering inspired by Merlin and Chandra~\cite{chandra1977optimal}.

Indeed, the appearance of graph-like structures in the context of conjunctive queries should not come as a shock. Merlin and Chandra, to compute inclusion  $\varphi\leqq \psi$ of CCQ queries, translate them into hypergraphs $G_\varphi, G_\psi$ with ``interfaces'' that represent free variables. Then $\varphi \leqq \psi$ iff there exists an interface-preserving homomorphism from $G_\psi$ to $G_\phi$. 

\subsection{Hypergraphs and Cospans}

Our goal in this part is the characterisation of GCQ diagrams as certain combinatorial structures. We start by introducing the notion of $\Sigma$-hypergraph.

\begin{definition}[$\Sigma$-hypergraph]\label{def:hyp}
Let $\Sigma$ be a monoidal signature.
A $\Sigma$-hypergraph $G$ is a set $G_V$ of vertices and, for each $R \in \Sigma_{n,m}$, a set of $R$-labeled hyperedges $G_{R}$, with source and target functions $s_R\from G_R \to (G_V)^n, t_R \from G_R \to (G_V)^m$.
A morphism $f\from G \to G'$ is a function $f_V\from G_V \to G'_{V}$ and a family $f_R \from G_{R} \to G'_{R}$, for each $R \in \Sigma_{n,m}$,  s.t.\ the following commutes.
 $$\xymatrix@R=11pt@C=10pt{
   (G_V)^n\ar[d]_(0.4){f_V} & E_R \ar[l]_(0.35){s_R} \ar[r]^(0.35){t_R} \ar[d]_(0.4){f_{R}} & (G_V)^m \ar[d]^(0.4){f_V} \\
   (G'_V)^n & E_R'  \ar[l]^(0.35){s'_R} \ar[r]_(0.35){t'_R} & (G'_V)^m
}$$
A $\Sigma$-hypergraph $G$ is finite if $G_V$ and $G_{R}$ are finite.
$\Sigma$-hypergraphs and morphisms form the category $\iHyp_\Sigma$. Its full subcategory of finite $\Sigma$-hypergraphs is denoted by $\Hyp_{\Sigma}$.
\end{definition}

We visualise hypergraphs as follows: $\node$ is a vertex and $\hyperedge$ is a hyperedge with ordered tentacles. An example is shown below left, where $S\in\Sigma_{1,0}$ and $R \in \Sigma_{2,1}$.
\begin{equation}
\lower12pt\hbox{$\includegraphics[height=1.3cm]{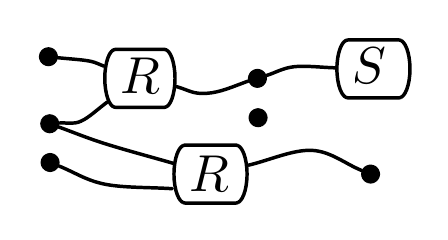}$} \quad \quad \quad \quad \lower12pt\hbox{$\includegraphics[height=1.5cm]{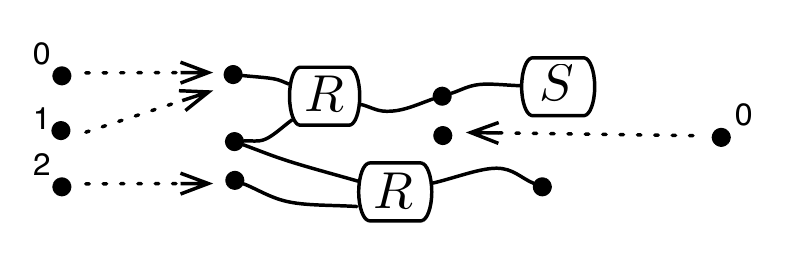}$}
  \label{eq:hypgraph}
\end{equation}

In order to capture GCQ diagrams, we need to equip hypergraphs with interfaces, as illustrated in~\eqref{eq:hypgraph} on the right.
Roughly speaking, an interface consists of two sets, called the left boundary and the right boundary. Each has an associated function to the underlying set of hypergraph vertices, depicted by the dotted arrows.
Graphical structures with interfaces are common in computer science, (e.g., in automata theory~\cite{glushkov1961abstract}, graph rewriting~\cite{HandbookDPO}, Petri nets~\cite{sassone2005congruence}). Categorically speaking, they are (discrete) cospans. 

\begin{definition}[$\mathsf{Cospan}$]
  \label{defn:cospan}
  Let $\catC$ be a finitely cocomplete category. A \emph{cospan} from $X$ to $Y$ is a
\begin{minipage}{0.78\textwidth}
pair of arrows $X \to A \leftarrow Y$ in $\catC$.
 A morphism $\alpha \from (X \to A \leftarrow Y) \Rightarrow (X \to B \leftarrow Y)$ is an arrow $\alpha \from A \to B$ in $\catC$ s.t. the diagram on the right commutes.
Cospans $X \to A \leftarrow Y$ and $X \to B \leftarrow Y$ are \emph{isomorphic} if 
\end{minipage}
\begin{minipage}{0.22\textwidth}
  \vspace{-15pt}
\begin{equation}\label{eq:preorder}
\raise14pt\hbox{$
\xymatrix@R=2pt@C=10pt{
      & A \ar[dd]^{\alpha} & \\
      X \ar@/_/[dr] \ar@/^/[ur] & & Y \ar@/_/[ul] \ar@/^/[dl] \\
      & B & \\
}$}
\end{equation}
\end{minipage}
there exists an isomorphism $A\to B$.
For $X \in \catC$, the \emph{identity cospan} is $X\xrightarrow{\id_X} X \xleftarrow{\id_X} X$.
The composition of $X\to  A \xleftarrow{f} Y$
and $Y\xrightarrow{g} B \leftarrow Z$ is $X\to A+_{f,g}B \leftarrow Z$, obtained by taking the pushout of $f$ and $g$. This data is the bicategory~\cite{benabou1967introduction}  $\Cospan{\catC}$: the objects are those of $\catC$, the arrows are cospans and 2-cells are homomorphisms. Finally, $\Cospan{\catC}$ has monoidal product given by the coproduct in $\catC$, with unit the initial object $0\in \catC$.
\end{definition}

To avoid the complications of non-associative composition, it is common to consider a \emph{category} of cospans, where isomorphic cospans are equated: let therefore $\Cospanleq{\catC}$ be the monoidal category that has isomorphism classes of cospans as arrows. Note that, when going from bicategory to category, after identifying isomorphic arrows it is usual to simply discard the 2-cells. Differently, we consider $\Cospanleq{\catC}$ to be locally preordered with $(X \to A \leftarrow Y) \leq (X \to B \leftarrow Y)$ if there exists a morphism $\alpha$ going \emph{the other way}: $\alpha \from (X \to B \leftarrow Y) \Rightarrow (X \to A \leftarrow Y)$. It is an easy exercise to verify that this (pre)ordering is well-defined and compatible with composition and monoidal product. Note that, in general, $\leq$ is a genuine preorder: i.e. it is possible that both $(X \to A \leftarrow Y) \leq (X \to B \leftarrow Y)$ and $(X \to B \leftarrow Y) \leq (X \to A \leftarrow Y)$ without the cospans being isomorphic.

\smallskip
Armed with the requisite definitions, we can be rigorous about hypergraphs with interfaces.
\begin{definition}\label{defn:disccospan}
The preorder-enriched category $\DiscCospan{\Hyp_{\Sigma}}$
is the full subcategory of $\Cospanleq{\Hyp_{\Sigma}}$ with objects the finite ordinals
and arrows (isomorphism classes of) finite hypergraphs, inheriting the preorder. We call its arrows discrete cospans.
\end{definition}
The above deserves an explanation: an ordinal $n$ can be considered as the discrete hypergraph with vertices $\{0,\dots, n-1\}$.  An arrow $n\to m$ in $\DiscCospan{\Hyp_{\Sigma}}$ is thus a cospan $n\to G \leftarrow m$ where $G$ is a hypergraph and $n\to G$ and $m \to G$ are functions to its vertices. The picture in~\eqref{eq:hypgraph} shows a discrete cospan $3 \to 1$, with dotted lines representing the two morphisms.

\subsection{Preordered cartesian bicategories}
Here we explore the algebraic structure of $\Cospanleq{\catC}$. It is closely
related to that of cartesian bicategories, yet---given the discussion above---it is more natural to consider $\Cospanleq{\catC}$ as a locally \emph{pre}ordered category. We therefore need a slight generalisation of Definition~\ref{def:cartbicat}.

\begin{definition}
A \emph{preordered cartesian bicategory} has the same structure as a cartesian bicategory (Definition~\ref{def:cartbicat}), with one difference: the ordering is not required to be a partial order, merely a preorder -- it is for this reason we separated the equational and inequational theories in Figures~\ref{fig:frobeniusBimonoid} and~\ref{fig:adjoint}. The definition of morphism is as expected.
\end{definition}

\begin{proposition}\label{prop:Cospanislocally}
$\Cospanleq{\catC}$ is a preordered cartesian bicategory. 
\end{proposition}

As a consequence, $\Cospanleq{\Hyp_\Sigma}$, and thus also $\DiscCospan{\Hyp_\Sigma}$, are preordered cartesian bicategories. The latter is of particular interest: the main result of this section, Theorem~\ref{thm:hypergraph}, states that $\DiscCospan{\Hyp_\Sigma}$ is the \emph{free} preordered cartesian bicategory on $\Sigma$, defined as follows.
\begin{definition}\label{defn:fterm}
The preordered prop $\preOrdSyntaxPROP$ has, as arrows $n \to m$, GCQ terms of sort $\sort{n}{m}$ modulo the smallest congruence generated by $=$ in Figures~\ref{fig:axsmc} and~\ref{fig:frobeniusBimonoid}. These are ordered by the smallest precongruence generated by $\leq$ in Figure~\ref{fig:adjoint}.
\end{definition}

\begin{remark}\label{rem:fterm}
 Intuitively, the ordered prop $\gcqpobc$ of Definition~\ref{def:gcqpobc} is the ``poset reduction'' of the preordered prop $\preOrdSyntaxPROP$ introduced above. We will make this formal in Section~\ref{sec:rel}. 
\end{remark}

\begin{figure}
\begin{align*}
\SynToCsp{\Bcomult} & = 1\to 1 \leftarrow 2, &
\SynToCsp{\Bmult} & = 2\to 1 \leftarrow 1, &
\SynToCsp{\idone} &= 1\to 1 \leftarrow 1 \\
\SynToCsp{\Bcounit} &= 1\to 1 \leftarrow 0, &
\SynToCsp{\Bunit} & = 0\to 1 \leftarrow 1, &
\SynToCsp{\Zeronet} & = 0 \to 0 \leftarrow 0 \\
\SynToCsp{\dsymNet} &= \!\cgr[height=0.7cm]{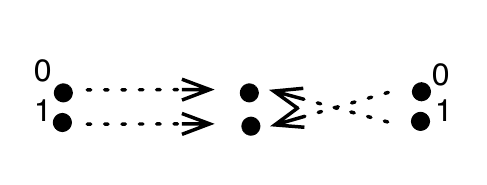}, &
\SynToCsp{c\poi d} & =\SynToCsp{c}\poi \SynToCsp{d}, & \SynToCsp{c\tns d} & =  \SynToCsp{c}\tns \SynToCsp{d} \\
\end{align*}
\vspace{-1.5cm}
\[
\SynToCsp{R} = \cgr[height=1.3cm]{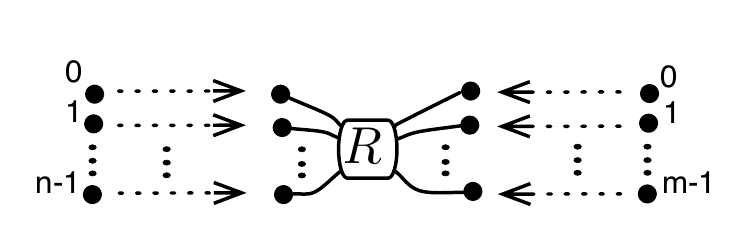}
\]
\caption{Inductive definition of the isomorphism $\SynToCsp{\cdot}\from \preOrdSyntaxPROP \to\DiscCospan{\Hyp_{\Sigma}}$. In the first two lines, the finite ordinal $n$ denotes the discrete hypergraph with $n$ vertexes, and the functions between ordinals are uniquely determined by initiality of $0$ and finality of $1$.}~\label{fig:gcqsemanticsgraphs}
\end{figure}

Theorem 3.3 in~\cite{bonchi2016rewriting} states that $\DiscCospan{\Hyp_{\Sigma}}$ and $\preOrdSyntaxPROP$ are isomorphic as mere categories, i.e.\ forgetting the preorders. We thus need only to prove that the preorder of the two categories coincides, that is for all $c,d$ in $\preOrdSyntaxPROP$,
\begin{equation}\label{eq:ordering}
c\leq d \text{ iff } \SynToCsp{c} \leq \SynToCsp{d}
\end{equation}
where $\SynToCsp{\cdot} \from \preOrdSyntaxPROP \to \DiscCospan{\Hyp_{\Sigma}}$  is the isomorphism from~\cite{bonchi2016rewriting} recalled in Figure~\ref{fig:gcqsemanticsgraphs}. The `only-if' part is immediate from Proposition~\ref{prop:Cospanislocally}.
An alternative proof consists of checking, for each of the inclusions $c \leq d$ in Figure~\ref{fig:adjoint}, that there exists a morphism of cospans from $\SynToCsp{d}$ to $\SynToCsp{c}$, as illustrated by the following example.
\begin{example}
The left and the right hand side of $(L_2)$ in Figure~\ref{fig:adjoint} for $R\in \Sigma_{1,1}$ are translated  via $\SynToCsp{\cdot}$ into the cospans on the left and right below. The morphism from the rightmost hypergraph to the leftmost one, depicted by the dashed lines, witnesses the preorder.
\begin{equation*}\label{eq:phi}
\cgr[height=1.5cm]{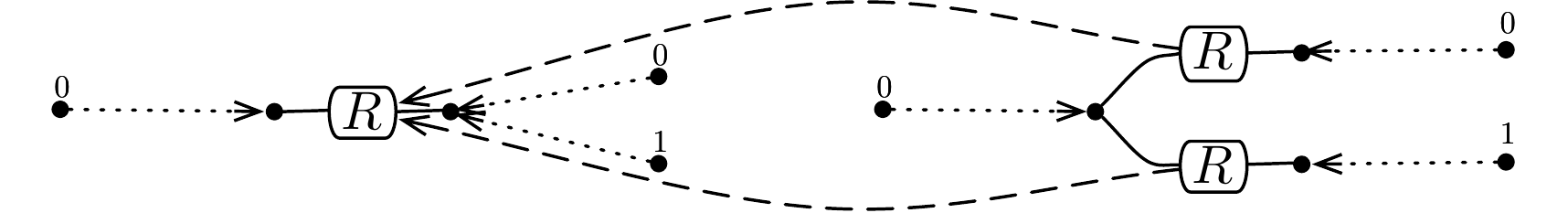}
\end{equation*}
\end{example}
The `if' part of \eqref{eq:ordering} requires some work. Its proof is given in full detail in Appendix~\ref{app:cospan}.

\begin{theorem}\label{thm:hypergraph}
  $\DiscCospan{\Hyp_{\Sigma}} \cong \preOrdSyntaxPROP$ as preordered cartesian bicategories. 
\end{theorem}

\begin{example}
Recall  Example \ref{ex:derivation}. The derivation corresponds via $\SynToCsp{\cdot}$ to the homomorphism of cospans of hypergraphs  illustrated in the Introduction. 
\end{example}


\section{Completeness for spans}\label{sec:CompSpan}

Having established a combinatorial characterisation of the free preordered cartesian bicategory, here we prove our central completeness result, Theorem~\ref{thm:completeness}. In the preordered setting, completeness holds for ``multirelational'' models: the role of the poset-enriched category $\Rel$ of sets and relations is taken by a (preorder-enriched)  bicategory of spans of functions.

\begin{definition}[$\mathsf{Span}$, $\Spanleq{}$]\label{defn:span}
  Given a finitely complete category $\catC$, the bicategory
  \begin{minipage}{0.72\textwidth}
  $\Span{\catC}$ is dual to that of cospans of Definition~\ref{defn:cospan}: it can be defined as $\Cospan{\op{\catC}}$.
More explicitly, objects are those of $\catC$, arrows of type $X \to Y$ are spans $X\leftarrow A \to Y$, composition $\poi$ is defined by pullback and $\tns$ by
categorical product. The 2-cells from
\end{minipage}
\begin{minipage}{0.28\textwidth}
  \vspace{-9pt}
\begin{equation}\label{eq:spanhom}
\raise17pt\hbox{$\xymatrix@R=5pt@C=20pt{
      & B  \ar@/^/[dr] \ar@/_/[dl]   & \\
      X & & Y  \\
      & A \ar@/_/[ur] \ar@/^/[ul]  \ar[uu]^{\alpha} & \\
}$}
\end{equation}
\end{minipage}
 $X\leftarrow A \rightarrow Y$ to $X\leftarrow B \rightarrow Y$
are span homomorphisms, that is arrows $\alpha\from A\to B$ such that the diagram on the right commutes.
As before, the bicategory $\Span{\catC}$ can be seen as a category by identifying isomorphic spans. We obtain a category
$\Spanleq{\catC}$, on which we define a preorder in a similar way to $\Cospanleq{\catC}$, but in the \emph{reverse} direction:
$(X \to A \leftarrow Y) \leq (X \to B \leftarrow Y)$ when there is a homomorphism~\eqref{eq:spanhom}.
\end{definition}

\smallskip
\begin{lemma}\label{lem:SpanleqCartBicat}
  \!\!\!$\Spanleq{\catC}$ is a preordered cartesian bicategory.
\end{lemma}

Models are now morphisms $\mathcal{M}\from \preOrdSyntaxPROP \to \Spanleq{\Set}$ of preordered cartesian bicategories. Observe that, since the interpretation of the monoid and comonoid structure is predetermined, a morphism is uniquely determined by its value on the object $1$ and on $R\in \Sigma$. In other words, a model consists of a set $\mathcal{M}(1)$ and, for each $R\in \Sigma_{n,m}$, a span $\mathcal{M}(1)^n \leftarrow Y \to \mathcal{M}(1)^m$. This data is exactly the definition of a (possibly infinite) $\Sigma$-hypergraph (Definition \ref{def:hyp}).
\begin{proposition} Morphisms $\mathcal{M}\from \preOrdSyntaxPROP \to \Spanleq{\Set}$ are in bijective correspondence with $\Sigma$-hypergraphs.
\label{pro:bijcorr}
\end{proposition}
Given this correspondence and the fact that $\preOrdSyntaxPROP \cong \DiscCospan{\Hyp_{\Sigma}}$, each hypergraph $G$ induces a morphism
$\mathcal{U}_{G} \from \DiscCospan{\Hyp_{\Sigma}} \to \Spanleq{\Set}$. 
Moreover, $G$ acts like a representing object of a contravariant Hom-functor, in the following sense: $\mathcal{U}_G$
 maps $n\xrightarrow{\iota} G' \xleftarrow{\omega} m$ to 
\[ \iHyp_{\Sigma}[n,G] \xleftarrow{\iota \poi -} \iHyp_{\Sigma}[G',G] \xrightarrow{\omega \poi -}  \iHyp_{\Sigma}[m,G] \]
where $\iHyp_{\Sigma}[G',G]$ is the set of hypergraph homomorphisms from $G'$ to $G$, and $(\iota \poi -)$ and $(\omega \poi -)$ are defined, given $f \in \iHyp_{\Sigma}[G',G]$, by 
$
(\iota \poi -)(f)=\iota \poi f \text{ and }
(\omega \poi -)(f)=\omega \poi f\text{.}
$
\begin{proposition}\label{prop:universal}
Suppose that
$n\xrightarrow{\iota} G' \xleftarrow{\omega} m$ a discrete cospan in $\DiscCospan{\Hyp_{\Sigma}}$. Then
\[
\mathcal{U}_{G}(n\stackrel{\iota}{\to} G' \stackrel{\omega}{\leftarrow} m) = 
\iHyp_{\Sigma}[n,G] \xleftarrow{\iota \poi -} \iHyp_{\Sigma}[G',G] \xrightarrow{\omega \poi -} \iHyp_{\Sigma}[m,G].\]
\end{proposition}
\begin{proof}
The conclusion of Theorem~\ref{thm:hypergraph} allows us to use induction on $n\stackrel{\iota}{\to} G' \stackrel{\omega}{\leftarrow} m$.
The inductive cases follow since the contravariant $\Hom$-functor sends colimits to limits.
Four of the base cases, $\SynToCsp{\Bcomult}$, $\SynToCsp{\Bmult}$, $\SynToCsp{\Bcounit}$ and $\SynToCsp{\Bunit}$, follow by the same argument, and the others ($\SynToCsp{\idone}$, $\SynToCsp{\dsymNet}$ and $\SynToCsp{R}$) are easy to check.
  The details are in Appendix~\ref{app:CompSpan}.
\end{proof}

\begin{theorem}[Completeness for $\Spanleq{\Set}$]\label{thm:completeness}
  Let $n \xrightarrow{\iota} G  \xleftarrow{\omega} m$ and $n \xrightarrow{\iota'} G'  \xleftarrow{\omega'} m$ be arrows in
  $\DiscCospan{\Hyp_{\Sigma}}$. If, for all morphisms $\mathcal{M} \from \DiscCospan{\Hyp_{\Sigma}} \to \Spanleq{\Set}$, we have  
  $\mathcal{M}(n \xrightarrow{\iota} G  \xleftarrow{\omega} m) 
   \leq 
  \mathcal{M}(n \xrightarrow{\iota'} G'  \xleftarrow{\omega'} m)$, then 
  $(n \xrightarrow{\iota} G  \xleftarrow{\omega} m) 
  \leq  
  (n \xrightarrow{\iota'} G'  \xleftarrow{\omega'} m)$.
\end{theorem}

\begin{proof}
  If the inequality holds for all morphisms, it holds for $ \mathcal{U}_{G}$.
  By the conclusion of Proposition~\ref{prop:universal},  there is a function $\alpha\colon \iHyp_{\Sigma}[G,G]  \to \iHyp_{\Sigma}[G',G] $ making the diagram on
 $$\xymatrix@R=5pt@C=30pt{
       &  \iHyp_{\Sigma}[G,G] \ar@/^/[dr]^(0.6){\omega \poi - } \ar@/_/[dl]_(0.6){\iota \poi - } \ar[dd]|{\alpha} & & & G' \ar[dd]|{\alpha(\id_G)} & \\
       \iHyp_{\Sigma}[n,G] & & \iHyp_{\Sigma}[m,G] & n \ar@/_/[dr]_{\iota} \ar@/^/[ur]^{\iota'} & & m \ar@/_/[ul]_{\omega'} \ar@/^/[dl]^{\omega} \\
      &  \iHyp_{\Sigma}[G',G] \ar@/_/[ur]_(0.6){\omega' \poi - } \ar@/^/[ul]^(0.6){\iota' \poi - } & & & G & \\
}
$$ 
the left commute. 
We take the identity $\id_G\in \iHyp_{\Sigma}[G,G]$ and consider $\alpha(\id_G) \colon G' \to G$. By the commutativity of the left diagram, we have that $\iota=\iota' \poi \alpha(\id_G)$ and $\omega=\omega' \poi \alpha(\id_G)$. This means that the right diagram commutes, 
that is $(n \xrightarrow{\iota}G  \xleftarrow{\omega} m ) \leq (n \xrightarrow{\iota'} G'  \xleftarrow{\omega'} m)$.
\end{proof}

\begin{remark}
The reader may have noticed that, in the above proof, $\mathcal{U}_{G}$ plays a role analogous to Chandra and Merlin's~\cite{chandra1977optimal} \emph{natural model} for the formula corresponding to $n \xrightarrow{\iota} G  \xleftarrow{\omega} m$. 
\end{remark}

Given the completeness theorem of this section, proving completeness for models of $\gcqpobc$ in $\Rel$ is a simple step that we illustrate in the next section.

\section{Completeness for relations}\label{sec:rel}
We conclude by showing how Theorem~\ref{thm:completeness} leads to a proof of Theorem~\ref{thm:main}.
The key observation lies in the tight connection between the preordered setting and the posetal one.
\begin{definition}\label{def:reduction}
  Let $\mathcal{C}$ be a preorder-enriched category. The poset-reduction of $\mathcal{C}$ is the category $\PosetCat{\mathcal{C}}$ having the same
  objects as $\mathcal{C}$ and morphisms in $\PosetCat{\mathcal{C}}$ are equivalence classes of those in $\mathcal{C}$ modulo $\sim = \leq \cap \geq$.
  Composition is inherited from $\mathcal{C}$; this is well-defined as $\sim$ is a congruence wrt composition. 
  
This assignment extends to a functor $\PosetCat{\left(\cdot\right)}$ from the category of preorder-enriched categories and functors to the category of poset-enriched ones. See Appendix \ref{app:proofsCompSpan} for details. 
\end{definition}
We have already seen, althoug implicitly, an example of this construction in passing from $\preOrdSyntaxPROP$ (Definition~\ref{defn:fterm}) to $\gcqpobc$ (Definition~\ref{def:gcqpobc}): it is indeed immediate to see that 
$\PosetCat{\left(\preOrdSyntaxPROP\right)} = \gcqpobc$. Another crucial instance is provided by the following observation, where $\Spantilde{\catC}$ is a shorthand for $\PosetCat{\left(\Spanleq{\catC}\right)}$.
\begin{proposition}\label{pro:spantilde}
  $\Spantilde{\Set} \cong \Rel$ as cartesian bicategories.
\end{proposition}

The above proposition implicitly makes use of the following fact.
\begin{proposition}
  The functor $\PosetCat{\left(\cdot\right)}$ maps preorder-enriched cartesian bicategories and morphisms into poset-enriched cartesian bicategories and morphisms.
  \label{pro:preordPosetBicat}
\end{proposition}

To conclude, it is convenient to  establish a general theory of completeness results.
\begin{definition}
  Let $\mathcal{C}, \mathcal{D}$ be preorder-enriched categories and let $\mathcal{F}$ be a class of preordered functors $\mathcal{C} \to \mathcal{D}$.
  We say that $\mathcal{C}$ is $\mathcal{F}$-complete for $\mathcal{D}$ if for all arrows $x,y$ in $\mathcal{C}$, 
  $\mathcal{M}(x) \leq \mathcal{M}(y)$ for all $\mathcal{M} \in \mathcal{F}$ entails that $x \leq y$.
  \label{def:generalCompleteness}
\end{definition}

\begin{lemma}[Transfer lemma]
  Let $\mathcal{C}, \mathcal{D}$ be preorder-enriched categories and $\mathcal{F}$ a class of preordered functors $\mathcal{C} \to \mathcal{D}$.
  Assume $\mathcal{C}$ to be $\mathcal{F}$-complete for $\mathcal{D}$.
  \begin{enumerate}
  \item  Then $\PosetCat{\mathcal{C}}$ is $\PosetCat{\mathcal{F}}$-complete for $\PosetCat{\mathcal{D}}$, where
  $\PosetCat{\mathcal{F}} = \{ \PosetCat{F} \mid F \in \mathcal{F} \}$.
  \item If $\mathcal{F} \subseteq \mathcal{F}'$, then $\mathcal{C}$ is $\mathcal{F}'$-complete for $\mathcal{D}$.
  \end{enumerate} 
  \label{lem:transfer}
\end{lemma}

All the pieces are now in place for a
\begin{proof}[Proof of Theorem~\ref{thm:main}]~
We need to show completeness---that is---assuming $c\leqq c'$, we need to prove $c \synleq c'$ for all GCQ terms $c$ and $c'$. 
Observe that $c \synleq c'$ if and only if 
\begin{equation}\tag{$\dag$}\label{eq:conc}
c \leq c'\text{ as arrows of }\gcqpobc \text{ (Definition~\ref{def:gcqpobc}).}
\end{equation}
 Moreover, using Proposition~\ref{pro:ModelsAreFunctors}, $c\leqq c'$ iff 
\begin{equation}\tag{$\ddag$}\label{eq:ass}
\mathcal{M}c \leq \mathcal{M}c'
\text{, for all morphisms of cartesian bicategories } \mathcal{M}\from \gcqpobc \to \Rel\text{.}
\end{equation}
Our task becomes, therefore, to show that~\eqref{eq:ass} implies~\eqref{eq:conc}. 
In other words, we need to prove $\gcqpobc$ to be $\mathcal{G}$-complete for $\Rel$, where $\mathcal{G}$ is the class of morphisms of cartesian bicategories of type $\from \gcqpobc \to \Rel$. Let $\mathcal{F}$ be the class of morphisms of preorder-enriched cartesian bicategories from $\preOrdSyntaxPROP$ to $\Spanleq{\Set}$.
Since, by Theorem~\ref{thm:completeness}, $\preOrdSyntaxPROP$ is $\mathcal{F}$-complete for $\Spanleq{\Set}$, we can conclude by
Lemma~\ref{lem:transfer}.1 that $\PosetCat{\left(\preOrdSyntaxPROP\right)}$ is $\PosetCat{\mathcal{F}}$-complete for $\PosetCat{\left(\Spanleq{\Set}\right)}$.
By Proposition~\ref{pro:spantilde}, this is equivalent to
$\gcqpobc$ being $\PosetCat{\mathcal{F}}$-complete for $\Rel$. Now, by Proposition~\ref{pro:preordPosetBicat} $\PosetCat{\mathcal{F}} \subseteq \mathcal{G}$, 
so the claim follows by Lemma~\ref{lem:transfer}.2.
\end{proof}

\section{Discussion, related and future work}\label{sec:discussion}
We introduced a string diagrammatic language for conjunctive queries and demonstrated a sound and complete axiomatisation for query equivalence and inclusion.
To prove completeness, we showed that our language provides an algebra able to express all hypergraphs and that our axioms characterise both hypergraph isomorphisms and existence of hypergraph morphisms. A recent result~\cite{cosme2017k4} introduced an extension of the allegorical fragment of the algebra of relations~\cite{tarski1941calculus} that is able to express all graphs with tree-width at most 2. Furthermore, the isomorphism of these graphs can be  axiomatised. The algebra in~\cite{cosme2017k4}, which is clearly less expressive than ours, can be elegantly encoded into our string diagrams (see Appendix~\ref{sec:Damien}). The same holds for the representable allegories by Freyd and Scedrov~\cite{freyd1990categories}.

We also prove completeness with respect to $\Spanleq{\Set}$, the structure of which is closely related to the \emph{bag semantics} of conjunctive queries in SQL. Indeed, the join of two SQL-tables is given by composition in $\Spanleq{\Set}$ and not in $\Rel$: in the resulting table the same row can occur several times. As we have seen, with the relational semantics, query inclusion can be decided with Chandra and Merlin's algorithm~\cite{chandra1977optimal} and its reduction to existence of a hypergraph homomorphism. On the other hand, decidability of inclusion for the bag semantic is, famously, open. Originally posed by Vardi and Chaudhuri~\cite{DBLP:conf/pods/ChaudhuriV93}, it has been studied for different fragments and extensions of conjunctive queries~\cite{DBLP:journals/tods/IoannidisR95,DBLP:journals/ipl/AfratiDG10,DBLP:conf/pods/JayramKV06}. It is worth mentioning that it is known~\cite{DBLP:journals/ejc/KoppartyR11} that there is a reduction to the homomorphism domination problem, which seems intimately related with our Proposition~\ref{prop:universal}. Unfortunately, the preorder in $\Spanleq{\Set}$---the existence of a span morphism---does \emph{not} directly correspond to bag inclusion: one must restrict to the existence of an \emph{injective} morphism. We leave this promising path for future work.

\bibliography{Bibliography}

\begin{thebibliography}{10}

\bibitem{Abramsky2008:CQM}
Samson Abramsky and Bob Coecke.
\newblock {Categorical quantum mechanics}.
\newblock {\em CoRR}, abs/1401.4973, 2008.

\bibitem{DBLP:journals/ipl/AfratiDG10}
Foto~N. Afrati, Matthew Damigos, and Manolis Gergatsoulis.
\newblock Query containment under bag and bag-set semantics.
\newblock {\em Inf. Process. Lett.}, 110(10):360--369, 2010.
\newblock URL: \url{https://doi.org/10.1016/j.ipl.2010.02.017}, \href
  {http://dx.doi.org/10.1016/j.ipl.2010.02.017}
  {\path{doi:10.1016/j.ipl.2010.02.017}}.

\bibitem{BaezErbele-CategoriesInControl}
John Baez and Jason Erbele.
\newblock Categories in control.
\newblock {\em Theory and Application of Categories}, 30:836--881, 2015.

\bibitem{benabou1967introduction}
Jean B{\'e}nabou.
\newblock Introduction to bicategories.
\newblock In {\em Reports of the Midwest Category Seminar}, pages 1--77.
  Springer, 1967.

\bibitem{bonchi2016rewriting}
Filippo Bonchi, Fabio Gadducci, Aleks Kissinger, Pawe{\l} Soboci{\'n}ski, and
  Fabio Zanasi.
\newblock Rewriting modulo symmetric monoidal structure.
\newblock In {\em Proceedings of the 31st Annual ACM/IEEE Symposium on Logic in
  Computer Science}, pages 710--719. ACM, 2016.

\bibitem{Bonchi2015}
Filippo Bonchi, Pawel Sobocinski, and Fabio Zanasi.
\newblock Full abstraction for signal flow graphs.
\newblock In {\em POPL 2015}, pages 515--526. ACM, 2015.

\bibitem{Bruni2006}
Roberto Bruni, Ivan Lanese, and Ugo Montanari.
\newblock A basic algebra of stateless connectors.
\newblock {\em Theoretical Computer Science}, 366(1--2):98--120, 2006.

\bibitem{Bruni2011}
Roberto Bruni, Hern\'{a}n~C. Melgratti, and Ugo Montanari.
\newblock A connector algebra for {P/T} nets interactions.
\newblock In {\em {CONCUR 2011}}, volume 6901 of {\em LNCS}, pages 312--326.
  Springer, 2011.
\newblock \href {http://dx.doi.org/10.1093/jigpal/6.2.349}
  {\path{doi:10.1093/jigpal/6.2.349}}.

\bibitem{Bruni2013}
Roberto Bruni, Hern\'{a}n~C. Melgratti, Ugo Montanari, and Pawe{\l}
  Soboci\'{n}ski.
\newblock Connector algebras for {C/E} and {P/T} nets' interactions.
\newblock {\em Log Meth Comput Sci}, 9(16), 2013.

\bibitem{brunihierarchical}
Roberto Bruni, Ugo Montanari, Gordon~D. Plotkin, and Daniele Terreni.
\newblock On hierarchical graphs: Reconciling bigraphs, gs-monoidal theories
  and gs-graphs.
\newblock {\em Fundam. Inform.}, 134(3-4):287--317, 2014.

\bibitem{carboni2008cartesian}
Aurelio Carboni, G~Max Kelly, Robert~FC Walters, and Richard~J Wood.
\newblock Cartesian bicategories ii.
\newblock {\em Theory and Applications of Categories}, 19(6):93--124, 2008.

\bibitem{carboni1987cartesian}
Aurelio Carboni and Robert~FC Walters.
\newblock Cartesian bicategories i.
\newblock {\em Journal of pure and applied algebra}, 49(1-2):11--32, 1987.

\bibitem{chamberlin1974sequel}
Donald~D Chamberlin and Raymond~F Boyce.
\newblock Sequel: A structured english query language.
\newblock In {\em Proceedings of the 1974 ACM SIGFIDET (now SIGMOD) workshop on
  Data description, access and control}, pages 249--264. ACM, 1974.

\bibitem{chandra1977optimal}
Ashok~K Chandra and Philip~M Merlin.
\newblock Optimal implementation of conjunctive queries in relational data
  bases.
\newblock In {\em Proceedings of the ninth annual ACM symposium on Theory of
  computing}, pages 77--90. ACM, 1977.

\bibitem{DBLP:conf/pods/ChaudhuriV93}
Surajit Chaudhuri and Moshe~Y. Vardi.
\newblock Optimization of \emph{Real} conjunctive queries.
\newblock In Catriel Beeri, editor, {\em Proceedings of the Twelfth {ACM}
  {SIGACT-SIGMOD-SIGART} Symposium on Principles of Database Systems, May
  25-28, 1993, Washington, DC, {USA}}, pages 59--70. {ACM} Press, 1993.
\newblock URL: \url{http://doi.acm.org/10.1145/153850.153856}, \href
  {http://dx.doi.org/10.1145/153850.153856} {\path{doi:10.1145/153850.153856}}.

\bibitem{codd1970relational}
Edgar~F Codd.
\newblock A relational model of data for large shared data banks.
\newblock {\em Communications of the ACM}, 13(6):377--387, 1970.

\bibitem{PQP}
B.~Coecke and A.~Kissinger.
\newblock {\em Picturing Quantum Processes. A First Course in Quantum Theory
  and Diagrammatic Reasoning}.
\newblock Cambridge University Press, 2016.

\bibitem{HandbookDPO}
A.~Corradini, U.~Montanari, F.~Rossi, H.~Ehrig, R.~Heckel, and M.~Loewe.
\newblock Algebraic approaches to graph transformation, part i: Basic concepts
  and double pushout approach.
\newblock In {\em Handbook of Graph Grammars}, pages 163--246. 1997.

\bibitem{cosme2017k4}
Enric Cosme-Ll{\'o}pez and Damien Pous.
\newblock K4-free graphs as a free algebra.
\newblock 2017.

\bibitem{diestel2017graph}
Reinhard Diestel.
\newblock {\em Graph theory}.
\newblock Springer Publishing Company, Incorporated, 2017.

\bibitem{freyd1990categories}
Peter~J Freyd and Andre Scedrov.
\newblock {\em Categories, allegories}, volume~39.
\newblock Elsevier, 1990.

\bibitem{ghica}
Dan~R. Ghica and Aliaume Lopez.
\newblock A structural and nominal syntax for diagrams.
\newblock {\em CoRR}, abs/1702.01695, 2017.
\newblock URL: \url{http://arxiv.org/abs/1702.01695}, \href
  {http://arxiv.org/abs/1702.01695} {\path{arXiv:1702.01695}}.

\bibitem{glushkov1961abstract}
Victor~Mikhaylovich Glushkov.
\newblock The abstract theory of automata.
\newblock {\em Russian Mathematical Surveys}, 16(5):1, 1961.

\bibitem{DBLP:journals/tods/IoannidisR95}
Yannis~E. Ioannidis and Raghu Ramakrishnan.
\newblock Containment of conjunctive queries: Beyond relations as sets.
\newblock {\em {ACM} Trans. Database Syst.}, 20(3):288--324, 1995.
\newblock URL: \url{http://doi.acm.org/10.1145/211414.211419}, \href
  {http://dx.doi.org/10.1145/211414.211419} {\path{doi:10.1145/211414.211419}}.

\bibitem{DBLP:conf/pods/JayramKV06}
T.~S. Jayram, Phokion~G. Kolaitis, and Erik Vee.
\newblock The containment problem for {REAL} conjunctive queries with
  inequalities.
\newblock In Stijn Vansummeren, editor, {\em Proceedings of the Twenty-Fifth
  {ACM} {SIGACT-SIGMOD-SIGART} Symposium on Principles of Database Systems,
  June 26-28, 2006, Chicago, Illinois, {USA}}, pages 80--89. {ACM}, 2006.
\newblock URL: \url{http://doi.acm.org/10.1145/1142351.1142363}, \href
  {http://dx.doi.org/10.1145/1142351.1142363}
  {\path{doi:10.1145/1142351.1142363}}.

\bibitem{LoriaCompleteness}
Emmanuel Jeandel, Simon Perdrix, and Renaud Vilmart.
\newblock A complete axiomatisation of the zx-calculus for clifford+ t quantum
  mechanics.
\newblock {\em {arXiv preprint arXiv:1705.11151}}, 2017.

\bibitem{DBLP:journals/ejc/KoppartyR11}
Swastik Kopparty and Benjamin Rossman.
\newblock The homomorphism domination exponent.
\newblock {\em Eur. J. Comb.}, 32(7):1097--1114, 2011.
\newblock URL: \url{https://doi.org/10.1016/j.ejc.2011.03.009}, \href
  {http://dx.doi.org/10.1016/j.ejc.2011.03.009}
  {\path{doi:10.1016/j.ejc.2011.03.009}}.

\bibitem{lack2004composing}
Stephen Lack.
\newblock Composing props.
\newblock {\em Theory and Applications of Categories}, 13(9):147--163, 2004.

\bibitem{mac2013categories}
Saunders Mac~Lane.
\newblock {\em Categories for the working mathematician}, volume~5.
\newblock Springer Science \& Business Media, 2013.

\bibitem{mellieszeilberger}
Paul{-}Andr{\'{e}} Melli{\`{e}}s and Noam Zeilberger.
\newblock A bifibrational reconstruction of lawvere's presheaf hyperdoctrine.
\newblock In {\em Proceedings of the 31st Annual {ACM/IEEE} Symposium on Logic
  in Computer Science, {LICS} '16, New York, NY, USA, July 5-8, 2016}, pages
  555--564, 2016.
\newblock URL: \url{http://doi.acm.org/10.1145/2933575.2934525}, \href
  {http://dx.doi.org/10.1145/2933575.2934525}
  {\path{doi:10.1145/2933575.2934525}}.

\bibitem{OxfordCompleteness}
Kang~Feng Ng and Quanlong Wang.
\newblock A universal completion of the zx-calculus.
\newblock {\em {arXiv preprint arXiv:1706.09877}}, 2017.

\bibitem{SadrzadehCC14}
Mehrnoosh Sadrzadeh, Stephen Clark, and Bob Coecke.
\newblock The frobenius anatomy of word meanings {I:} subject and object
  relative pronouns.
\newblock {\em CoRR}, abs/1404.5278, 2014.

\bibitem{sassone2005congruence}
Vladimiro Sassone and Pawe{\l} Soboci{\'n}ski.
\newblock A congruence for petri nets.
\newblock {\em Electronic Notes in Theoretical Computer Science},
  127(2):107--120, 2005.

\bibitem{tarski1941calculus}
Alfred Tarski.
\newblock On the calculus of relations.
\newblock {\em The Journal of Symbolic Logic}, 6(3):73--89, 1941.

\end{thebibliography}

\appendix

\section{An encoding of the algebra of graphs with tree-width at most 2}\label{sec:Damien}
Cosme and Pous introduced in~\cite{cosme2017k4}  an algebra that is able to express all and only the graphs with tree-width at most 2~\cite{diestel2017graph}. Its syntax is given below.
\begin{equation*}
  c ::= \top \;|\; c \wedge c  \;|\; \id \;|\; c \poi c \;|\; \op{c} \;|\; R 
\end{equation*}
Note that this extends the allegorical fragment of relational algebra in~\cite{freyd1990categories} with $\top$.
Figure~\ref{fig:Damien} shows a simple encoding of this algebra into GCQ. It is immediate to verify that the (cospans of) graphs associated to these GCQ terms via the map $\SynToCsp{\cdot}$ (Figure~\ref{fig:gcqsemanticsgraphs}) coincides with the graph semantics provided by Figure 2 in~\cite{cosme2017k4}. 
 
 \begin{figure}
$$\begin{array}{ccc}
\top & \mapsto & \lower4pt\hbox{$\includegraphics[width=2cm]{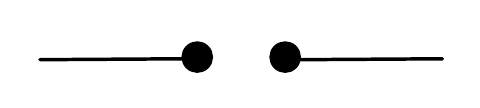}$}\\
c \wedge c& \mapsto & \lower12pt\hbox{$\includegraphics[width=2cm]{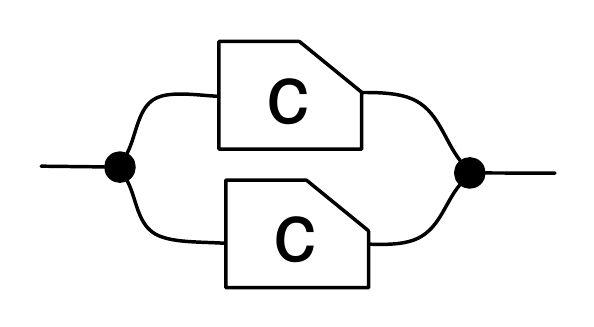}$}\\
id & \mapsto &  \lower4pt\hbox{$\includegraphics[width=2cm]{graffles/id.pdf}$}\\
c \poi c & \mapsto & \lower6pt\hbox{$\includegraphics[width=2cm]{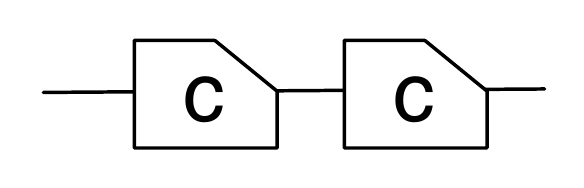}$}\\
\op{c} & \mapsto & \lower10pt\hbox{$\includegraphics[width=2cm]{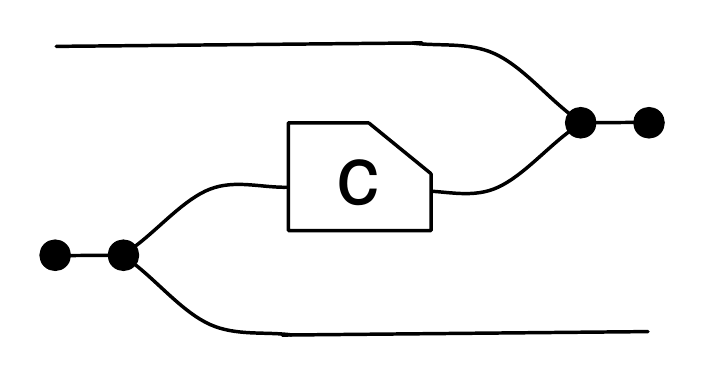}$}\\
R & \mapsto & \lower5pt\hbox{$\includegraphics[width=1.1cm]{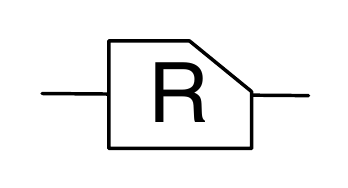}$}\\
 \end{array}$$
 \caption{An encoding of~\cite{cosme2017k4} into GCQ}\label{fig:Damien}
 \end{figure}

\section{A translation from GCQ to CCQ}\label{s:GCQtoCCQ}

\begin{figure*}
\[
\Lambda(\Bcomult) =  1,2 \vdash (x_0 = y_0) \wedge (x_0 = y_1), \qquad
\Lambda(\Bcounit) = 1,0 \vdash \top
\]
\[
\Lambda(\Bmult) = 2,1 \vdash (x_0 = y_0) \wedge (x_1 = y_0), \qquad
\Lambda(\Bunit) = 0,1 \vdash \top, 
\]
\[ 
\Lambda(\Zeronet) = 0,0 \vdash \top, \qquad 
\Lambda(\idone) = 1,1 \vdash x_0 = y_0, \qquad
\Lambda(\dsymNet) = 2,2 \vdash (x_0 = y_1) \wedge (x_1 = y_0)
\]
\medskip
\begin{prooftree}
\lower4pt\hbox{$\includegraphics[height=20pt]{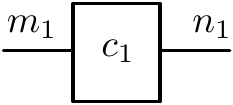}$}
\raise4pt\hbox{$\quad \longmapsto\quad  m_1,n_1 \vdash \Lambda(c_1)$}
\qquad
\lower4pt\hbox{$\includegraphics[height=20pt]{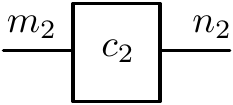}$}
\raise4pt\hbox{$\quad \longmapsto\quad  m_2,n_2 \vdash \Lambda(c_2)$}
\justifies
\lower22pt\hbox{$\includegraphics[height=45pt]{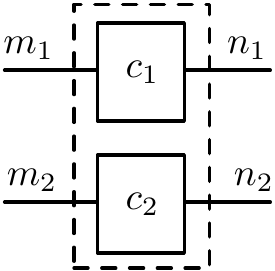}$}
\ \longmapsto\ m_1+m_2,n_1+n_2 \vdash \Lambda(c_1) \wedge (\Lambda(c_2)[x_{[m_1,m_1+m_2-1]},y_{[n_1,n_1+n_2-1]}/x_{[0,m_2-1]},y_{[0,n_2-1]}])
\using (\tns)
\end{prooftree}
\begin{prooftree}
\lower4pt\hbox{$\includegraphics[height=20pt]{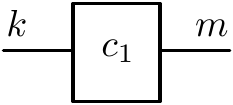}$}
\raise4pt\hbox{$\quad \longmapsto\quad  k,m \vdash \Lambda(c_1)$}
\qquad
\lower4pt\hbox{$\includegraphics[height=20pt]{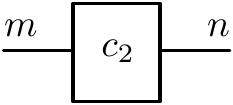}$}
\raise4pt\hbox{$\quad \longmapsto\quad  m,n \vdash \Lambda(c_2)$}
\justifies
\lower12pt\hbox{$\includegraphics[height=30pt]{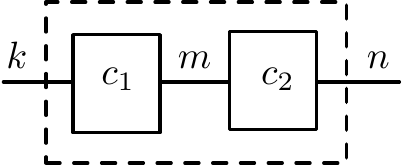}$}
\quad \longmapsto\quad k,n \vdash \exists \overset{\rightarrow}{z}.\; (\Lambda(c_1)[\overset{\rightarrow}{z}/\overset{\rightarrow}{y}]) \wedge (\Lambda(c_2)[\overset{\rightarrow}{z}/\overset{\rightarrow}{x}])
\using
(\poi)
\end{prooftree}
\caption{Translation $\Lambda$ from GCQ to CCQ.\label{fig:revtrans}}
\end{figure*}

To translate GCQ diagrams to CCQ formulas we need to introduce a minor syntactic variant of CCQ, this time assuming two countable sets of variables $Var_l = \{x_i \;|\; i\in\N \}$ and $Var_r = \{y_i \;|\; i\in\N\}$. The idea is that a diagram $c:\sort{n}{m}$ will translate to a formula that has its free variables in $\{x_0,\dots,x_{n-1}\}\cup\{y_0,\dots,y_{m-1}\}$, i.e. there are ``left'' free variables $\overset{\rightarrow}{x}$ and ``right'' free variables $\overset{\rightarrow}{y}$.

\begin{definition}
We write $n,m\vdash \phi$ if 
\[fr(\phi)\subseteq\{x_0,\dots,x_{n-1}\}\cup\{y_0,\dots,y_{m-1}\} \text{ and}\]
\[n+m \vdash \phi[x_{[n,n+m-1]}/y_{[0,m-1]}].\]
\end{definition}

Next, for $R\in\Sigma_{n,m}$ we assume a CCQ signature in which $R$ is a relation symbol with arity $n+m$.
Then, given a GCQ model $\mathcal{M}=(X,\rho)$ we can obtain a CCQ model $\Lambda(\mathcal{M})=(X,\rho')$ in the obvious way.

With the aid of the above we can give a recursive translation $\Lambda$ from GCQ terms to CCQ formulas. The details are given in Figure~\ref{fig:revtrans}. 

The following confirms that $\Lambda$ preserves semantics.
\begin{proposition}\label{thm:diagramstologic}
Fix a GCQ model $\mathcal{M}=(X,\rho)$ and suppose that $c:\sort{n}{m}$ is a GCQ formula.
Then $(\overset{\rightarrow}{v},\overset{\rightarrow}{w})\in \densem{c:\sort{n}{m}}{\mathcal{M}}$
iff $(\overset{\rightarrow}{v},\overset{\rightarrow}{w})\in \densem{n+m \vdash \Lambda(c)}{\Lambda(\mathcal{M})}$.
\end{proposition}
\begin{proof}
Induction on the derivation of $c:\sort{n}{m}$.
\end{proof}

The following is immediate from the definition of the translations $\Theta$ and $\Lambda$.
\begin{lemma}\label{lem:lambdatheta}
Suppose that $n\vdash \phi$ is a CCQ formula and $\mathcal{M}$ a CCQ model. Then 
\[
\densem{n\vdash \phi}{\mathcal{M}} = \densem{\Lambda\Theta(n \vdash \phi)}{\Lambda\Theta(\mathcal{M})}.
\qed
\]
\end{lemma}
In spite of the above, $\Theta$ and $\Lambda$ are not quite the inverse of each other. The reason is ``bureaucratic'': the image of $\Theta$ only covers diagrams of type $(n,0)$, which from the point of view of the algebra of monoidal categories is not a particularly interesting class since we miss the power of categorical composition. Similarly, the reverse translation $\Lambda$ does not seem logically natural, since, e.g. the translation of categorical composition involves both $\wedge$ and $\exists$. However, in spite of these superficial differences, Propositions~\ref{prop:CCQGCQ} and~\ref{prop:GCQCCQ} guarantee that the two languages indeed do have the same expressive power. 

\begin{example}
An interesting case is \lower7pt\hbox{$\includegraphics[height=20pt]{graffles/cup.pdf}$}. 
It is the translation, via $\Theta$, of $2\vdash x_0 = x_1$. Returning to CCQ via $\Lambda$, we obtain $2,0 \vdash \exists z.\; (x_0 = z) \wedge (x_1 = z) \wedge \top$.  The formulas are quite different---syntactically speaking---but they are logically equivalent. 
\end{example}

\begin{example}
The case of GCQ terms $\Zeronet$ and $\lower5pt\hbox{$\includegraphics[height=15pt]{graffles/boxedbone.pdf}$}$ is also interesting. The first translates via $\Lambda$ to $0,0\vdash \top$, the second to $0,0\vdash \exists z_0. \top \wedge \top$.
\end{example}

\section{Syntactic sugar}\label{app:sugar}
We write $\idoneL{n}$ as an abbreviation for a bundle of $n$ wires. Interestingly, all the basic connectives can be lifted to operate on bundles instead of single
strings. This is very useful in working with larger string diagrams, or diagrams of arbitrary size. Note however, that no additional connectives are introduced,
just recursively defined syntactic sugar. The formal definition of the components can be found in Figure~\ref{fig:SyntacticSugar}.
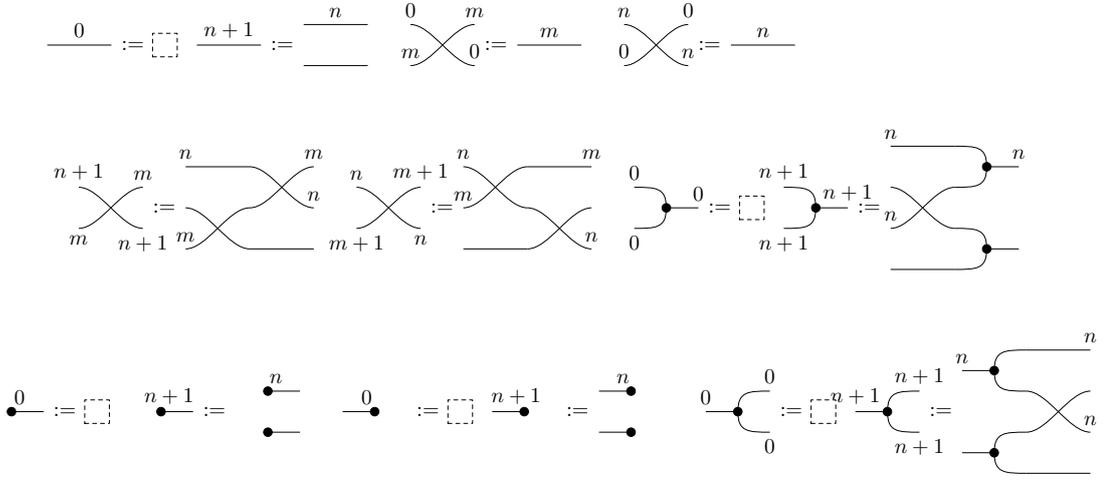
\begin{figure*}
  \centering
  \begin{subfigure}{\textwidth}
\[
\scalebox{\scalefact}{\begin{tikzpicture}
\draw (0.0,3.0) .. node [pos=0.5, anchor=south] {$0$} controls (1.0,3.0) and (2.0,3.0) .. (3.0,3.0);
\draw (7.0,3.0) .. node [pos=0.5, anchor=south] {$n+1$} controls (8.0,3.0) and (9.0,3.0) .. (10.0,3.0);
\draw (12.0,1.5) .. controls (13.0,1.5) and (14.0,1.5) .. (15.0,1.5);
\draw (12.0,4.5) .. node [pos=0.5, anchor=south] {$n$} controls (13.0,4.5) and (14.0,4.5) .. (15.0,4.5);
\draw (17.0,1.5) .. node [pos=0, anchor=south] {$m$} node [pos=1, anchor=south] {$m$} controls (18.0,1.5) and (19.0,4.5) .. (20.0,4.5);
\draw (17.0,4.5) .. node [pos=0, anchor=south] {$0$} node [pos=1, anchor=south] {$0$} controls (18.0,4.5) and (19.0,1.5) .. (20.0,1.5);
\draw (22.0,3.0) .. node [pos=0.5, anchor=south] {$m$} controls (23.0,3.0) and (24.0,3.0) .. (25.0,3.0);
\draw (27.0,1.5) .. node [pos=0, anchor=south] {$0$} node [pos=1, anchor=south] {$0$} controls (28.0,1.5) and (29.0,4.5) .. (30.0,4.5);
\draw (27.0,4.5) .. node [pos=0, anchor=south] {$n$} node [pos=1, anchor=south] {$n$} controls (28.0,4.5) and (29.0,1.5) .. (30.0,1.5);
\draw (32.0,3.0) .. node [pos=0.5, anchor=south] {$n$} controls (33.0,3.0) and (34.0,3.0) .. (35.0,3.0);
\path (3.0,3.0) -- node {$:=$} (5.0,3.0);
\path (10.0,3.0) -- node {$:=$} (12.0,3.0);
\path (20.0,3.0) -- node {$:=$} (22.0,3.0);
\path (30.0,3.0) -- node {$:=$} (32.0,3.0);
\node at (5.5, 3.0) {$\Zeronet$};
\end{tikzpicture}}
\]
  \end{subfigure}
  \begin{subfigure}{\textwidth}
\[
\scalebox{\scalefact}{\begin{tikzpicture}
\draw (0.0,4.5) .. node [pos=0, anchor=north] {$m$} node [pos=1, anchor=south] {$m$} controls (1.0,4.5) and (2.0,7.5) .. (3.0,7.5);
\draw (0.0,7.5) .. node [pos=0, anchor=south] {$n+1$} node [pos=1, anchor=north] {$n+1$} controls (1.0,7.5) and (2.0,4.5) .. (3.0,4.5);
\draw (5.0,3.0) .. node [pos=0, anchor=south] {$m$} controls (6.0,3.0) and (7.0,6.0) .. (8.0,6.0);
\draw (5.0,6.0) .. controls (6.0,6.0) and (7.0,3.0) .. (8.0,3.0);
\draw (5.0,9.0) .. node [pos=0, anchor=south] {$n$} controls (6.0,9.0) and (7.0,9.0) .. (8.0,9.0);
\draw (8.0,3.0) .. controls (9.0,3.0) and (10.0,3.0) .. (11.0,3.0);
\draw (8.0,6.0) .. node [pos=1, anchor=south] {$m$} controls (9.0,6.0) and (10.0,9.0) .. (11.0,9.0);
\draw (8.0,9.0) .. node [pos=1, anchor=south] {$n$} controls (9.0,9.0) and (10.0,6.0) .. (11.0,6.0);

\draw (13.0,4.5) .. node [pos=0, anchor=north] {$m+1$} node [pos=1, anchor=south] {$m+1$} controls (14.0,4.5) and (15.0,7.5) .. (16.0,7.5);
\draw (13.0,7.5) .. node [pos=0, anchor=south] {$n$} node [pos=1, anchor=north] {$n$} controls (14.0,7.5) and (15.0,4.5) .. (16.0,4.5);
\draw (18.0,3.0) .. controls (19.0,3.0) and (20.0,3.0) .. (21.0,3.0);
\draw (18.0,6.0) .. node [pos=0, anchor=south] {$m$} controls (19.0,6.0) and (20.0,9.0) .. (21.0,9.0);
\draw (18.0,9.0) .. node [pos=0, anchor=south] {$n$} controls (19.0,9.0) and (20.0,6.0) .. (21.0,6.0);
\draw (21.0,3.0) .. controls (22.0,3.0) and (23.0,6.0) .. (24.0,6.0);
\draw (21.0,6.0) .. node [pos=1, anchor=south] {$n$} controls (22.0,6.0) and (23.0,3.0) .. (24.0,3.0);
\draw (21.0,9.0) .. node [pos=1, anchor=south] {$m$} controls (22.0,9.0) and (23.0,9.0) .. (24.0,9.0);

\draw (27.5,6.0) .. node [pos=1, anchor=north] {$0$} controls (27.5,4.5) and (27.0,4.5) .. (26.0,4.5);
\draw (27.5,6.0) .. node [pos=1, anchor=south] {$0$} controls (27.5,7.5) and (27.0,7.5) .. (26.0,7.5);
\draw (27.5,6.0) .. node [pos=1, anchor=south] {$0$} controls (27.5,6.0) and (28.0,6.0) .. (29.0,6.0);
\draw (27.5,6.0) \blk {};

\draw (34.5,6.0) .. node [pos=1, anchor=north] {$n+1$} controls (34.5,4.5) and (34.0,4.5) .. (33.0,4.5);
\draw (34.5,6.0) .. node [pos=1, anchor=south] {$n+1$} controls (34.5,7.5) and (34.0,7.5) .. (33.0,7.5);
\draw (34.5,6.0) .. node [pos=1, anchor=south] {$n+1$} controls (34.5,6.0) and (35.0,6.0) .. (36.0,6.0);
\draw (34.5,6.0) \blk {};
\draw (38.0,1.5) .. controls (39.0,1.5) and (40.0,1.5) .. (41.0,1.5);
\draw (38.0,4.5) .. node [pos=0, anchor=south] {$n$} controls (39.0,4.5) and (40.0,7.5) .. (41.0,7.5);
\draw (38.0,7.5) .. controls (39.0,7.5) and (40.0,4.5) .. (41.0,4.5);
\draw (38.0,10.5) .. node [pos=0, anchor=south] {$n$} controls (39.0,10.5) and (40.0,10.5) .. (41.0,10.5);
\draw (42.5,3.0) .. controls (42.5,1.5) and (42.0,1.5) .. (41.0,1.5);
\draw (42.5,3.0) .. controls (42.5,4.5) and (42.0,4.5) .. (41.0,4.5);
\draw (42.5,3.0) .. controls (42.5,3.0) and (43.0,3.0) .. (44.0,3.0);
\draw (42.5,3.0) \blk {};
\draw (42.5,9.0) .. controls (42.5,7.5) and (42.0,7.5) .. (41.0,7.5);
\draw (42.5,9.0) .. controls (42.5,10.5) and (42.0,10.5) .. (41.0,10.5);
\draw (42.5,9.0) .. node [pos=1, anchor=south] {$n$} controls (42.5,9.0) and (43.0,9.0) .. (44.0,9.0);
\draw (42.5,9.0) \blk {};
\path (3.0,6.0) -- node {$:=$} (5.0,6.0);
\path (16.0,6.0) -- node {$:=$} (18.0,6.0);
\path (29.0,6.0) -- node {$:=$} (31.0,6.0);
\node at (31.5,6.0) {$\Zeronet$};
\path (36.0,6.0) -- node {$:=$} (38.0,6.0);
\end{tikzpicture}}
\]
  \end{subfigure}
  \begin{subfigure}{\textwidth}
\[
\scalebox{\scalefact}{\begin{tikzpicture}
\draw (1.5,6.0) .. node [pos=0.5, anchor=south] {$0$} controls (1.5,6.0) and (2.0,6.0) .. (3.0,6.0);
\draw (1.5,6.0) \blk {};

\draw (8.5,6.0) .. node [pos=0.5, anchor=south] {$n+1$} controls (8.5,6.0) and (9.0,6.0) .. (10.0,6.0);
\draw (8.5,6.0) \blk {};

\draw (13.5,4.5) .. controls (13.5,4.5) and (14.0,4.5) .. (15.0,4.5);
\draw (13.5,4.5) \blk {};
\draw (13.5,7.5) .. node [pos=0.5, anchor=south] {$n$} controls (13.5,7.5) and (14.0,7.5) .. (15.0,7.5);
\draw (13.5,7.5) \blk {};

\draw (18.5,6.0) .. node [pos=0.5, anchor=south] {$0$} controls (18.5,6.0) and (18.0,6.0) .. (17.0,6.0);
\draw (18.5,6.0) \blk {};

\draw (25.5,6.0) .. node [pos=0.5, anchor=south] {$n+1$} controls (25.5,6.0) and (25.0,6.0) .. (24.0,6.0);
\draw (25.5,6.0) \blk {};

\draw (30.5,4.5) .. controls (30.5,4.5) and (30.0,4.5) .. (29.0,4.5);
\draw (30.5,4.5) \blk {};
\draw (30.5,7.5) .. node [pos=0.5, anchor=south] {$n$} controls (30.5,7.5) and (30.0,7.5) .. (29.0,7.5);
\draw (30.5,7.5) \blk {};

\draw (35.5,6.0) .. node [pos=1, anchor=south] {$0$} controls (35.5,6.0) and (35.0,6.0) .. (34.0,6.0);
\draw (35.5,6.0) .. node [pos=1, anchor=north] {$0$} controls (35.5,4.5) and (36.0,4.5) .. (37.0,4.5);
\draw (35.5,6.0) .. node [pos=1, anchor=south] {$0$} controls (35.5,7.5) and (36.0,7.5) .. (37.0,7.5);
\draw (35.5,6.0) \blk {};

\draw (42.5,6.0) .. node [pos=1, anchor=south] {$n+1$} controls (42.5,6.0) and (42.0,6.0) .. (41.0,6.0);
\draw (42.5,6.0) .. node [pos=1, anchor=north] {$n+1$} controls (42.5,4.5) and (43.0,4.5) .. (44.0,4.5);
\draw (42.5,6.0) .. node [pos=1, anchor=south] {$n+1$} controls (42.5,7.5) and (43.0,7.5) .. (44.0,7.5);
\draw (42.5,6.0) \blk {};
\draw (47.5,3.0) .. controls (47.5,3.0) and (47.0,3.0) .. (46.0,3.0);
\draw (47.5,3.0) .. controls (47.5,1.5) and (48.0,1.5) .. (49.0,1.5);
\draw (47.5,3.0) .. controls (47.5,4.5) and (48.0,4.5) .. (49.0,4.5);
\draw (47.5,3.0) \blk {};
\draw (47.5,9.0) .. node [pos=1, anchor=south] {$n$} controls (47.5,9.0) and (47.0,9.0) .. (46.0,9.0);
\draw (47.5,9.0) .. controls (47.5,7.5) and (48.0,7.5) .. (49.0,7.5);
\draw (47.5,9.0) .. controls (47.5,10.5) and (48.0,10.5) .. (49.0,10.5);
\draw (47.5,9.0) \blk {};
\draw (49.0,1.5) .. controls (50.0,1.5) and (51.0,1.5) .. (52.0,1.5);
\draw (49.0,4.5) .. controls (50.0,4.5) and (51.0,7.5) .. (52.0,7.5);
\draw (49.0,7.5) .. node [pos=1, anchor=south] {$n$} controls (50.0,7.5) and (51.0,4.5) .. (52.0,4.5);
\draw (49.0,10.5) .. node [pos=1, anchor=south] {$n$} controls (50.0,10.5) and (51.0,10.5) .. (52.0,10.5);

\path (3.0,6.0) -- node {$:=$} (5.0,6.0);
\node at (5.5,6.0) {$\Zeronet$};
\path (10.0,6.0) -- node {$:=$} (12.0,6.0);
\path (20.0,6.0) -- node {$:=$} (22.0,6.0);
\node at (22.5,6.0) {$\Zeronet$};
\path (27.0,6.0) -- node {$:=$} (29.0,6.0);
\path (37.0,6.0) -- node {$:=$} (39.0,6.0);
\node at (39.5,6.0) {$\Zeronet$};
\path (44.0,6.0) -- node {$:=$} (46.0,6.0);
\end{tikzpicture}}
\]
  \end{subfigure}
  \caption{Syntactic sugar for labelled strings}
  \label{fig:SyntacticSugar}
\end{figure*}

\section{Proofs}
\subsection{Proofs of Section~\ref{sec:CCQ}}

\begin{proof}[Proof of Proposition~\ref{pro:syntax}]
  The `only if' direction is a straightforward induction on the derivation of a formula generated by~\eqref{eq:calculus}. The `if' is a trivial induction on derivations obtained from $\{(\top),(R),(\exists),(=),(\wedge),(\mathsf{Sw}_{n,k}),(\mathsf{Id}_n),(\mathsf{Nu}_n)\}$.
\end{proof}

\subsection{Proofs of Section~\ref{sec:GCQ}}
\begin{proof}[Proof of Proposition~\ref{thm:logictodiagrams}]
Easy induction on the derivation of $n\vdash \phi$.
\end{proof}

\begin{proof}[Proof of Proposition~\ref{prop:CCQGCQ}]
Immediate from Proposition~\ref{thm:logictodiagrams} in \S\ref{sec:expressivity} and Proposition~\ref{thm:diagramstologic} and Lemma~\ref{lem:lambdatheta} in Appendix~\ref{s:GCQtoCCQ}.
\end{proof}

\begin{proof}[Proof of Proposition~\ref{prop:GCQCCQ}]
  The translation is given in Figure~\ref{fig:revtrans}. The claim then follows by Proposition~\ref{thm:logictodiagrams} in \S\ref{sec:expressivity} and Proposition~\ref{thm:diagramstologic} and Lemma~\ref{lem:lambdatheta} in Appendix~\ref{s:GCQtoCCQ}.
\end{proof}

\subsection{Proofs of Section~\ref{sec:todiagram}}

\begin{proof}[Proof of Lemma~\ref{lem:precongruence}]
  Follows immediately from the definition of semantics and relational composition / tensor in Figure~\ref{fig:gcqsemantics}.
\end{proof}

\begin{proof}[Proof of Proposition~\ref{pro:smc}]
  For each fixed model the axioms of Figure~\ref{fig:axsmc} are satisfied because the category of relations with monoidal product $\times$ is symmetric monoidal.
\end{proof}

\subsection{Proofs of Section~\ref{sec:axioms}}
\begin{proof}[Proof of Lemma~\ref{lem:gcqpobcbicat}]
For every object $n$, the monoid and comonoid are given by the syntactic sugar $\BcomultL{n}$, $\BcounitL{n}$, $\BmultL{n}$ and $\BunitL{n}$  (Appendix~\ref{app:sugar}). That these satisfy the required laws is easily proven by induction on $n$.
   Note that the definition of $\gcqpobc$ asserts that for every $R \in \Sigma_{n,m}$, 
  $\Rcircuit$ is a lax comonoid morphism in $\gcqpobc$, but the definition of cartesian bicategory requires this for 
  \emph{all} arrows. This can be easily derived by another induction, though.
  \label{prf:gcqpobcbicat}
\end{proof}

\begin{proof}[Proof of Lemma~\ref{lem:ModelsAreFunctors}]
  In the easy direction, to extract a model $\mathcal{M} = (X, \rho)$ from a morphism of cartesian bicategories $\mathcal{F} \from \gcqpobc \to \Rel$,
  define $X := \mathcal{F}(1)$ and let $\rho(R) := \mathcal{F}(R)$ for $R \in \Sigma$,

 In the reverse, given a model $\mathcal{M} = (X, \rho)$, we observe that the semantics map $\densem{\cdot}{\mathcal{M}}$ (Figure~\ref{fig:gcqsemantics}) gives rise to a morphism of cartesian bicategories $\densem{\cdot}{\mathcal{M}} \from \gcqpobc \to \Rel$. To prove that it is well defined and preserves the ordering, one can easily see that the axioms of $\syneq$ and $\synleq$ are sound. By the inductive definition, $\densem{\cdot}{\mathcal{M}}$  preserves composition $\poi$ and tensor $\tns$. Finally, we observe that, by definition, $\densem{\cdot}{\mathcal{M}}$ maps the monoids and comonoids of  $ \gcqpobc$ into those of $\Rel$.
 \label{prf:ModelsAreFunctors}
\end{proof}

\subsection{Proofs of Section~\ref{sec:cospans}}\label{app:cospan}

We first show a proof of Proposition~\ref{prop:Cospanislocally} and then we provide some intermediate results that are necessary for proving Theorem~\ref{thm:hypergraph}. 

\begin{proof}[Proof of Proposition~\ref{prop:Cospanislocally}]
  We need to endow every object with a monoid and comonoid structure, prove these structures to be adjoint and furthermore satisfy the special Frobenius property.
  \begin{enumerate}
    \item Define the monoid/comonoid structure on every object:
  Let $X \in \mathcal{C}$ and define $\mu \from X + X \to X$, $\mu = \langle \id , \id \rangle$ and $\eta \from 0 \to X$ the unique such morphism,
  where $0$ is the initial object.
  Now it is standard that $\mu$ and $\eta$ turn $X$ into a monoid in $\mathcal{C}$. 
  Since there is a monoidal functor $F \from \mathcal{C} \to \Cospanleq{\mathcal{C}}$ sending $f \from X \to Y$ to the cospan
  \[
    \begin{tikzcd}[cramped, sep=small]
      X \arrow{r}{f} & Y & Y \arrow[swap]{l}{\id}
    \end{tikzcd}
  \]
  we see that $F(\mu), F(\eta)$ define a monoid structure on $X$ in $\Cospanleq{\mathcal{C}}$.
  Furthermore, there is a duality operation $\op{\bullet}$ on $\Cospanleq{\mathcal{C}}$ given by turning around the cospan, i.e. mapping
  \[
    \begin{tikzcd}[cramped, sep=small]
      X \arrow{r}{f} & Z & Y \arrow[swap]{l}{g}
    \end{tikzcd}
  \]
  to
  \[
    \begin{tikzcd}[cramped, sep=small]
      Y \arrow{r}{g} & Z & X \arrow[swap]{l}{f}
    \end{tikzcd}
  \]
  Now define the comonoid structure on every object via $\op{F(\mu)}$ and $\op{F(\eta)}$.
  It is easy to see that every morphism in $\Cospanleq{\mathcal{C}}$ is a lax comonoid homomorphism,
  which follows from the fact that every morphism in $\mathcal{C}$ preserves the monoid structure $\mu$, $\eta$.
\item The monoid and comonoid structures are adjoint:
  We prove in general that for $f \from X \to Y$ a morphism in $\mathcal{C}$, $F(f)$ is right-adjoint to $\op{F(f)}$.
  This will in particular imply the adjointness between the comonoid and the monoid structure.
  Let
  \[
    \begin{tikzcd}[cramped, sep=small]
    & P & \\
    Y \arrow{ur}{i} & & \arrow[swap]{ul}{j} Y \\
    & X \arrow{ul}{f} \arrow[swap]{ur}{f} & \\
    \end{tikzcd}
  \]
  be a pushout, then computing $\op{F(f)} \poi F(f)$ yields
  \[
    \begin{tikzcd}[cramped, sep=small]
    & & P & & \\
    & Y \arrow{ur}{i} & & \arrow[swap]{ul}{j} Y & \\
    Y \arrow{ur}{\id_{Y}} & & \arrow{ul}{f} X \arrow[swap]{ur}{f} & & \arrow[swap]{ul}{\id_{Y}} Y \\
    \end{tikzcd}
  \]
  By the universal property of the pushout, there exists a morphism $\eta \from P \to Y$ such that
  \[
    \begin{tikzcd}[cramped, sep=small]
    & P \arrow{dd}{\eta} & \\
    Y \arrow[swap]{dr}{\id_{Y}} \arrow{ur}{i} & & \arrow{dl}{\id_{Y}} \arrow[swap]{ul}{j} Y \\
    & Y & \\
    \end{tikzcd}
  \]
  commutes, therefore $\id_{Y} \leq \op{F(f)} \poi F(f)$.
  Computing $F(f) \poi \op{F(f)}$ on the other hand yields:
  \[
    \begin{tikzcd}[cramped, sep=small]
    & & Y & & \\
    & Y \arrow{ur}{\id_{Y}} & & \arrow[swap]{ul}{\id_{Y}} Y & \\
    X \arrow{ur}{f} & & \arrow{ul}{\id_{Y}} Y \arrow[swap]{ur}{\id_{Y}} & & \arrow[swap]{ul}{f} X \\
    \end{tikzcd}
  \]
  Now clearly the diagram
  \[
    \begin{tikzcd}[cramped, sep=small]
    & X \arrow{dd}{f} & \\
    X \arrow{ur}{\id_{X}} \arrow[swap]{dr}{f} & & \arrow[swap]{ul}{\id_{X}} \arrow{dl}{f} Y \\
    & Y & \\
    \end{tikzcd}
  \]
  commutes, hence $F(f) \poi \op{F(f)} \leq \id_{X}$.
\item The monoid and comonoid enjoy the special Frobenius property:
  For this it suffices to see that
  \[
    \begin{tikzcd}[cramped, sep=small]
      & X & \\
      X + X \arrow{ur}{\mu} & & \arrow[swap]{ul}{\mu} X + X \\
      & \arrow{ul}{\id_X + \mu} X + X + X \arrow[swap]{ur}{\mu + \id_X} \\
    \end{tikzcd}
  \]
  is a pushout. To prove speciality of this Frobenius structure, it suffices to see that the diagram
  \[
    \begin{tikzcd}[cramped, sep=small]
      & X & \\
      X \arrow{ur}{\id_X} & & X \arrow[swap]{ul}{\id_X} \\
      & X + X \arrow{ul}{\mu} \arrow[swap]{ur}{\mu} & \\
    \end{tikzcd}
  \]
  is a pushout as well.
  \end{enumerate}
\end{proof}

We will prove in Theorem~\ref{thm:hypergraph} that $\DiscCospan{\Hyp_\Sigma} \cong \preOrdSyntaxPROP$.
It is convenient to begin with $\Sigma=\varnothing$.
Consider the category $\Fin$: objects are finite ordinals $n=\{0,\dots, n-1\}$ and arrows all functions. Then $\Cospanleq{\Fin}$ is the free preordered cartesian bicategory on the empty signature.

\begin{theorem}\label{thm:emptysignature}
  $\Cospanleq{\Fin} \cong \preOrdSyntaxPROPEmpty$ as preordered cartesian bicategories.
\end{theorem}
\begin{proof}
  The translation in Figure~\ref{fig:gcqsemanticsgraphs} defines an isomorphism $\SynToCsp{\cdot}\from \preOrdSyntaxPROPEmpty \to \Cospanleq{\Fin}$ (first three lines).
  The translation $\CspToSyn{\cdot} \from \Cospanleq{\Fin} \to \preOrdSyntaxPROPEmpty$ can be found in~\cite[Theorem 3.3]{bonchi2016rewriting},
  where the authors also prove that it defines an isomorphism between categories, i.e. if one forgets about the ordering.
  It thus suffices to prove, that both translations preserve the ordering.
  For $c,d \in \preOrdSyntaxPROPEmpty$, we have
  \[
    c \leq d \text{ implies } \SynToCsp{c} \leq \SynToCsp{d}
  \]
  by Proposition~\ref{prop:Cospanislocally}. Consider a morphism of cospans
  \[
    \begin{tikzcd}[cramped, sep=small]
      & S \arrow{dd}{\alpha} & \\
      n \arrow{ur} \arrow{dr} & & m \arrow{dl} \arrow{ul} \\
      & T & \\
    \end{tikzcd}
  \]
  We want to prove
  \[
    \CspToSyn{ n \rightarrow T \leftarrow m} \leq \CspToSyn{ n \rightarrow S \leftarrow m}
  \]
Since every function $\alpha \from S \to T$ can be decomposed into sums and compositions of $2 \rightarrow 1$ and $0 \rightarrow 1$ as demonstrated for example
in \cite[VII.5]{mac2013categories}, we can consider only these cases.
In the case $\alpha \from 0 \to 1$, we have $n = m = 0$ and we have to prove
\[
  \raisebox{-2mm}{
\begin{tikzpicture}
\draw (0.25,1.0) .. controls (0.25,1.0) and (0.0,1.0) .. (0.5,1.0);
\draw (0.25,1.0) \blk {};
\draw (0.5,1.0) .. controls (1.0,1.0) and (1.5,1.0) .. (2.0,1.0);
\draw (2.25,1.0) .. controls (2.25,1.0) and (2.5,1.0) .. (2.0,1.0);
\draw (2.25,1.0) \blk {};
\end{tikzpicture}}
\leq \Zeronet
\]
which is axiom $(UC)$.
The case $\alpha \from 2 \to 1$, can be further reduced by the following observation:
Given a diagram 
  \[
    \begin{tikzcd}[cramped, sep=small]
      & 2 \arrow{dd} & \\
      n_1 + n_2 \arrow{ur} \arrow{dr} & & m_1 + m_2 \arrow{dl} \arrow{ul} \\
      & 1 & \\
    \end{tikzcd}
  \]
one easily computes the composite of spans
\[
  \begin{tikzcd}[cramped, sep=small]
    & 2 & & 2 & & 2 & \\
    n_1 + n_2 \arrow{ur} & & 2 \arrow[swap]{ul}{\id} \arrow{ur}{\id} & & 2 \arrow[swap]{ul}{\id} \arrow{ur}{\id} & & \arrow[swap]{ul} m_1 + m_2 \\
  \end{tikzcd}
\]
to be $n_1 + n_2 \rightarrow 2 \leftarrow m_1 + m_2$
and
\[
  \begin{tikzcd}[cramped, sep=small]
    n_1 + n_2 \arrow[swap]{dr} & & 2 \arrow{dl}{\id} \arrow[swap]{dr} & & 2 \arrow{dl} \arrow[swap]{dr}{\id} & & \arrow{dl} m_1 + m_2 \\
    & 2 & & 1 & & 2 & \\
  \end{tikzcd}
\]
to be $n_1 + n_2 \rightarrow 1 \leftarrow m_1 + m_2$
By compositionality, it suffices to consider the case
  \[
    \begin{tikzcd}[cramped, sep=small]
      & 2 \arrow{dd} & \\
      2 \arrow{ur} \arrow{dr} & & 2 \arrow{dl} \arrow{ul} \\
      & 1 & \\
    \end{tikzcd}
  \]
  which corresponds to
  \[
  \raisebox{-4mm}{
\begin{tikzpicture}
\draw (0.75,2.0) .. controls (0.75,1.0) and (0.5,1.0) .. (0.0,1.0);
\draw (0.75,2.0) .. controls (0.75,3.0) and (0.5,3.0) .. (0.0,3.0);
\draw (0.75,2.0) .. controls (0.75,2.0) and (1.0,2.0) .. (1.5,2.0);
\draw (0.75,2.0) \blk {};
\draw (2.25,2.0) .. controls (2.25,2.0) and (2.0,2.0) .. (1.5,2.0);
\draw (2.25,2.0) .. controls (2.25,1.0) and (2.5,1.0) .. (3.0,1.0);
\draw (2.25,2.0) .. controls (2.25,3.0) and (2.5,3.0) .. (3.0,3.0);
\draw (2.25,2.0) \blk {};
\end{tikzpicture}} \leq
  \raisebox{-4mm}{
\begin{tikzpicture}
\draw (0.0,1.0) .. controls (0.5,1.0) and (1.0,1.0) .. (1.5,1.0);
\draw (0.0,3.0) .. controls (0.5,3.0) and (1.0,3.0) .. (1.5,3.0);
\end{tikzpicture}}
  \]
  which is axiom $(MC)$.
\end{proof}

As explained in Section~\ref{sec:cospans}, by virtue of~\cite[Theorem 3.3]{bonchi2016rewriting}, we only need to prove \eqref{eq:ordering}.
The `only-if' part is immediate from Proposition~\ref{prop:Cospanislocally}.

For the `if' part of \eqref{eq:ordering}, we first derive a technical result for disconnected cospans. A cospan of hypergraphs is said to be \emph{disconnected} if it is of the form $\SynToCsp{R_0} \tns \SynToCsp{R_1} \tns \dots \SynToCsp{R_n}$ for arbitrary $R_0, \dots R_n \in \Sigma$. 

\begin{lemma}\label{lemma:disconnected}
Let $n \stackrel{\iota}{\to} E \stackrel{\omega}{\leftarrow} m$ and $n' \stackrel{\iota'}{\to} E' \stackrel{\omega'}{\leftarrow} m'$ be disconnected cospans.
If there are functions $f\from n \to n'$, $g\from m \to m'$ and $h\from E \to E'$ s.t.\ the following commutes
\begin{equation}\label{eq:disconnected}
\raise10pt\hbox{$
\xymatrix@R=10pt@C=40pt{
n \ar[r]^-{\iota} \ar[d]_{f} & E \ar[d]_{h} &  m \ar[d]^{g} \ar[l]_-{\omega} \\
n' \ar[r]_-{\iota'} & E' & \ar[l]^-{\omega'} m'
}$}
\end{equation}
then $\CspToSyn{n \xrightarrow{f\poi \iota'} E' \xleftarrow{g\poi \omega'} m} \leq \CspToSyn{n \xrightarrow{\iota} E \xleftarrow{\omega} m}$.
\end{lemma}
\begin{proof}
  First note that in the case of disconnected cospans, $h$ uniquely determines $f$ and $g$.
  Furthermore, to give a hypergraph homomorphism $h \from E \to E'$ is the same as giving a label-preserving function between their sets of hyperedges.
  We will therefore identify $E$ and $E'$ with their sets of hyperedges.
  It is now sufficient to consider the case where $E'$ is a singleton. Indeed, in general, $E'$ as a finite set is a coproduct $E'\cong 1+1+\dots+1$ (with possibly different labels). Let $n'_k \overset{\iota'_k}{\ra} 1 \overset{\omega'_k}{\la} m'_k$ be the $k$th hyperedge with its interface. Let $i_k$ denote the $k$th injection, then pulling back~\eqref{eq:disconnected} along $i_k$ yields
\[
\xymatrix@R=10pt@C=30pt{
n_k \ar[r]^-{\iota_k} \ar[d]_{f_k}& E_k \ar[d]_{h_k} & \ar[l]_-{\omega_k} m_k \ar[d]^{g_k} \\
n'_k \ar[r]_-{\iota'_k} & 1 & \ar[l]^-{\omega'_k} m'_k
}
\]
Since coproducts are stable under pullback in the category of sets, we have that $n\ra E \la m$ is, as a cospan, isomorphic
to $\bigoplus_k (n_k \ra E_k \la m_k)$. Establishing that, for each $k$, 
$\CspToSyn{n_k \xrightarrow{f_k\poi \iota'_k} 1 \xleftarrow{g_k\poi \omega'_k} m_k} \leq \CspToSyn{n_k \xrightarrow{\iota_k} E_k \xleftarrow{\omega_k} m_k}$ thus suffices to conclude the statement of the lemma:
\begin{align*}
\CspToSyn{n \xrightarrow{f\poi \iota'} E' \xleftarrow{g\poi \omega'} m}  
&\cong \CspToSyn{ \bigoplus_k n_k \xrightarrow{f_k\poi \iota'_k} 1 \xleftarrow{g_k\poi \omega'_k} m_k} \\
&\cong \bigoplus_k \CspToSyn{  n_k \xrightarrow{f_k\poi \iota'_k} 1 \xleftarrow{g_k\poi \omega'_k} m_k} \\
&\leq \bigoplus_k \CspToSyn{n_k \xrightarrow{\iota_k} E_k \xleftarrow{\omega_k} m_k} \\
&\cong \CspToSyn{\bigoplus_k n_k \xrightarrow{\iota_k} E_k \xleftarrow{\omega_k} m_k} \\
&\cong \CspToSyn{n \xrightarrow{\iota} E \xleftarrow{\omega} m}.
\end{align*}

We thus let $n' \to E' \leftarrow m'$ consist of a single hyperedge with label $R\in \Sigma_{i,j}$, yielding 
$$
\CspToSyn{n' \xrightarrow{\iota'} E' \xleftarrow{\omega'} m'}=  \lower12pt\hbox{$\includegraphics[height=1cm]{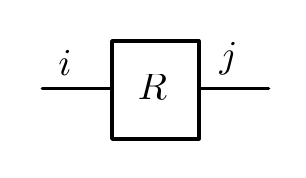}$}$$ 
and thus $n'=i$ and $m'=j$.
Since $h$ is label-preserving, every hyperedge in $E$ has to be labelled with $R$ as well, so we can forget about labels now.

Turning to $E$, we note that it suffices to consider cases where the size $|E|$ of $E$ is either 
 $0$ or $2$, yielding diagrams
\begin{equation}\label{eq:caseone}
\raise10pt\hbox{$
\xymatrix@R=10pt@C=30pt{
  0 \ar[r] \ar[d]_{\text{!`}}& 0 \ar[d]_{\text{!`}} & \ar[l] 0 \ar[d]^{\text{!`}} \\
{i} \ar[r]_-{\iota'} & 1 & \ar[l]^-{\omega'} {j}
}$}, \text{ and}
\end{equation}
\begin{equation}\label{eq:casetwo}
\raise10pt\hbox{$
\xymatrix@R=10pt@C=30pt{
i+i \ar[r]^-{\iota' + \iota'} \ar[d]_{\nabla}& 2 \ar[d]_{!} & \ar[l]_-{\omega' + \omega'} j+j \ar[d]^{\nabla} \\
{i} \ar[r]_-{\iota'} & 1 & \ar[l]^-{\omega'} {j}
}$}.
\end{equation}
Indeed, the result for $|E| \geq 2$ can be obtained inductively like this:
For $|E| = 2$ this is~\eqref{eq:casetwo}, for $|E| = n + 2 > 2$, consider the diagram:
\begin{equation}
\raise10pt\hbox{$
\xymatrix@R=10pt@C=50pt{
  n \times i + i+i \ar[r]^-{n \times \iota' + \iota' + \iota'} \ar[d]_{\id_{n \times i} + \nabla}& n + 2 \ar[d]_{\id_{n} + !} &
  \ar[l]_-{n \times \omega' + \omega' + \omega'} n \cdot j + j+j \ar[d]^{\id_{n \times j} + \nabla} \\
  {n \times i + i} \ar[r]_-{n \times \iota' + \iota'} & n + 1 & \ar[l]^-{n \times \omega' + \omega'} {n \cdot j + j}
}$}.
\end{equation}
which is the coproduct of~\eqref{eq:casetwo} with the identity on the disconnected cospan with $n$ edges. The bottom row is taken care of by induction.
\medskip

We are left with only two cases to check.
For~\eqref{eq:caseone},  
\[\CspToSyn{0 \ra 0 \la 0}\quad = \quad  \Zeronet, \quad \text{ and}\]
\[\CspToSyn{0 \ra 1 \la 0}\quad = \lower12pt\hbox{$\includegraphics[height=1cm]{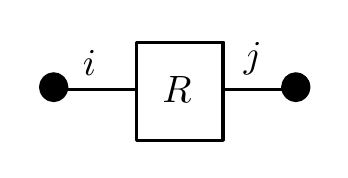}$}.\] 
The following derivation thus suffices. 
\begin{equation*}
  \lower12pt\hbox{$\includegraphics[height=1cm]{graffles/Rbone.pdf}$} \overset{(L_1)}{\leq} \lower5pt\hbox{$\includegraphics[height=0.7cm]{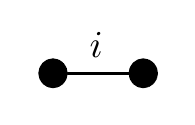}$} \overset{(UC)}{\leq}  \lower11pt\hbox{$\includegraphics[height=1cm]{graffles/empty.pdf}$}.
\end{equation*}

For~\eqref{eq:casetwo}, 
\[
\CspToSyn{i+i \ra 1 \la j+j}= \lower12pt\hbox{$\includegraphics[height=1cm]{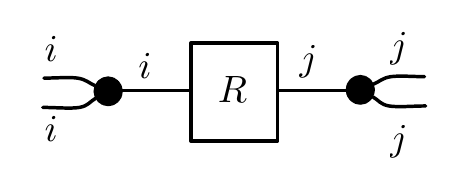}$}, 
\text{ and} 
\]
\[
\CspToSyn{i+i \xrightarrow{\iota' + \iota'} 2 \xleftarrow{\omega' + \omega'} j+j}= \lower18pt\hbox{$\includegraphics[height=1.5cm]{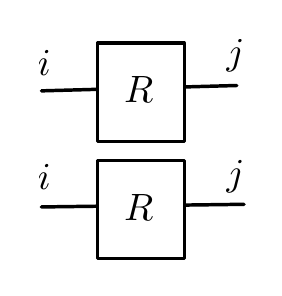}$}.
\]
The following derivation thus completes the proof:
\begin{equation*}
  \lower12pt\hbox{$\includegraphics[height=1cm]{graffles/der1.pdf}$} \overset{(L_2)}{\leq}  \lower18pt\hbox{$\includegraphics[height=1.5cm]{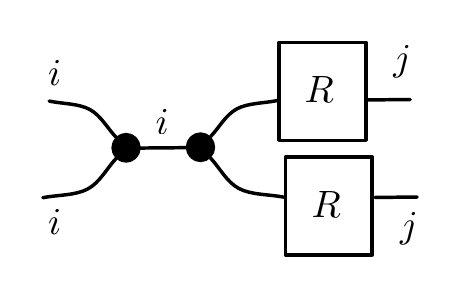}$} \overset{(MC)}{\leq}  \lower18pt\hbox{$\includegraphics[height=1.5cm]{graffles/der3.pdf}$}
\end{equation*}
\end{proof}

We have now all the ingredients to prove Theorem~\ref{thm:hypergraph}.

\begin{proof}[Proof of Theorem~\ref{thm:hypergraph}]
The proof relies on a result appearing in the proof of Theorem 3.3 in~\cite{bonchi2016rewriting}: every discrete cospan of hypergraphs $n\stackrel{\iota}{\to} G \stackrel{\omega}{\leftarrow} m$ can be written as the composition
\begin{equation*}
\xymatrix@R=11pt@C=10pt{
& G_V & & G_V \tns \tilde{E} & & G_V\\
n \ar[ur]^{\iota} & & \ar[ul]_{[id, j]} G_V+\tilde{n} \ar[ur]_{id \tns i} & & \ar[ul]^{id \tns o} G_V + \tilde{m} \ar[ur]^{[id, p]} & & \ar[ul]_\omega m
}
\end{equation*}
where $\tilde{n}\stackrel{i}{\to} \tilde{E} \stackrel{o}{\leftarrow} \tilde{m}$ is disconnected, $G_V$ is the set of vertices of $G$\footnote{Since cospans are taken up-to isomorphism and since $G$ is finite one can always assume, without loss of generality, that $G_V$ is a finite ordinal.}, $j\from \tilde{n} \to G_V$ and $j\from \tilde{m} \to G_V$ maps the vertices of $\tilde{n}\to \tilde{E} \leftarrow \tilde{m}$ into those of $G$. An example is shown below.
\[
\includegraphics[height=37pt]{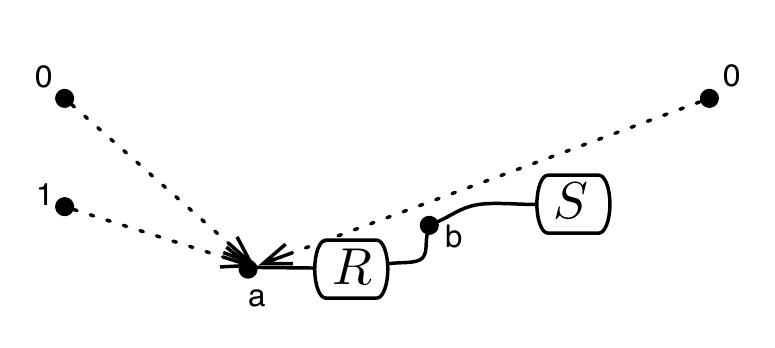} \quad \raisebox{17pt}{=} \quad\includegraphics[height=37pt]{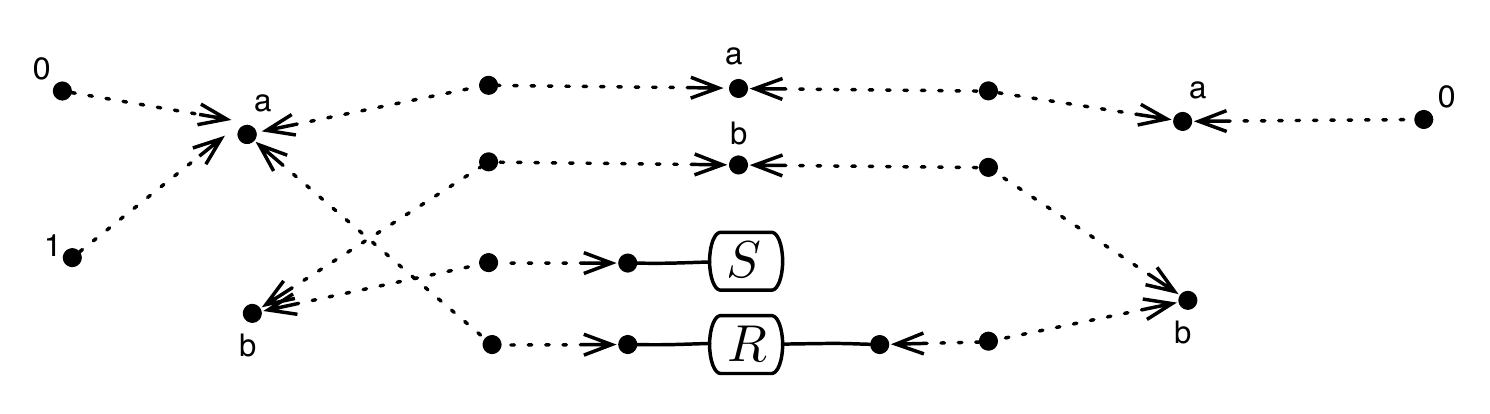}
\]

From the previous discussion, we only need to prove the right-to-left implication of \eqref{eq:ordering}. We will show that if  $  n \to G' \leftarrow m\leq n\to G \leftarrow m$ then $ \CspToSyn{n \to G' \leftarrow m}  \leq   \CspToSyn{n\to G \leftarrow m}  $.

Assume now that $ n \stackrel{\iota'}{\to} G' \stackrel{\omega'}{\leftarrow} m\leq    n\stackrel{\iota}{\to} G \stackrel{\omega}{\leftarrow} m  $, i.e., there exists an $f\from G \to G'$ such that $f \iota = \iota'$ and $f \omega =\omega'$. The morphism $f$ induces $f_V \from G_V \to G'_V$, $f_E \from \tilde{E} \to \tilde{E}'$, $f_{\tilde{n}} \from \tilde{n} \to \tilde{n'}$ and $f_{\tilde{m}} \from \tilde{m} \to \tilde{m'}$ making the following commute.
\begin{equation*}
\xymatrix@R=11pt@C=10pt{
& G_V \ar[dd]|{f_V} & & G_V \tns \tilde{E} \ar[dd]|{f_V \tns f_E} & & G_V \ar[dd]|{f_V}\\
n \ar[dr]_{\iota'} \ar[ur]^{\iota} & & \ar[ul]_{[id, j]} G_V+\tilde{n} \ar[dd]|{f_V\tns f_{\tilde{n}}} \ar[ur]^{id \tns i} & & \ar[ul]_{id \tns o} G_V + \tilde{m} \ar[dd]|{f_V\tns f_{\tilde{m}}} \ar[ur]^{[id, p]} & & \ar[ul]_\omega m \ar[dl]^{\omega'} \\
& G'_V & & G'_V \tns \tilde{E'} & & G'_V\\
  & & \ar[ul]^{[id, j']} G'_V +\tilde{n'} \ar[ur]_{id \tns i'} & & \ar[ul]^{id \tns o'} G'_V + \tilde{m'} \ar[ur]_{[id, p']} & & \\
}
\end{equation*}
From the commutativity of the above diagram, one has:

$$\begin{array}{lrclr}
 (\gamma_1\df) & n\to G'_V   \leftarrow G_V+\tilde{n}  & \leq & n \to G_V \leftarrow G_V+\tilde{n} & (\dfop \delta_1) \\ 
 (\gamma_2\df) & G_V+\tilde{m}  \to     G'_V \leftarrow m  & \leq & G_V+\tilde{m} \to G_V \leftarrow m & (\dfop \delta_2)  \\
 (\gamma_3\df) & G_V \to G'_V \leftarrow G_V & \leq & G_V \to G_V \leftarrow G_V & (\dfop \delta_3) \\
 (\gamma_4\df) & \tilde{n}   \to   \tilde{E'}  \leftarrow \tilde{m}  & \leq & \tilde{n} \to \tilde{E} \leftarrow \tilde{m}  & (\dfop \delta_4)
 \end{array}$$

Since the first three inequations only involve sets and functions, one can use the conclusion of Theorem~\ref{thm:emptysignature} and deduce that: 
\begin{eqnarray*}
\CspToSyn{\gamma_1} &\leq& \CspToSyn{\delta_1} \\ 
\CspToSyn{\gamma_2}
& \leq &\CspToSyn{\delta_2 } \\
 \CspToSyn{\gamma_3}& \leq& \CspToSyn{\delta_3 }
 \end{eqnarray*}
From the fourth inequation, via Lemma~\ref{lemma:disconnected}, one obtains
\begin{eqnarray*}
\CspToSyn{\gamma_4}
& \leq & \CspToSyn{\delta_4 }
\end{eqnarray*}
and conclude as follows.
\begin{align*}
\CspToSyn{n \to G' \leftarrow m} 
&= \CspToSyn{\gamma_1 \poi (\gamma_3 \tns \gamma_4 ) \poi \gamma_2} \\
&= \CspToSyn{\gamma_1} \poi (\CspToSyn{\gamma_3} \tns \CspToSyn{\gamma_4} ) \poi \CspToSyn{\gamma_2} \\
&\leq \CspToSyn{\delta_1} \poi (\CspToSyn{\delta_3} \tns \CspToSyn{\delta_4} ) \poi \CspToSyn{\delta_2}  \\
&= \CspToSyn{\delta_1 \poi (\delta_3 \tns \delta_4 ) \poi \delta_2} \\
&= \CspToSyn{n \to G \leftarrow m}
 \end{align*}
\end{proof}

\subsection{Proof of Section~\ref{sec:CompSpan}}
\label{app:CompSpan}

\begin{proof}[Proof of Lemma~\ref{lem:SpanleqCartBicat}]
Immediate from Proposition~\ref{prop:Cospanislocally} by duality. 
\end{proof}

\begin{proof}[Proof of Proposition~\ref{pro:bijcorr}]
  As stated in the main text, $\mathcal{M}$ is uniquely determined by the set $\mathcal{M}(1)$ and, for each $R\in \Sigma_{n,m}$, a span $\mathcal{M}(R)\colon \mathcal{M}(1)^n \to \mathcal{M}(1)^m$. This data is that of a (possibly infinite) hypergraph (Definition~\ref{def:hyp}). 
\end{proof}

\begin{proof}[Proof of Proposition~\ref{prop:universal}]
By definition, $\mathcal{U}_{G}(1) = G_V$ and $\mathcal{U}_{G}(\SynToCsp{R})=(G_V)^n \stackrel{s_R}{\leftarrow} G_R \stackrel{t_R}{\to} (G_V)^m$ 
for each $R\in \Sigma_{n,m}$. During the proof, we also use implicitly the fact that $(G_V)^n$ is $\iHyp_{\Sigma}[n,G]$.

The conclusion of Theorem~\ref{thm:hypergraph} allows us to argue by induction on $n\stackrel{\iota}{\to} G' \stackrel{\omega}{\leftarrow} m$.
The base cases are $\SynToCsp{\Bcomult}$, $\SynToCsp{\Bmult}$, $\SynToCsp{\Bcounit}$, $\SynToCsp{\Bunit}$, $\SynToCsp{\idone}$, $\SynToCsp{\dsymNet}$ and $\SynToCsp{R}$. Let us consider the last of these, where  $n\stackrel{\iota}{\to} G' \stackrel{\omega}{\leftarrow} m$ is
$$
\SynToCsp{R} =\cgr[height=1.5cm]{graffles/basiccospanFZ.pdf}\text{.}$$

Any homomorphism $f\from G' \to G$ maps its single hyperedge to an $R$-hyperedge of $G$, call it $e_f$, the $n$ vertices in the image of $\iota$ to the source of $e_f$ ($\iota \poi f = s_R(e_f)$) and the $m$ vertices in the image of $\omega$ to the target of $e_f$ ($\omega \poi f = t_R(e_f)$). 
This means that the following diagram commutes.
\begin{equation*}
\xymatrix@R=5pt@C=30pt{
      & G_R  \ar@/^/[dr]^{t_R} \ar@/_/[dl]_{s_R}   & \\
      {(G_V)}^n & & {(G_V)}^m  \\
      &  \iHyp_{\Sigma}[G',G] \ar@/_/[ur]_(0.6){\omega \poi - } \ar@/^/[ul]^(0.6){\iota \poi - }  \ar[uu]^{e_{-}} & \\
}
\end{equation*}
The function $e_{-} \colon  \iHyp_{\Sigma}[G',G] \to G_R$ is clearly an isomorphism of spans.
The other base cases are simpler or, as stated in the main text, follow from the fact that $\iHyp_{\Sigma}[\_, G]$ maps colimits to limits, which also
immediately implies the inductive case. Nevertheless, we spell out the full details of one inductive case here. Let
\[
n\xrightarrow{\iota} G' \xleftarrow{\omega} m = (n\xrightarrow{\iota_1} G_1 \xleftarrow{\omega_1} k) \poi (k\xrightarrow{\iota_2} G_2 \xleftarrow{\omega_1} m).
\]
By definition, there are morphisms $a$ and $b$ such that the following commutes and the central region is a pushout.
$$\xymatrix@R=5pt@C=25pt{
      n \ar@/_/[drr]|{\iota} \ar[r]^{\iota_1} &  G_1  \ar[rd]|a& k  \ar[l]_{\omega_1} \ar[r]^{\iota_2} & G_2\ar[ld]|b  & m \ar[l]_{\omega_2} \ar@/^/[dll]|{\omega} \\
 & & G' 
}$$
By the induction hypothesis:
\begin{align*}
\mathcal{U}_{G}(n\xrightarrow{\iota_1} G_1 \xleftarrow{\omega_1} k) &= {(G_V)}^n \xleftarrow{\iota_1 \poi -} \iHyp_{\Sigma}[G_1,G] \xrightarrow{\omega_1 \poi -}  {(G_V)}^k\text{, and}\\
\mathcal{U}_{G}(k\xrightarrow{\iota_2} G_2 \xleftarrow{\omega_2} m) &=  {(G_V)}^k \xleftarrow{\iota_2 \poi -} \iHyp_{\Sigma}[G_2,G] \xrightarrow{\omega_2 \poi -} {(G_V)}^m\text{.}
\end{align*}
As composition in $\Span{\Set}$ is given by pullback, by functoriality $\mathcal{U}_{G}(n\xrightarrow{\iota} G' \xleftarrow{\omega} m)$ is the span ${(G_V)}^n \leftarrow A \to {(G_V)}^m$ below
$$\xymatrix@R=5pt@C=10pt{
      & & A \ar@/^/[dr]^(0.4){\pi_2} \ar@/_/[dl]_(0.4){\pi_1}   & \\
      {(G_V)}^n &  \iHyp_{\Sigma}[G_1,G] \ar[l]^(0.6){\iota_1 \poi -} \ar[r]_(0.6){\omega_2 \poi -}  & {(G_V)}^k  &  \iHyp_{\Sigma}[G_2,G] \ar[l]^(0.6){\iota_2 \poi -} \ar[r]_(0.6){\omega_2 \poi -} & {(G_V)}^m
}$$
where $$A=\{(f,g)\in \iHyp_{\Sigma}[G_1,G] \times \iHyp_{\Sigma}[G_2,G] \mid \omega_1\poi f= \iota_2 \poi g\}$$
that is, the set of all pairs of morphisms $f\colon G_1 \to G$ and $g\colon G_2\to G$ that agrees on the common boundary $k$.
Our aim is to prove that the above span is isomorphic to ${(G_V)}^n \xleftarrow{\iota \poi -} \iHyp_{\Sigma}[G',G] \xrightarrow{\omega \poi -} {(G_V)}^m$: intuitively, this means that pairs of morphisms $(f,g)$ in $A$ are in one-to-one correspondence with morphisms in $\iHyp_{\Sigma}[G',G]$.

We start by defining $\bar{\cdot}\colon A \to \iHyp_{\Sigma}[G',G]$. For $(f,g)\in A$, i.e.,  $\omega_1 \poi f= \iota_2 \poi g$, the morphism $\overline{(f,g)}\colon G'\to G$ is given by the universal property of the pushout $G'$.
$$\xymatrix@R=5pt@C=15pt{
      n \ar@(dr,l)[drr]|{\iota} \ar[r]^{\iota_1} &  G_1 \ar@(dl,l)[dddr]|{f} \ar[rd]|a& k  \ar[l]_{\omega_1} \ar[r]^{\iota_2} & G_2 \ar@(dr,r)[dddl]|{g} \ar[ld]|b  & m \ar[l]_{\omega_2} \ar@(dl,r)[dll]|{\omega} \\
 & & G' \ar@{.>}[dd]|(0.4){\overline{(f,g)}} \\\\
 & & G'
}$$
Observe that $\iota \poi \overline{(f,g)}=\iota_1 \poi f$ and that $\omega \poi \overline{(f,g)}=\omega_2 \poi g$. This means that the following diagram commutes.
$$\xymatrix@R=5pt@C=30pt{
      & A  \ar[dd]^{\bar{\cdot}}  \ar@/^/[dr]^{\pi_1\poi (\omega_2 \poi - )} \ar@/_/[dl]_{\pi_1\poi (\iota_1 \poi - )}   & \\
      {(G_V)}^n & & {(G_V)}^m  \\
      &  \iHyp_{\Sigma}[G',G] \ar@/_/[ur]_(0.6){\omega \poi - } \ar@/^/[ul]^(0.6){\iota \poi - } & \\
}
$$
It is easy to see that $\bar{\cdot}\colon A \to \iHyp_{\Sigma}[G',G]$ is an isomorphism: its inverse maps each $h\colon G'\to G$ to $(a;h,b;h)\in A$. 

This concludes the proof for the case of composition. The case of tensor is similar but uses the universal property of coproducts rather than of pushouts.
\end{proof}

\subsection{Proofs of Section~\ref{sec:rel}}\label{app:proofsCompSpan}

In this appendix, we first show that $\PosetCat{\mathcal{(\cdot)}}$ from Definition \ref{def:reduction} is a functor and then prove Lemmas and Propositions of Section~\ref{sec:rel}.

Observe that we have a canonical functor that mediates between a preorder-enriched category $\mathcal{C}$ and its poset-reduction $\PosetCat{\mathcal{C}}$.
This functor $ \mathcal{A}_{\mathcal{C}} \from \mathcal{C} \to \PosetCat{\mathcal{C}}$ is identity on objects and sends a morphism in $\mathcal{C}$ to
its $\sim$-equivalence class in $\PosetCat{\mathcal{C}}$. We will omit the subscript on $\mathcal{A}$ whenever possible.

An immediate consequence of the definition is that $\mathcal{A}$ preserves and reflects the ordering in the following sense:
\begin{lemma}\label{lem:refl}
For $\mathcal{C}$ a preorder-enriched category, and $x,y$ morphisms in $\mathcal{C}$, we have
$\mathcal{A}(x)\leq \mathcal{A}(y)$ if and only if
$x \leq y$.\qed
\end{lemma}
\begin{proof}
  This follows from the observation that $x \leq y$ if and only if $a \leq b$ for some $a \geq x$ and $b \leq y$.
\end{proof}

The functors $\mathcal{A}$ exhibit the following universal property:
\begin{lemma}
  For every preordered functor $F \from \mathcal{C} \to \mathcal{D}$ between the preordered category $\mathcal{C}$ and poset-enriched category $\mathcal{D}$,
  there is a unique poset-enriched functor $G \from \PosetCat{\mathcal{C}} \to \mathcal{D}$ such that the following diagram commutes:
  \[
  \begin{tikzcd}
    \mathcal{C} \arrow{r}{F} \arrow[swap]{d}{\mathcal{A}_{\mathcal{C}}} & \mathcal{D} \\
    \PosetCat{\mathcal{C}} \arrow[swap]{ur}{G}
  \end{tikzcd}
  \]
  In particular, for every preordered functor $H \from \mathcal{C} \to \mathcal{C}'$ for a preorder-enriched category $\mathcal{C}'$, there is a unique functor
  $\PosetCat{H} \from \PosetCat{\mathcal{C}} \to \PosetCat{\mathcal{C}'}$ such that the following commutes:
  \[
  \begin{tikzcd}
    \mathcal{C} \arrow{r}{H} \arrow[swap]{d}{\mathcal{A}_{\mathcal{C}}} & \mathcal{C'} \arrow{d}{\mathcal{A}_{\mathcal{C'}}} \\
    \PosetCat{\mathcal{C}} \arrow{r}{\PosetCat{H}} & \PosetCat{\mathcal{C'}} \\
  \end{tikzcd}
  \]
  \label{lem:PosetCatUnivProp}
\end{lemma}

\begin{proof}
  For a morphism $f \in \mathcal{C}$ let $[f]$ denote the equivalence class of $f$ modulo $\sim$.
  Then setting $G([f]) = F(f)$ is well-defined, since $\mathcal{D}$ is a poset-enriched category. $G$ defines a functor since $\sim$ is a congruence, hence
  compatible with composition. Since $A_{C}$ is surjective on objects and morphisms, there can be at most one such functor $G$, hence $G$ is unique.
\end{proof}

In other words, we get a function, $\PosetCat{\left(\cdot\right)}$, that turns functors between preorder-enriched categories into functors between the
associated poset-enriched ones.
To prove that this assignment is functorial, namely that it preserves identities and composition is trivial by the above definition.

\begin{proof}[Proof of Proposition~\ref{pro:spantilde}]
  We recall a well-known construction of the ordinary category of relations: a span $X\xleftarrow{f} A \xrightarrow{g} Y$ induces a relation $R_A \subseteq X\times Y$ by factorising $A\xrightarrow{[f,g]}X\times Y$ as a surjection followed by an injection; the injection, when composed with the projections, yields a \emph{jointly-injective} span. These, up-to span isomorphism, are the same thing as subsets $R_A\subseteq X\times Y$. This procedure respects composition and monoidal product, yielding a functorial mapping $\Spanleq{\Set} \to \Rel$ on objects and arrows.
Given the above, it suffices to show that there exists a span homomorphism 
$(X\leftarrow A\rightarrow Y)\to(X\leftarrow B\rightarrow Y)$ iff $R_A\subseteq R_B$ as relations. The `only if' direction is implied by the nature of factorisations of functions: given a homomorphism of spans, we obtain an induced function $R_A \to R_B$, illustrated as the dotted function below, which is an injection since it is the first part of a factorisation of an injection.
\[
\xymatrix@R=10pt{
{A} \ar[r] \ar@{->>}[d] & {B} \ar@{->>}[r] & {R_B} \ar@{_(->}[d] \\
{R_A} \ar@{_(->}[rr] \ar@{.>}[urr]
& & {X\times Y}
}
\]
For the `if' part, since (by the axiom of choice) surjective functions split, we obtain $R_B\to B$. Then $\xymatrix@C=10pt{A\ar@{->>}[r]&{R_A}\ar[r] & {R_B} \ar[r] & B}$ is easily shown to be a homomorphism of spans.
\end{proof}

\begin{proof}[Proof of Proposition~\ref{pro:preordPosetBicat}]
  We stated the axioms of preordered cartesian bicategories and cartesian bicategories in a way that makes the first part obvious.
  Given a morphism $F \from \mathcal{B}_1 \to \mathcal{B}_2$ of preorder-enriched cartesian bicategories, clearly $\PosetCat{F}$ is still an order-preserving monoidal functor.
It also preserves the monoid and comonoid structures, because the monoid/comonoid on an object $X \in \PosetCat{{\mathcal{B}_i}}$ is the equivalence
  class of the monoid/comonoid structure in $\mathcal{B}_i$.
\end{proof}

\begin{proof}[Proof of Lemma~\ref{lem:transfer}] The second item is trivial. For the first one, let $x,y$ be morphisms in $\PosetCat{\mathcal{C}}$ such that $G(x) \leq G(y)$ for all $\mathcal{G} \in \PosetCat{\mathcal{F}}$. We want to prove $x \leq y$.
  Now let $F \in \mathcal{F}$ be arbitrary. Then $\PosetCat{F}(x) \leq \PosetCat{F}(y)$ by assumption on $x,y$. Since morphisms in $\PosetCat{\mathcal{C}}$ are just
  equivalence classes of morphisms in $\mathcal{C}$, choose representatives, i.e. morphisms $f, g$ in $\mathcal{C}$ such that
  $\mathcal{A}(f) = x$ and $\mathcal{A}(g) = x$.
  Since the diagram
  \[
    \begin{tikzcd}
      \mathcal{C} \arrow{r}{F} \arrow{d}{\mathcal{A}} & \mathcal{D} \arrow{d}{\mathcal{A}} \\
      \PosetCat{\mathcal{C}} \arrow{r}{\PosetCat{F}} & \PosetCat{\mathcal{D}}
    \end{tikzcd}
  \]
  commutes, we get
  \[
    \mathcal{A}(F(f)) = \PosetCat{F}(\mathcal{A}(f)) = \PosetCat{F}(x) \leq \PosetCat{F}(y) = \PosetCat{F}(\mathcal{A}(g)) = \mathcal{A}(F(g)).
  \]
  Since $\mathcal{A}$ reflects the ordering (Lemma~\ref{lem:refl}), we get $F(f) \leq F(g)$. But $F \in \mathcal{F}$ was arbitrary, therefore $f \leq g$,
  since $\mathcal{C}$ is $\mathcal{F}$-complete for $\mathcal{D}$. But therefore
  \[
    x = \mathcal{A}(f) \leq \mathcal{A}(g) = y
  \]
  which finishes the proof.
\end{proof}

\end{document}